%% file: ms.tex
\pdfoutput=1
\documentclass{scrartcl}
\usepackage[utf8x]{inputenc}
\usepackage[OT4]{fontenc}
\usepackage{amsmath,amsthm,latexsym,amssymb,amsfonts,color,graphicx}
\usepackage{bbm}
\usepackage{tikz}

\usepackage{authblk}
\usepackage{color}

\usepackage{hyperref}

\renewcommand{\it}{\normalfont \itshape}

\newcommand{\braket}[2]{\left< #1 | #2 \right>}
\newcommand{\CP}{\mathbb{CP}}

\newcommand{\jj}{\jmath}
\renewcommand{\j}{\jmath}
\newcommand{\geom}{\Delta}

\newcommand{\G}{{\{G\}}}
\newcommand{\tG}{{\{\tilde{G}\}}}
\newcommand{\Gp}{{\{G^+\}}}
\newcommand{\Gm}{{\{G^-\}}}

\newcommand{\cdop}{\cdot'}
\newcommand{\Hodgp}{\ast'}
\newcommand{\lrcornep}{\lrcorner'}
\newcommand{\llcornep}{\llcorner'}

\newcommand{\cdou}{\underline{\cdot}}
\newcommand{\Hodgu}{\underline{\ast}}
\newcommand{\lrcorneu}{\underline{\lrcorner}}
\newcommand{\llcorneu}{\underline{\llcorner}}
\newcommand{\perpu}{\underline{\perp}}
\newcommand{\T}{\underline{t}}

\newcommand{\Heprl}{{\mathcal H}_{EPRL(\jj,\rho)}}

\DeclareMathOperator{\iu}{i}
\DeclareMathOperator{\Inv}{Inv}

\newcommand{\R}{{\mathbb R}}
\newcommand{\C}{{\mathbb C}}
\newcommand{\Z}{{\mathbb Z}}
\newcommand{\B}{{\mathbb B}}
\newcommand{\M}{{\mathbb M}}

\DeclareMathOperator{\SSL}{SL}
\DeclareMathOperator{\SO}{SO}
\DeclareMathOperator{\OO}{O}
\DeclareMathOperator{\GL}{GL}
\DeclareMathOperator{\St}{St}
\newcommand{\SL}{\SSL(2,\C)}
\DeclareMathOperator{\SU}{SU}

\newcommand{\widebar}{\overline}
\newcommand{\PHi}{\varPhi}

\newcommand{\Hodge}{\ast}

\newcommand{\z}{\mathbf{z}}
\newcommand{\n}{\mathbf{n}}

\renewcommand{\u}{\mathbf{u}}
\renewcommand{\v}{\mathbf{v}}
\renewcommand{\r}{\mathbf{r}}

\newcommand{\bvec}[2]{\left(
   \begin{array}{c} {#1} \\ {#2}
    \end{array}\right)}%
\newcommand{\bmat}[1]{\left(
   \begin{array}{cc} #1
    \end{array}\right)}%

\newtheorem{thm}{Theorem}

\newtheorem{lm}{Lemma}
\newtheorem{df}{Definition}

\DeclareMathOperator{\sgn}{sgn}

\DeclareMathOperator{\tr}{tr}

\DeclareMathOperator{\Arg}{Arg}

\newcommand{\vc}[1]{\left(
   \begin{array}{c} #1 
    \end{array}\right)}%

\begin{document}

\title{Asymptotic analysis of the EPRL model with timelike tetrahedra}

\author[1]{Wojciech Kami\'nski\thanks{wkaminsk@fuw.edu.pl}}
\author[1,2]{Marcin Kisielowski\thanks{marcin.kisielowski@gmail.com}}
\author[2]{Hanno Sahlmann\thanks{hanno.sahlmann@gravity.fau.de}}
\affil[1]{Instytut Fizyki Teoretycznej, Wydzia{\l} Fizyki, Uniwersytet Warszawski, ul. Pasteura 5 PL-02093 Warszawa, Poland}
\affil[2]{Institute for Quantum Gravity, Department of Physics, Friedrich-Alexander Universit\"at Erlangen-N\"urnberg (FAU), 
Staudtstr. 7 D-91058 Erlangen, Germany}

\date{\today}

\maketitle

\abstract{We study the asymptotic behaviour of the vertex amplitude for the EPRL spin foam model extended to include timelike tetrahedra. We analyze both, tetrahedra of signature $---$ (standard EPRL), as well as of signature $+--$ (Hnybida-Conrady extension), in a unified fashion. However, we assume all faces to be of signature $--$. We find that the critical points of the extended model are described again by $4$-simplices and the phase of the amplitude is equal to the Regge action. Interestingly, in addition to the Lorentzian and Euclidean sectors there appear also split signature $4$-simplices.}

\input{Introduction}
\tableofcontents

\input{chapter-Vertex_amplitude}
\input{chapter-Stationary_phase}
\input{chapter-Reality-conditions}
\input{chapter-Bivectors-spinors}
\input{chapter-Geometric-reconstruction}
\input{chapter-Other-signatures}

\input{chapter-Classification}
\input{chapter-Difference-of-phases}
\input{chapter-Hessian}
\input{chapter-Conclusions}
\input{Appendix}

\bibliography{bibliography}{}
\bibliographystyle{hieeetr-mod}

\end{document}

%% file: Introduction.tex
%!TEX root = Spin_foams_with_time_like_tetrahedra.tex

%\tableofcontents

\section{Introduction}

%\subsection{Motivations}

Spin foam models have many origins. On the one hand they appear in the calculation of  transition amplitudes in the (topological) quantum theory of 3d gravity (the Turaev Viro model \cite{turaev-viro} and, less formal,  the Ponzano-Regge model \cite{PR,pr-model}). 
On the other hand spin foams can be regarded as Feynman diagrams of some auxiliary theory \cite{ReisenbergerPath,Reisenberger2000,Reisenberger2001,Pietri2000,GFTII,GFTIII,GFTIV,GFTV}, which can be chosen such that these diagrams correspond to triangulations of (pseudo-) manifolds. The spin foam sum then becomes a sum over geometries, and the underlying theory, if defined independently, gives a well defined resummation. Prime examples for this viewpoint are matrix models. Both interpretations taken together suggest that spin foam models are relevant also for quantum gravity also in dimension four, although it is not a topological theory anymore.

Most of the modern 4d spin foam models are implementations of the path integral of (topological) BF theory, combined with the imposition of additional constraints on the $B$ field \cite{Barrett2000, PietriClass, Alexandrov2008, Livine2008b, EPRL, FK, Baratin, Perez2012}. The resulting amplitudes can be interpreted as the transition amplitude between spin network states of loop quantum gravity (LQG) \cite{Kaminski2009,Alexandrov2011}. For a careful analysis of the prospects and problems with linking spin foams and LQG see \cite{Alesci2011,Thiemann2013}. 

States in LQG represent purely spatial geometry. Thus for the LQG interpretation of the spin foam amplitudes, at least the triangulations of the boundary must consist of spatial tetrahedra. In fact, in the standard EPRL-FK model \cite{EPRL,FK,Kaminski2009}, all tetrahedra are spacelike. On the other hand, the BF theory with group $\SU(1,1)\approx \SSL(2,\R)$  (the Ponzano-Regge model in 3d in Lorentzian signature) contains sums over all (continuous and discrete series) representations. According to the coadjoint orbit method \cite{Kirillov, WittenCoadj} it should correspond to appearance of both timelike and spacelike edges and thus also mixed signature triangles (this is 3d analog of timelike tetrahedra). In fact it has been confirmed for specific cases \cite{Davids, DavidsDoc}. 
Moreover, the absence of timelike tetrahedra in LQG is in itself puzzling and attracted attention \cite{Perez2001}.\footnote{In fact, there is a way to reconcile LQG with BF theory by projective spin network technique \cite{Livine2002} as well as canonical analysis that allow timelike signature \cite{Alexandrov2005, Liu2017}.}
Thus it is important to note that there is an extension of the standard EPRL model which incorporates tetrahedra with spacelike and with timelike normals \cite{Conrady2010}. In the present paper we will work with this \emph{extended} EPRL model. 

The are now many spin foam models, but as there is no unique way of constructing them, it is not clear if any of them is a good candidate for quantum gravity. One important test is an asymptotic analysis of the models.
The methodology for such analysis is now quite established (see for example \cite{Frank-thesis}) based on many developments \cite{Barrett2003,Frank2,FrankEPRL,freidel-louapre,LivineSpeziale,4dWilliams} etc. The method of choice for 4d models, pioneered by \cite{FrankEPRL}, is an extended stationary phase approximation with Livine-Speziale coherent states as boundary states \cite{Livine2007}.  In this approach the amplitude is written as an integral with oscillatory integrant $e^S$ (see \eqref{eq:vertexamplitude}) with $S$ called an action. The key role of is played by critical points of $S$, i.e. stationary points for which the real part of the action vanishes. Each such point contributes to the asymptotic formula with an oscillatory term and if the action does not have any critical points, the amplitude is exponentially suppressed. The asymptotic behaviour of the vertex amplitudes is expected to be governed by the Regge action of a suitably constructed $4$-simplex $\Delta$ built out of the boundary data.\footnote{We note that the analysis of the asymptotic behaviour of the vertex amplitude is only the first step in the process, because the proper behaviour of the vertex amplitude does not necessarily imply ``right'' asymptotic behaviour of the spin foam sum \cite{Hellmann2013} (see however \cite{Han2013}).}
Indeed, in the standard EPRL model, the asymptotic limit of the evaluation (see \cite{Frank2, FrankEPRL} and for an overview \cite{Carlo, Perez2012}) is governed by the Regge action for discrete gravity without cosmological constant \cite{Regge}
\begin{equation}
\label{eq:regge}
 S_\Delta=\sum_{i<j} A_{ij}\theta_{ij}
\end{equation}
where $A_{ij}$ are areas of the faces and $\theta_{ij}$ are the corresponding dihedral angles (see section \ref{sub:summary} for a precise statement). 
This is similar to the Ponzano-Regge asymptotic formula \cite{PR, Roberts} valid in three dimensional gravity. Our main result is a classification of the critical points in the generalized EPRL model. Moreover, we compute the difference of phases of two critical points and show how it relates to the Regge action (with the definition of dihedral angles given by \cite{Suarez}, see section \ref{sec:Regge} for details).

The model with timelike tetrahedra is considerably more difficult than standard EPRL thus it was not considered seriously so far. However, timelike tetrahedra are natural candidates for some boundary conditions, for example in spin foam cosmology \cite{Bianchi2010a,Rennert2016}, or in spin foam black hole calculations (if such will be done in the future). We hope that providing the asymptotic analysis in the present work will support and boost research in this direction.

Adding the possibility for timelike tetrahedra dramatically expands the complexity of the asymptotic analysis: Now a given $4$-simplex can contain tetrahedra of different types, and hence many different cases arise. In order to manage this complexity in a reasonable way we have  developed a unified treatment of all the cases independent of the type of the tetrahedra. This unified treatment, which brings out the essence of the argument while streamlining the technical issues, is the second result of our work. We expect it to be useful for future development using the vertex amplitudes. 

One source of the complications for timelike tetrahedra is the change in the type of subgroup, 
from $\SU(2)$ in the standard case to $\SU(1,1)$, in the definition of the amplitudes. Our treatment is building on the work \cite{Conrady2010a} for the $\SU(1,1)$ case. 

Another source of complications in the generalized model is that, since the geometry is much more involved, it is not clear that there are always two (in the degenerate case one) critical points. 
Our work clarifies the situation in the following way: In the standard EPRL model two critical points are related by time reversal, since the boundary data for spacelike tetrahedra is invariant under this transformation. This is no longer true for mixed boundary data. We show how to construct the second critical point in the general case and provide a complete classification. For an overview see figure \ref{fig:class}. We also clarify the meaning of the orientation matching condition, which is more involved in the timelike set-up.

As described above, our analysis is quite general. Still, we we impose two restrictions on the geometry, and make two important assumptions. One restriction is that we assume all the faces to be spacelike (signature $--$). We expect (see \ref{sec:types-faces}) that in the case of timelike faces some nongeometric contributions will occur in the asymptotic behaviour. The other restriction is the exclusion of null tetrahedra and null faces. 
It is not known how to construct the vertex amplitude in this case. However, as we explain in section \ref{sec:conjectures}, we expect that one can still see contributions from those configurations in the EPRL model and it seems that they are responsible for problems in application of the stationary phase method.

Our work is a first step towards obtaining asymptotic behaviour of the amplitude. In section \ref{sec:main_results} we will briefly review the setting and state our main result, Theorem \ref{thm:1}.  We describe remaining issues in the section \ref{sec:conjectures}. Since our result is not easy to state concisely, and since the proof is quite long, we will give a more thorough introduction in sections \ref{sec:main_results} and \ref{sec:types-faces} and the outline of the proof in section \ref{sec:Outline}.

\input{main_results.tex}

\subsection{Type of faces}\label{sec:types-faces}
As we mentioned earlier, our analysis is restricted to the case when all faces are of signature $(--)$ that is when the diagonal simplicity conditions of EPRL are \cite{Conrady2010}:
\begin{equation}
 \rho=2\gamma\jj.
\end{equation}
In the case of signature $(-+)$ the simplicity constraint is:
\begin{equation}
 2\jj=\gamma \rho.
\end{equation}
This case is more complicated. Firstly, no EPRL like construction is known (see \cite{Rennert2016} for recent developments). Secondly, the coherent state proposition \cite{Conrady2010} in the line of Freidel-Krasnov model differs from representation theoretic construction in the style of EPRL. We expect that the asymptotic analysis in this case may lead to some unphysical sectors due to non-extremality of the necessary embeddings\footnote{This problem is probably absent in Barrett-Crane model \cite{Barrett2000}.}. 

We will give heuristic explanation of this fact. With the use of normal $N$ we can identify bivector $B$ in Minkowski spacetime $M^4$ with two vectors $\vec{L},\vec{K}$ perpendicular to $N$ by formula
\begin{equation}
 B=\Hodge (N\wedge \vec{L})+N\wedge \vec{K}.
\end{equation}
Regarding bivectors as generators of $\SO(1,3)$ group we can write classical version of EPRL constraints imposed on the embedding. After some rearrangements we obtain:
\begin{equation}\label{con:EPRL}
 \vec{L}\cdot (\vec{L}+\gamma\vec{K})=0,\quad (\vec{L}+\gamma\vec{K})^2=0.
\end{equation}
These are quadratic constraints whereas simplicity would be:
\begin{equation}\label{con:simplicity}
 \vec{L}+\gamma\vec{K}=0.
\end{equation}
Let us now describe case by case if from \eqref{con:EPRL} follows \eqref{con:simplicity}.
\begin{itemize}
 \item In the standard EPRL case: $N$ is timelike so $\vec{L}$ and $\vec{K}$ are spacelike so from \eqref{con:EPRL} simplicity follows trivially as the only spacelike vector with norm zero is the zero vector.
 \item In the extended EPRL case with spacelike face: $N$ is spacelike but $\vec{L}$ is timelike. From \eqref{con:EPRL} we know that $\vec{L}+\gamma\vec{K}$ is null. But there is no nontrivial null vector that is perpendicular to a nontrivial timelike vector, so $\vec{L}+\gamma\vec{K}=0$.
 \item In the extended EPRL case with timelike face: $N$ is spacelike and $\vec{L}$ is spacelike. In this situation it might happen that $\vec{L}+\gamma\vec{K}\not=0$ (simplicity fails).
\end{itemize}
As the simplicity \eqref{con:simplicity} is crucial for the reconstruction of $4$ simplex we suspect that there might be some non-geometric configurations for timelike faces.

\subsection{Asymptotic analysis (conjecture)}\label{sec:conjectures}

The vertex amplitude is given by an oscillatory integral over an infinite domain. We would like to apply version of stationary phase analysis as presented for example in \cite{Hormander}. However in order to give a proper asymptotic result several technical conditions need to be satisfied. 
\begin{enumerate}
 \item {\it finiteness}. It is not known whether the evaluation of a spin network in the extended EPRL model is, in fact, finite. It was proven, but only for the standard EPRL setting (with all tetrahedra spacelike), that the integral defining the amplitude is absolutely convergent \cite{Baez2001,Kaminski2010}. 
 The condition (well-definiteness of the definition) is necessary for any statement about the amplitude to make sense. 
 \item \textit{lack of boundary contributions}. As the integral is over infinite domain with boundaries (or equivalently using compactification over domain with boundary) there might be some additional contributions to asymptotic expansion given by a  version of Watson lemma. This is true as well in standard EPRL case (as the integration is over infinite domain), however in the extended set-up boundary is much more complicated. Preliminary check in the simpler Barrett-Crane model suggests that such contributions in fact appear but only for special boundary configurations that correspond to null (degenerate tetrahedra)\footnote{In Barrett-Crane model only areas are specified but boundary contribution corresponds to a geometry of 4 simplex with null tetrahedra}. In extended EPRL set-up one has additionally finite boundary of integrations. The integral at this boundary decays fast unless the spin of suitable representation is zero (degenerate or null face). This suggests that boundary contributions are related to null structures and we make the following conjecture: For nondegenerate boundary data (that excludes null faces and tetrahedra) there are no contributions from the boundary of integration. This  issue was not addressed in any asymptotic analysis of the Lorentzian spin foam models so far. It is not a problem for the Euclidean models, as there is no boundary of the integration.
 \item \textit{critical point non-degeneracy condition}. An additional assumption is that, after suitably taking into account the symmetry of the action, the remaining matrix of second derivatives (Hessian) at the critical point is non-degenerate. This condition was never addressed in full generality in the analysis of the vertex amplitude. It was only checked for specific configurations in the standard EPRL model \cite{Frank2}. It is known that for the Barrett-Crane \cite{Barrett2000} spin foam models there are special configurations for which Hessian is degenerate \cite{Kaminski2013}. However, this is probably related to the fact that the proper variables used in semiclassical description of this model are areas and conjugated angles \cite{Barrett:1994nn,Barrett1999,Dittrich:2008va} and there is some singularity in going from these variables to shapes. We conjecture that the Hessian is non-degenerate for the EPRL spin foam model for non-degenerate boundary data if the reconstructed $4$ simplex is non-degenerate. The Hessian determinant is important for the asymptotic analysis as it determines asymptotic measure factor in the spin foam path integral. However, it is notoriously difficult to compute in terms of geometric quantities. The only known result for generic case is for the Barrett-Crane model \cite{Kaminski2013}. If non-degeneracy fails, the scaling behaviour is different.
\end{enumerate}
These issues definitely deserve a separate treatment and we leave them for future research. Under these conditions our results would give a proper asymptotic expansion of the extended EPRL vertex amplitude (see \cite{Hormander} and for the form of asymptotic vertex amplitude \cite{Frank2}).

\subsection{Outline of the proof}\label{sec:Outline}

In this paper we decided to give basically complete and detailed proof of the result. There are several reasons for that. First of all many of the subtle details of the classification of solutions are scattered in the literature and it is not easy to determined state of the affair. It is hard to find proof (extending to general case i.e. not using special properties of $\SU(2)$) that there are at most two critical points or that there are no critical points if orientations matching condition is not satisfied. Moreover, the extended EPRL model has more complicated geometry than the standard one. For example relation between two different critical points is more complicated than in standard case (they are not merely related by conjugation). 

Because of that our paper is painfully long so for convenience of the reader we provide here an outline of the proof.
We decided to work directly with the integral of EPRL $Y$-maps contracted using invariant bilinear form $\beta$. It is in contrast to approach of \cite{Frank2} where $\beta$ integrals over $(\CP^1\times \CP^1)^{10}$ are perform analytically and a critical point is described by $\SL$ group elements $g_i$. The reason is that in extended case there are many subcases and the result would lead to higher complexity of the problem. This however means that there are more variables 
and our critical points are characterized by $\SL$ group elements $g_i$ and and in addition spinors $\z_{ij}$.

First task is to write all the cases in the unified manner allowing to treat them at once. This is done in Section \ref{sec:extended} and Section \ref{sec:action}. After achieving this we concentrate on describing critical point conditions. In order to do this we need to derive certain properties of the action Section \ref{sec:traceless} (detailed exposition in appendix \ref{app:traceless}). Critical points are written in terms of group and spinor variables. We can then use traceless matrices instead of spinors (proved in Section \ref{sec:traceless}) obtaining what we call {\it $\SL$ solution}. Up to additional information of spin lifts this are equivalent to {\it geometric $\SO(1,3)$ solution} introduced in Section \ref{sec:stationary}. The latter consists of a bunch of $\SO(1,3)$ group elements that transform given boundary face bivectors such that they fit together. 

As we know from standard EPRL asymptotic analysis two degenerate critical points can correspond to one geometric solution but in different signature. In order to incorporate this we develop theory of geometric solutions in arbitrary signature in Section \ref{sec:geometric}. We show relation of the $4$ simplex to geometric solutions (classifications in terms of reconstructed $4$ simplices) in general. We have now direct correspondence between critical points and $4$ simplices. In Section \ref{sec:other} we show how to construct $4$ simplex of different signature from two degenerate critical points (vector geometries) in $\SO(1,3)$ signature. In order to finish the proof we need classification that there are only two nondegenerate $\SO(1,3)$, two vector geometries or single critical point. This is done in Section \ref{sec:classification} together with the proof that if lengths matching condition is satisfied, but not the orientations matching condition then there is no critical point.

Parallelly to geometric description we compute difference of the phases between two critical points. As in the case of critical point description we lift the computation from original spinor description into the $\SO(1,3)$ geometric description (or other signature computation). The problem with this approach is that in this way we loose some finite ambiguity of multiple of $\pi$ that need to be regained by the topological deformation argument. Most of the phase computation is done in Section \ref{sec:phase}.

%% file: main_results.tex
%!TEX root = Spin_foams_with_time_like_tetrahedra.tex

\subsection{Main results} \label{sec:main_results}
\begin{figure}\label{fig:4_simplex_graph}
\begin{center}
\includegraphics[scale=0.3]{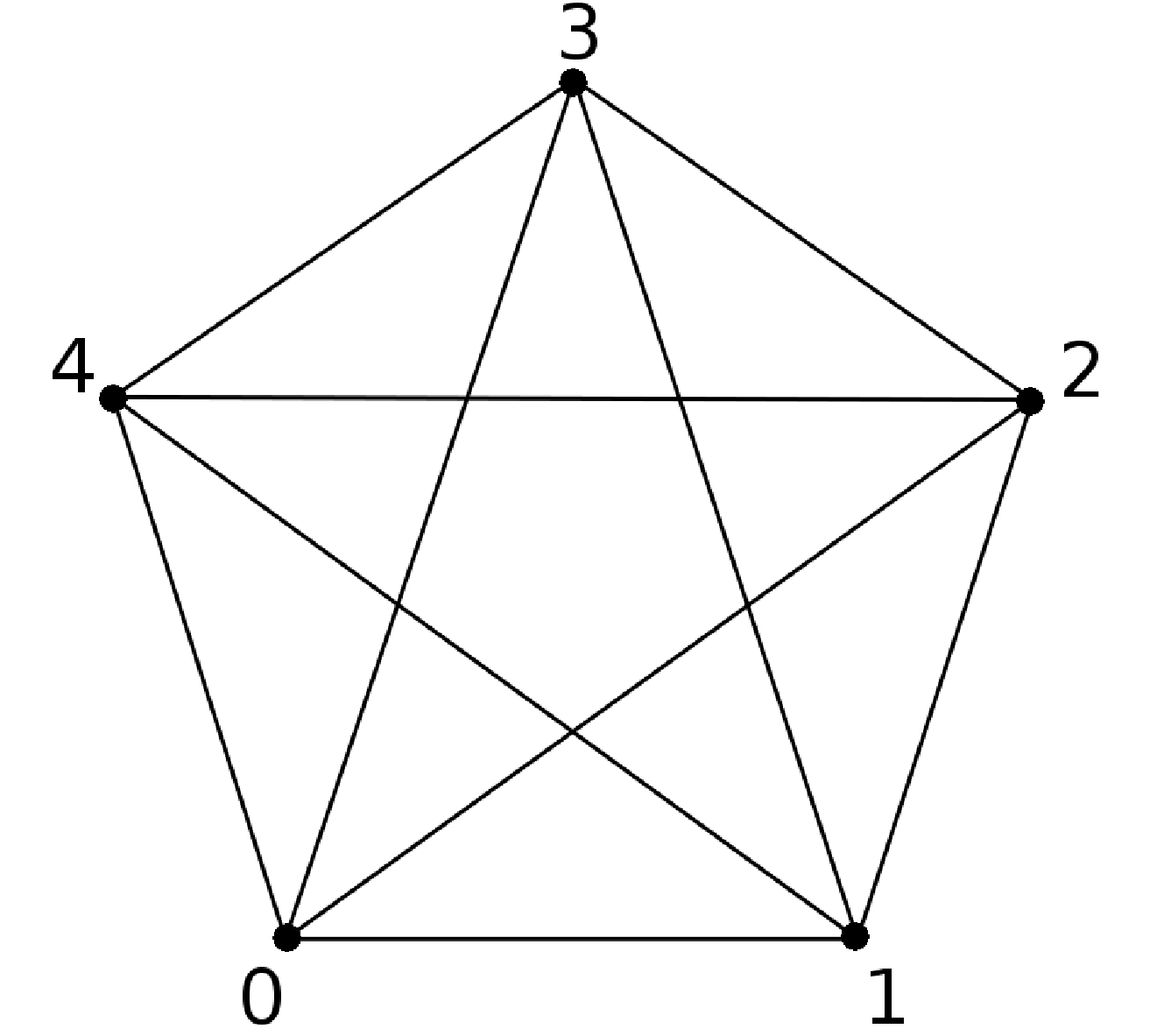}
\caption{Complete graph on five nodes $0,\ldots, 4$.}
\end{center}
\end{figure}
In the standard EPRL-FK model \cite{EPRL,FK,Kaminski2009} the vertex amplitude assigns a complex number to each SU(2) spin network $s=(\Gamma,\jj,\iota)$. The asymptotic limit has been studied in the case when $\Gamma$ is a complete graph on five nodes $0,\ldots, 4$ (the unique graph connecting each pair of nodes with precisely one link) \cite{Frank2}. This means that spin networks with nodes labelled by the Livine-Speziale coherent intertwiners \cite{LivineSpeziale} were considered and the asymptotic limit when the spins $\jj$ (all different than $0$) are uniformly rescaled by a large parameter $\Lambda$ was studied. The technical tool used in the calculation is an extended stationary phase method \cite{Hormander}. The amplitude can be written in the following form (see Section \ref{sec:amplitude}):
\begin{equation}\label{eq:vertexamplitude}
A(s)=c \int_{\SL^5} \prod_{i} dg_i \delta(g_5) \int_{(\CP^1\times \CP^1)^{10}}
\left(\prod_{i<j} M_{ij}\right) e^{S(g_i, \z_{ij})},
\end{equation}
where $M_{ij}$ is a measure on $(\CP^1\times \CP^1)^{10}$ parametrized by spinors $\z_{ij}, \z_{ji} \in \mathbb{C}^2\backslash\{0\}$, which is invariant under the scaling $\jj_{ij}\mapsto \Lambda \jj_{ij}$, $c$ is a factor that scales approximately as $c\mapsto c \Lambda^{20}$ for large $\Lambda$ and $S(g_i, z_{ij})$ scales as $S(g_i, z_{ij})\mapsto \Lambda S(g_i, z_{ij})$. The number of critical points depends on the (coherent) spin-network $s$. 

Let us notice that inside the definition of a coherent state some arbitrary phase is hidden, so the total result will be influenced by an arbitrary phase that scales uniformly with the scaling of the labels.
In the analysis of \cite{Frank2, Frank3} this phase was fixed. However the phase choice was done in a global way by considering the whole graph -- not locally by choosing the phases separately for each node. This convention is useful only in considering vertices separately as the phases would be chosen differently for the same nodes glued to two different vertices (see however \cite{Hellmann2013, Hellmann-Kaminskishort, Han2013, Han2011})\footnote{The sum of phases for 4-simplices glued together has a physical meaning.}. The choice of the phase also leads to additional phase terms in the action of \cite{Frank2} ($\pi$ from thin wedges see Section 6.2 of \cite{Frank2}) that can be removed by a more judicious choice of overall coherent state phase. For these reasons we prefer to keep the phase unspecified. This leads to an additional arbitrary overall phase in the asymptotic limit. As a result we do not study the phase of each oscillatory term but a difference of phases between different terms.

In the original formulation of the EPRL model only tetrahedra that are spacelike were considered. As a result $s$ is an $\SU(2)$ spin network. Our goal is to study the asymptotics of the EPRL-FK vertex amplitude generalized to include tetrahedra that are timelike \cite{Conrady2010,Conrady2010b}. We limit to the subcase where the faces are of signature $--$. As a result \cite{Conrady2010} we will restrict to discrete series of the $\SU(1,1)$ representations that are labelled with spins $\jj$\footnote{In the other subcase of faces of signature $+-$ (see \cite{Conrady2010} for explanation of semiclassical origin of the notion) continuous series of SU(1,1) representations are used. We expect that that there may be some non-geometric critical points in the $+-$ case and therefore we leave it for further research (see \ref{sec:types-faces} for a comment).}. We will consider a generalization of $\SU(2)$ spin networks, where the links are labelled with spins $\jj$\footnote{In fact by representations $(\jj,\rho)$ of $\SL$ satisfying simplicity constraint, that is in our case $\rho=2\gamma \jj$.} but the nodes are labelled either with $\SU(2)$ intertwiners or $\SU(1,1)$ intertwiners. In this paper we determine the critical points of the action $S$ for the generalized vertex amplitude written in the form \eqref{eq:vertexamplitude}. This is a crucial step for determination of the asymptotic behaviour of the vertex amplitude.

To each coherent state we will associate vectors $v_{ij}\in M^4$ in Minkowski spacetime (see section \ref{sec:3d-character}) perpendicular to the standard normal $N_i^{can}$ ($e_0=(1,0,0,0)^T$ or $e_3=(0,0,0,1)^T$ depending on the type of embedding) such that bivectors ${\Hodge}(v_{ij}\wedge N_i^{can})$ are spacelike. We will call them the boundary data. Following \cite{Frank2} we will say that (see \ref{sec:3d-character})
\begin{itemize}
 \item they satisfy \textit{closure condition} if $\sum_{j} v_{ij}=0$,
 \item they are \textit{non-degenerate} if for every node $i$ every $3$ out of $4$ vectors $v_{ij}$ are independent.
\end{itemize}
With such vectors for each $i$ we can build, by the Minkowski theorem, \cite{Minkowski, Minkowski2, Bianchi2010} (for simplices in arbitrary signature, see section \ref{sec:bivectors2}) a tetrahedron in $N_i^{can\perp}$, the faces of which have signature $--$. The type of normal determines type of tetrahedra. For timelike normal $e_0$ we have spacelike tetrahedra, for spacelike normal $e_3$, timelike (in fact mixed signature) tetrahedra.
We can determine lengths of the edges of each tetrahedron $i$. 

\begin{figure}\label{fig:4_simplex_notation}
\begin{center}
\includegraphics[scale=0.2]{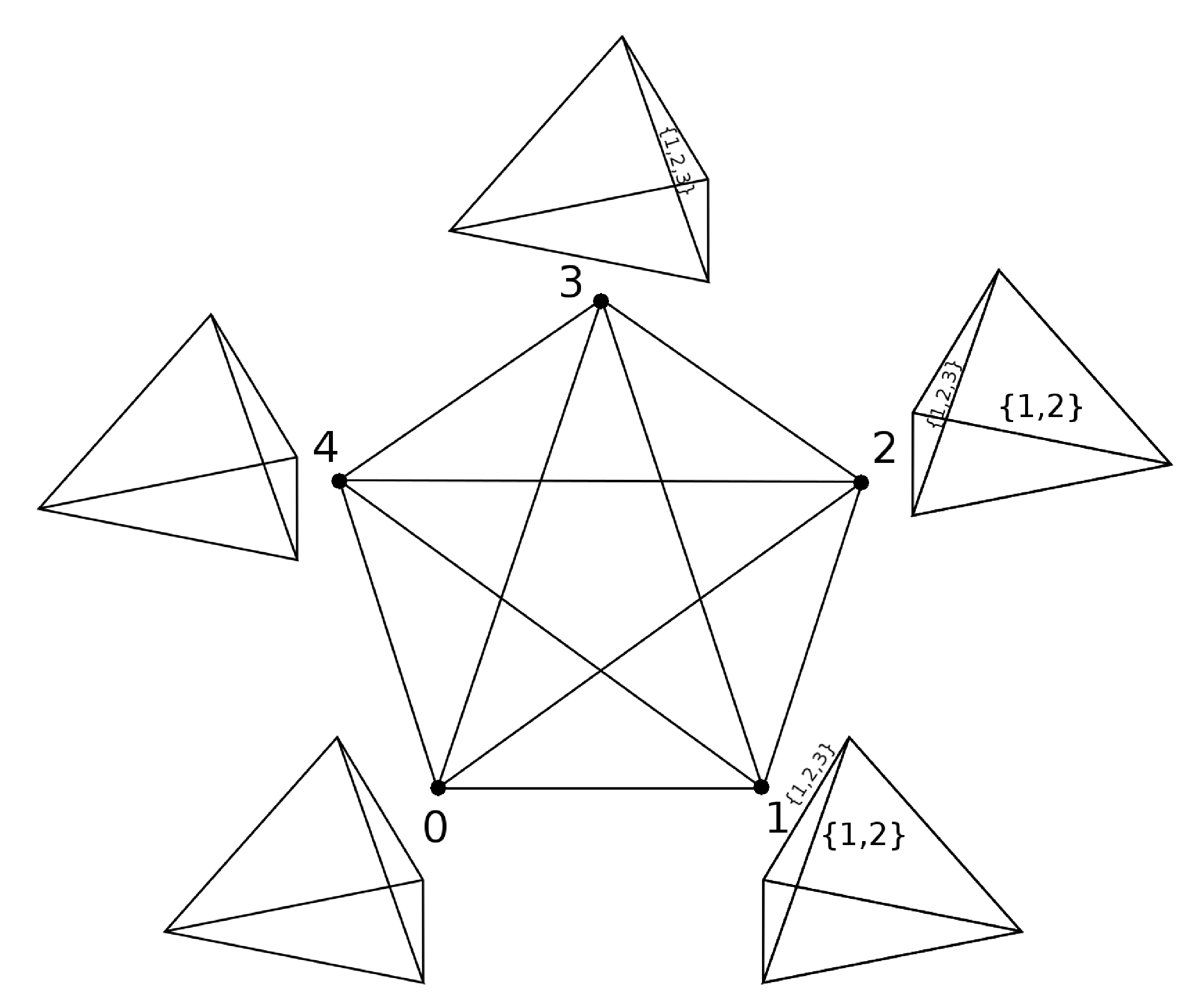}
\caption{We consider a $4$-simplex dual to the graph from figure \ref{fig:4_simplex_graph}. The nodes of the graph label tetrahedra of the simplex, links (sets of two nodes) label faces, edges between faces are labeled by sets of three nodes.}
\end{center}
\end{figure}

Let us consider a topological $4$-simplex, the dual to the graph. The nodes of the graph label tetrahedra of the simplex, links (sets of two nodes) label faces, edges between faces are labeled by sets of three nodes (see figure \ref{fig:4_simplex_notation}. These nodes correspond to three tetrahedra sharing the edge. With every edge we can associate a length coming from application of Minkowski theorem to the boundary data in any node in the set. We will say that the \textit{lengths matching condition} (definition \ref{df:lengths}) is satisfied if these lengths determined from all three tetrahedra are the same. 

If the lengths matching condition is satisfied then we can look for the Gram matrix constructed from such lengths (see section \ref{sec:Gram}). We can reconstruct a unique (up to reflections, rotations and shifts) $4$ simplex from the Gram matrix. There are $5$ cases for the signature of these simplices
\begin{itemize}
 \item Lorentzian ($+---$),
 \item Euclidean ($----$),
 \item split signature ($++--$),
 \item degenerate ($+--$),
 \item degenerate ($---$).
\end{itemize}
By the Minkowski theorem (for simplices in arbitrary signature, see section \ref{sec:bivectors2}) we can also determine orientations of the reconstructed tetrahedra. If they match the orientations of tetrahedra of one of the reconstructed $4$-simplices, we will say that \textit{orientations matching condition} is satisfied (see Section \ref{sec:Gram} definition \ref{df:orient} for precise statement).

Our main result is summarized by the following theorem (see also the diagram on figure 1):
%\begin{figure}
%\begin{center}
%\begin{tikzpicture}[sibling distance=10em,
%  decision/.style = {shape=rectangle, rounded corners,
%    draw, align=center,
%    top color=white, bottom color=blue!20},
%    arrow/.style={left}]]
%  \node [decision] {Lenght matching}
%    child { node [decision] {Orientation matching} child { node [decision] {Reconstructed $4$-simplex\\ non-degenerate?} child { node [decision] {$2$ critical points} edge from parent [->] node [arrow] {Yes} 
%      }
%child { node [decision] {$1$ critical point} edge from parent [->] node [arrow] {No} 
%      }      
%    edge from parent [->] node [arrow] {Yes}
%     }
%   child { node [decision] {No critical\\ points} edge from parent [->] node [arrow] {No} 
%      }
%    edge from parent [->] node [arrow] {Yes}       
%    }
%    child { node [decision] {At most $1$\\ critical point} edge from parent [->] node [arrow] {No} 
%      } ;
%\end{tikzpicture}
%\caption{\label{fig:class}Classification of the critical points of the action $S$. We assume that the boundary data is non-degenerate and satisfies the closure condition.}
%\end{center}
%\end{figure}
\begin{figure}
\begin{center}
\begin{tikzpicture}[sibling distance=10em,
  decision/.style = {shape=rectangle, rounded corners,
    draw, align=center,
    top color=white, bottom color=blue!20},
    arrow/.style={left}]]
    \node [decision] {Closure condition?} child { node [decision] {No critical points} edge from parent [->] node [arrow] {No} 
      }
child { node [decision] {Boundary data\\ non-degenerate?} child { node [decision] {Not addressed\\ in Theorem \ref{thm:1}} edge from parent [->] node [arrow] {No} 
      }
child {node [decision] {Length matching?}
    child { node [decision] {Orientation matching?} child { node [decision] {Reconstructed $4$-simplex\\ non-degenerate?} child { node [decision] {$2$ critical points} edge from parent [->] node [arrow] {Yes} 
      }
child { node [decision] {$1$ critical point} edge from parent [->] node [arrow] {No} 
      }      
    edge from parent [->] node [arrow] {Yes}
     }
   child { node [decision] {No critical\\ points} edge from parent [->] node [arrow] {No} 
      }
    edge from parent [->] node [arrow] {Yes}       
    }
    child { node [decision] {At most $1$\\ critical point} edge from parent [->] node [arrow] {No} 
      } 
 edge from parent [->] node [arrow] {Yes} 
      }
 edge from parent [->] node [arrow] {Yes} 
      } ;
\end{tikzpicture}
\caption{\label{fig:class}Classification of the critical points of the action $S$ according to Theorem \ref{thm:1}.}
\end{center}
\end{figure}
\begin{thm}
\label{thm:1}
If the boundary data $v_{ij}$ does not satisfy the closure condition, the action has no critical point.
Let us assume that the boundary data satisfies the closure condition and is non-degenerate. 
If the lengths matching condition is not satisfied then there exists at most one critical point of the action $S$ with the interpretation of a vector geometry \cite{Frank2}. The vector geometry can occur only if all normals are of the same type. 
If the lengths matching condition is satisfied but orientations matching condition is not, then the action has no critical points. 
If both conditions are satisfied, then let us consider the reconstructed $4$-simplex $\Delta$ for the boundary data:
\begin{itemize}
 \item If the reconstructed $4$-simplex $\Delta$ is Lorentzian then there exist two critical points $(g_i, z_{ij})$ and $(\tilde{g}_i, \tilde{z}_{ij})$. The difference
 $$
 \Delta S = S(g_i, z_{ij})-S(\tilde{g}_i, \tilde{z}_{ij})
 $$
 is given by the Regge action $S_\Delta$ without cosmological constant \eqref{eq:regge} for the flat $4$-simplex $\Delta$
 $$
 \Delta S=2 \iu r S_\Delta \mod 2\pi \iu,
 $$
 where $r=\pm 1$ is the Plebański orientation (definition \ref{df:r}).
 \item If the reconstructed $4$-simplex $\Delta$ is Euclidean or of split signature then there exist two critical points $(g_i, z_{ij})$ and $(\tilde{g}_i, \tilde{z}_{ij})$. The difference $\Delta S$ is given by the Regge action:
 $$
 \Delta S=\frac{2 \iu r}{\gamma} S_\Delta \mod 2\pi \iu,
 $$
 where $\gamma$ is the Barbero-Immirzi parameter.
 \item If the reconstructed $4$ simplex is degenerate then there exists a single critical point.
\end{itemize}
Cases of signature $----$ (and $++--$ respectively) can occur only if all $N_i^{can}$ are $e_0$ (and $e_3$ respectively). Degenerate cases can occur only if all normals are of the same type (either $e_0$ or $e_3$). 
\end{thm}

The case when all $N_i^{can}$ are equal to $e_0$ was proven before (standard EPRL \cite{Frank2, Frank}) but the other cases are our new result. As one can see, the case with timelike tetrahedra is similar to the previously obtained case with spacelike tetrahedra. The main difference is that in addition to Euclidean and Lorentzian signature also split-signature case is present. 

Let us notice at the end that in our convention faces of the tetrahedra have areas $A_{ij}=\frac{1}{2}\rho_{ij}$ instead of $\jj_{ij}$ in the convention of \cite{Frank2}. This is just a total rescaling of the  action by the Barbero-Immirzi parameter $\gamma$. Our convention is compatible with LQG area operator spectrum \cite{Ashtekar-Lewandowski, thiemann}\footnote{In LQG the area of the face would be $A_{ij}=8\pi G\hbar\ \frac{1}{2}\rho_{ij}$ that would lead to the action $\frac{1}{8\pi G\hbar}S_{\Delta}$. We are using units $8\pi G\hbar=1$.}.

%% file: chapter-Vertex_amplitude.tex
%!TEX root = Spin_foams_with_time_like_tetrahedra.tex
\section{Vertex amplitude in extended EPRL model}\label{sec:extended}

In this section we will describe extended EPRL embeddings \cite{EPRL,Conrady2010} and the construction of the vertex amplitude for a simplicial graph $\Gamma$ that was described in the introduction.
The vertex amplitude in the spin foam models \cite{Baez, spinfoams2, Alexandrov2011} is the evaluation of a certain spin network. This spin network consists of links labelled by irreducible unitary representations of $\SL$ and nodes labelled by invariants in the tensor product of representations from links meeting in the node. These invariants are elements of a certain distribution space. Na\"{\i}ve evaluation would be the contraction of invariants according to the prescription given by the spin network. This prescription, while valid for a compact group, gives infinity for $\SL$ and some gauge fixing \cite{Barrett-Crane, Baez2001, Barrett2000} is necessary (see \ref{sec:amplitude}).

The spin network we consider is defined on a graph with five nodes $0,\ldots, 4$ connecting every node to every other one with precisely one link. We parametrize links by the two nodes $(ij)$ that they connect. The invariants at the nodes are given by group averaging the EPRL Y-maps (see section \ref{sec:Ymap}),
\begin{align}
 \Inv_{\St(N^{can}_n)} \bigotimes_{l\ni n} \Heprl&\rightarrow \Inv_{\SL} \bigotimes_{l\ni n} D_{\jj_l,\rho_l}\\
 \psi&\rightarrow \int_{\SL/\St(N^{can}_n)} d[g]\ g\psi, 
\end{align}
where the group $\St(N^{can}_n)\subset\SL$ is a subgroup stabilizing the normal $N^{can}_n$, $l$ are links meeting in the node $n$ and $\Heprl$ is certain representation of $\St(N^{can}_n)$. Invariants of the subgroup $\St(N^{can}_n)$ are thus together with labels and type of the subgroup (determined by $N^{can}_n$) the boundary data for the vertex amplitude.  

We are interested in the asymptotic analysis for large labels. In such a case one needs to specify a family of boundary states with certain semiclassical limit. Examples of such states are given by coherent states (they scale nicely with the scaling of labels) integrated over the subgroup:
\begin{equation}
 \psi=\int_{\St(N^{can}_n)} dg\ g\left(\otimes_{l\ni n}\Psi_l\right).
\end{equation}
We will now describe in unified manner the choice of coherent states and their image under the Y-map. We will present also the form of the vertex amplitude (using bilinear forms $\beta$ used to contract invariants \cite{Frank2}).

\subsection{Notation}

Let us summarize notation about spinors and $\SL\rightarrow \SO_+(1,3)$ double cover. We use signature $+---$.

Let us introduce $\sigma_\mu$ as follows 
\begin{equation}
 \sigma_0={\mathbb I},\ \sigma_i \text{ Pauli matrices}
\end{equation}
We have the following isomorphism from Minkowski space $M^4$ into Hermitian $2\times 2$ matrices 
\begin{equation}
 M^4\ni N^\mu\rightarrow \eta_N=N^\mu\sigma_\mu\text{ with } N_\mu N^\mu=\det \eta_N
\end{equation}
The symplectic form is defined as
\begin{equation}
\label{eq:symp}
 \omega=\bmat{ 0 & 1\\-1& 0},\quad \bar{\omega}=\omega,\ \omega^{-1}=\omega^T=-\omega,
\end{equation}
We will also use a notation for two spinors $\z$ and $\mathbf{w}$
\begin{equation}
 [\z,\mathbf{w}]=\z^T\omega \textbf{w}=z_0w_1-z_1w_0,\quad \z=\bvec{z_0}{z_1},\quad \mathbf{w}=\bvec{w_0}{w_1}
\end{equation}
Let us also introduce
\begin{equation}
 \hat{\sigma}_\mu=-\omega \sigma_\mu^T\omega=\left\{\begin{array}{ll}
                                                  \sigma_0 & \mu=0\\
                                                  -\sigma_i & \mu=i
                                                 \end{array}\right.,\quad \hat{\eta}_N=N^\mu\hat{\sigma}_\mu \text{ Hermitian}
\end{equation}
The following holds
\begin{align}
 \sigma_\nu\hat{\sigma}_\mu+\sigma_\mu\hat{\sigma}_\nu=2g_{\mu\nu}{\mathbb I}\\
 \hat{\sigma}_\nu\sigma_\mu+\hat{\sigma}_\mu\sigma_\nu=2g_{\mu\nu}{\mathbb I}
\end{align}
Together $\sigma_\mu$ and $\hat{\sigma}_\mu$ form $\gamma_\mu$ matrix, but as we always work in either self-dual or anti-self-dual representation we prefer such a notation.

From standard commutation relation we have
\begin{equation}
 N^\mu M_\mu=\frac{1}{2}\tr \eta_N \hat{\eta}_M,
\end{equation}
The isomorphism $\pi$ from $\SL$ to $\SO_+(1,3)$ is defined by
\begin{equation}
  \pi(u)(N)^\mu\sigma_\mu=u\ (N^\mu\sigma_\mu)\ u^\dagger
\end{equation}
Lie algebra of $\SO_+(1,3)$ can be identify with bivectors ($2$ forms). The action of bivectors on vectors is defined by contraction with the use of the metric
\begin{equation}
 B(v)=B\lrcorner v
\end{equation}
The identification of $so(1,3)$ with $sl(2,\C)$ is then given on the basic simple bivectors by
\begin{equation}
so(1,3)\ni B=a\wedge b\longrightarrow \B=\frac{1}{4}(\eta_a\hat{\eta}_b-\eta_b\hat{\eta}_a)\in sl(2,\C)
\end{equation}
With the right choice of the orientation the Hodge ${\Hodge}$ operation corresponds to multiplication by $i$ of the traceless matrix.

\subsection{Representations of $\SL$}\label{sec:repr}

Unitary irreducible representations ($D_{\j,\rho}$) from the principal series $(\jj,\rho)$ are functions 
\begin{equation}
 \C^2\setminus\{0\}\ni \z\rightarrow \Psi(\z)\in\C,\quad,\quad \z=\bvec{z_0}{z_1}
\end{equation}
satisfying the condition
\begin{equation}
 \Psi(e^{i\phi+r} \z)=e^{i(2\jj\phi+\rho r)-2r} \Psi(\z),\quad 
\end{equation}
with the action of $\SL$ defined by
\begin{equation}
 g\Psi(\z)=\Psi(g^T \z)
\end{equation}
We are using the convention of \cite{Frank2}, as opposed to that of \cite{Bargmann}. The latter is equivalent to the action
\begin{equation}
 g\Psi(\z)=\Psi(g^{-1}\z)
\end{equation}
These two actions can be related due to
\begin{equation}
 g^{-1}=\omega^{-1}g^T\omega
\end{equation}
The scalar product for two such functions $\Psi_1$ and $\Psi_2$ is defined as follows:
Let us introduce a form
\begin{equation}\label{eq:theform}
 \overline{\Psi_1(\z)}\Psi_2(\z)\Omega_{\z}
\end{equation}
where
\begin{equation}\label{eq:Omegaz}
\Omega_\z:=\frac{\iu}{2}\left(z_0 dz_1-z_1 dz_0 \right)\wedge \left(\overline{z_0} d\overline{z}_1 - \overline{z}_1 d \overline{z}_0 \right),
\quad \z=\bvec{z_0}{z_1} 
\end{equation}
The form \eqref{eq:theform} is invariant under the scaling transformation
\begin{equation}
 \z\rightarrow \lambda \z,\quad \lambda\in \C^*
\end{equation}
and is annihilated by the generator of this transformation, thus it descends to a form on $\CP^1$
\begin{equation}
 \langle \Psi_1,\Psi_2\rangle=\int_{\CP^1} \overline{\Psi_1(\z)}\Psi_2(\z)\Omega_\z
\end{equation}

\subsection{Subgroups preserving the normal $N$}
\label{sec:subgroup}

The subgroup of $\SO_+(1,3)$ that preserves the normal $N$ is the image of the subgroup of $\SL$ that preserves $\eta_N$ as follows
\begin{equation}\label{eq:sub}
 \St(N)=\{g\in \SL\colon \pi(g)N=N\}=\{g\in \SL\colon g\eta_N g^\dagger=\eta_N\}
\end{equation}
this is equivalent to $g^T$ preserving Hermitian form defined by
\begin{equation}
 \langle \z_1,\z_2\rangle_N=\z_1^\dagger \eta_N^T \z_2=\z_2^T \eta_N\widebar{\z_1}
\end{equation}
It follows from the fact that
\begin{equation}
 \langle g^T\z_1,g^T\z_2\rangle_N=\z_2^T \underbrace{g\eta_N\bar{g}^T}_{=\eta_N}\widebar{\z_1}=\langle \z_1,\z_2\rangle_N
\end{equation}
for $g\in \St(N)$.

\subsection{The Y map}
\label{sec:Ymap}
Let us review the generalized EPRL construction. As explained in the introduction, we will work only with spacelike surfaces. Let us consider a normal vector $N$
\begin{equation}
 N_\mu N^\mu=t,\quad t\in\{-1,1\}
\end{equation}
where $t=1$ for timelike, $t=-1$ for spacelike vector. Let us consider stabilizing group $\St(N)$.
In the case of $t=1$, 
\begin{equation}
 \St(N)\cong \SU(2)
\end{equation}
and it is exactly $\SU(2)$ for $N=e_0=(1,0,0,0)$. For $t=-1$ 
\begin{equation}
 \St(N)\cong \SU(1,1)
\end{equation}
and it is exactly $\SU(1,1)$ for $N=e_3=(0,0,0,1)$.

We will consider only standard normals $e_0=(1,0,0,0)$ and $e_3=(0,0,0,1)$. We will also denote either of them by $N^{\text{can}}$ as we would like to work in a unified setup. We denote $t^{can}=N^{can}\cdot N^{can}$.

The EPRL $Y$ map is an embedding of the unitary representation of the stabilizing group $\St(N)$ into unitary representation $(\j,\rho=2\gamma\jj)$ of $\SL$ that satisfies a certain extremality condition. Let us recall the choice of $\Heprl$ made by \cite{EPRL, Conrady2010}
\begin{enumerate}
 \item Spin $\j$, $D_{\j}$  representation of $\SU(2)$ embedded into $D_{\j,\rho}$ for $N^{can}=e_0$
 \item Discrete series $D_{\j}^\pm$ of spin $\j$ representations of $\SU(1,1)$ embedded into $D_{\j,\rho}$ for $N^{can}=e_3$
\end{enumerate}

\subsection{Coherent states}

All three families of representations have certain common features. Let us consider the generator of rotations around the $z$ axis. We can introduce bases of eigenfunctions. In all three cases there exists an extreme eigenvalue. The corresponding eigenfunctions are (Perelomov-)coherent states \cite{Perelomov, Conrady2010a, Bargmann} (see Appendix \ref{app:decomp})\footnote{We do not need to consider $-\j$ for $\SU(2)$ since it will be obtained from $\j$ eigenstate by rotation.}
\begin{itemize}
 \item For $N^{can}=e_0$
 \begin{equation}
 F_{\j\j}(\z)=\sqrt{\frac{2\j+1}{2\pi}} \braket{\z}{\z}_{N^{can}}^{\iu \frac{\rho}{2}-1-\j} \braket{\n_0}{\z}_{N^{can}}^{2\j}, 
 \end{equation}
 \item For $N^{can}=e_3$
 \begin{align}
&F^{+}_{\j\j}(\z)=\sqrt{\frac{2\j-1}{\pi}}\theta(\braket{\z}{\z}_{N^{can}}) \braket{\z}{\z}_{N^{can}}^{\iu \frac{\rho}{2}-1+\j} \braket{\z}{\n_0}_{N^{can}}^{-2\j},\\
&F^{-}_{\j-\j}(\z)=\sqrt{\frac{2\j-1}{\pi}}\theta(-\braket{\z}{\z}_{N^{can}})(-\braket{\z}{\z}_{N^{can}})^{\iu \frac{\rho}{2}-1+\j} (-\braket{\z}{\n_1}_{N^{can}})^{-2\j},
 \end{align}
\end{itemize}
where $\n_0=(1,0)$ and $\n_1=(0,1)$ and $\theta$ is the Heaviside step function.

All other coherent states are obtained by transforming these basic coherent states by group action of $\St(N^{can})$. In fact these states can be parametrized by spinors:
\begin{equation}
\Psi^{\n}(\z)=\sqrt{\frac{2\j+1}{2\pi}}\braket{\z}{\z}_{N^{can}}^{\iu \frac{\rho}{2}-1-\j}\braket{\n}{\z}_{N^{can}}^{2\j},\quad \braket{\n}{\n}_N^{can}=1,\ N^{can}=e_0
\end{equation}
and for $N^{can}=e_3$
\begin{align}
\Psi^{+,\n^+}(\z)&=\theta(\braket{\z}{\z}_{N^{can}})\sqrt{\frac{2\j-1}{\pi}}\braket{\z}{\z}_{N^{can}}^{\iu \frac{\rho}{2}-1+\j}\braket{\z}{\n^+}^{-2j}_{N^{can}}\\
\Psi^{-,\n^-}(\z)&=\theta(-\braket{\z}{\z}_{N^{can}})\sqrt{\frac{2\j-1}{\pi}}(-\braket{\z}{\z}_{N^{can}})^{\iu \frac{\rho}{2}-1+\j}(-\braket{\z}{\n^-}_{N^{can}})^{-2\j},
\end{align}
where $\braket{\n^\pm}{\n^\pm}_{N^{can}}=\pm 1$.

\subsection{Bilinear invariant form}

Let us recall the definition of the bilinear, $\SL$-invariant form \cite{gelfand5}
\begin{align}
\beta(\Phi, \Phi'):=&\frac{\sqrt{\left(\frac{\rho}{2}\right)^2+\jj^2}}{\pi}
\int_{\CP^1\times \CP^1}\underbrace{\Omega_{\z}\wedge \Omega_{\z'}|[\z,\z']|^{-2} [\z,\z']^{-\iu\frac{\rho}{2}-\jj} {\overline{[\z,\z']}}^{-\iu \frac{\rho}{2}+\jj}}_{\beta(\z,\z')}\nonumber\\
&\cdot{}\Phi(\z)\Phi'(\z')
\end{align}
where $\Phi, \Phi'$ are elements of the $\SL$-representation $(\jj,\rho)$, and $\Omega_\z$ is defined in \eqref{eq:Omegaz}.

\subsection{The vertex amplitude}
\label{sec:amplitude}

We can now consider the vertex amplitude in the spin foam model. Let us consider the pentagonal graph \cite{Frank2} with five nodes and links connecting each node with each other.
\begin{itemize}
 \item for every node $i$ we choose a canonical normal $N_i^{can}$ that determines an embedded subgroup (either $\SU(2)$ or $\SU(1,1)$)
 \item for every directed link $ij$ starting in the node $i$ we chose a type of embedded representation (in case of $\SU(1,1)$ it can be $D_{\j_{ij}}^+$ or $D_{\j_{ij}}^-$ and for $\SU(2)$ it is $D_{\j_{ij}}$)
 \item for every directed link $ij$ starting in the node $i$ we chose a spinor $\n_{ij}$ that determines a coherent state (in representation $(\jj_{ij},\rho_{ij})$ of $\SL$) that we will denote by $\Psi_{ij}(\z)$
 \end{itemize}
This data is a boundary data for the vertex amplitude of the extended EPRL Spin Foam model \cite{Frank2, Conrady2010}.

The vertex amplitude is given by an integral
$$
A_v=\int_{\SL^5} \prod_{i} dg_i \delta(g_5) \prod_{i<j}\beta(g_i\Psi_{ij}, g_j\Psi_{ji}), 
$$
where 
$$
g_i\Psi_{ij}(\z)=\Psi_{ij}(g_i^T\z)
$$
and the $\delta(g_5)$ is a gauge fixing. We will denote $\z_{ij}$ and $\z_{ji}$ the variables $\z$ and $\z'$ in the $\beta$ integral for nodes $i,j$.

%% file: chapter-Stationary_phase.tex
%!TEX root = Spin_foams_with_time_like_tetrahedra.tex
\section{Stationary phase approximation}\label{sec:action}

In this section we will write the amplitude from section \ref{sec:amplitude} in a form suitable for the stationary phase method. 
We will consider the uniform scaling of the representations (we remind that $\rho$ and $\jj$ are related by the simplicity constraint)
$$
(\jj_{ij},\rho_{ij})\mapsto (\Lambda \jj_{ij},\Lambda \rho_{ij}),\quad \Lambda\in{\mathbb N}_+
$$
and we want to organize the integral in the form $c(\Lambda)\int f e^{\Lambda S}$. This will provide a definition of the action. The amplitude is written as integral of the product of many terms and we will focus on each term separately. The total action will be the sum of these partial contributions.

We will write conditions for the critical points of this action in terms of holomorphic derivatives, since that is the form that we will use later.

We will assume that $(\jj_{ij},\rho_{ij})\not=0$. Again, we expect that representations $(\jj_{ij},\rho_{ij})=(0,0)$ are of special interest. They should correspond to null surfaces.

\subsection{Scaled amplitude}

We denote the normal stabilized by the embedded group by $N_i^{can}$,
\begin{equation}
 N_i^{can}\cdot N_i^{can}=t_i^{can},\quad \langle \n_{ij},\n_{ij}\rangle_{N_i^{can}}=s_{ij},\quad s_{ij},t_{i}^{can}\in\{-1,1\}
\end{equation}
The coherent states decompose under the scaling in the following way, 
$$
\Psi_{ij}^\Lambda(\z_{ij})=c_{ij}(\Lambda)m_{ij}(\z)p_{ij}(\z_{ij})^\Lambda
$$
where $m_{ij},p_{ij}$ independent of $\Lambda$
\begin{equation}
 m_{ij}(\z_{ij})=\theta(s_{ij}\langle \z_{ij}|\z_{ij}\rangle_{N_i^{can}})(s_{ij}\langle \z_{ij}|\z_{ij}\rangle_{N_i^{can}})^{-1} 
\end{equation}
with notice that for $N^{can}_i=e_0$, $s_{ij}\langle \z_{ij}|\z_{ij}\rangle_{N_i^{can}}>0$ on the whole $\C^2\setminus \{0\}$.
Function $p_{ij}$ is given explicitly (uniform description) by
\begin{equation}
 p_{ij}=p_{ij}^\rho p_{ij}^{\jj}
\end{equation}
where
\begin{align}
 p_{ij}^\rho=\left( s_{ij}\langle g_i^T \z_{ij},g_i^T \z_{ij}\rangle_{N_i^{can}}\right)^{i\frac{\rho_{ij}}{2}},\quad
 p_{ij}^{\jj}=\left\{\begin{array}{ll}
                 p_{ij}^{aux}, & t_i^{can}=-1\\ \widebar{p_{ij}^{aux}}, & t_i^{can}=1
                \end{array}\right.
\end{align}
with
\begin{equation}
 p_{ij}^{aux}=\left(\frac{s_{ij}\langle g_i^T \z_{ij}, \n_{ij}\rangle_{N_i^{can}}}{(s_{ij}\langle g_i^T\z_{ij},g_i^T\z_{ij}\rangle_{N_i^{can}})^{1/2}}\right)^{t_i^{can}\ 2\jj_{ij}}
\end{equation}
The factors $c_{ij}(\Lambda)$ are independent from $\z$ and $g$
\begin{equation}
 c_{ij}(\Lambda)=\left\{\begin{array}{ll}
                \sqrt{\frac{2\Lambda\j_{ij}+1}{2\pi}}, &  N_i^{can}=e_0\text{ and } D_{\j} \text{ representation}\\
                \sqrt{\frac{2\Lambda\j_{ij}-1}{\pi}}, &  N_i^{can}=e_3\text{ and } D^\pm_{\j} \text{ representation}
               \end{array}\right.
\end{equation}
Similarly, we can decompose the integral kernel of $\beta$ as 
$$
\beta_{ij}^\Lambda(\z_{ij},\z_{ji})=c_{ij}^\beta(\Lambda) m^\beta_{ij}(\z_{ij},\z_{ji})p^\beta_{ij}(\z_{ij},\z_{ji})^\Lambda
$$
where $c_{ij}^\beta(\Lambda)=\Lambda\frac{\sqrt{\left(\frac{\rho_{ij}}{2}\right)^2+\jj_{ij}^2}}{\pi}$,
\begin{align}
m^\beta_{ij} (\z_{ij},\z_{ji})&=\Omega_{\z_{ij}}\wedge \Omega_{\z_{ji}}|[\z_{ij},\z_{ji}]|^{-2}\\ 
p^\beta_{ij}(\z_{ij},\z_{ji})&= [\z_{ij},\z_{ji}]^{-\iu\frac{\rho_{ij}}{2}-\jj_{ij}} {\overline{[\z_{ij},\z_{ji}]}}^{-\iu \frac{\rho_{ij}}{2}+\jj_{ij}}
\end{align}
The differential form
$$
M_{ij}:=m^\beta_{ij} (\z_{ij},\z_{ji})m_{ij}(g_i^T\z_{ij})m_{ji}(g_j^T\z_{ji})\qquad
$$
descends to $\CP^1\times \CP^1$ as a measure (smooth in the interior of its support). 

Similarly\footnote{The scaling invariance of $P_{ij}$ and descendent property of $\mu_{ij}$ is a general fact.}, as the integrand form is invariant under rescaling of $\z$ also the part scaled uniformly with $\Lambda$ need to be invariant. The amplitude
$$
P_{ij}:=p^\beta_{ij}(\z_{ij},\z_{ji})p_{ij}(g_i^T\z_{ij})p_{ji}(g_j^T\z_{ji})
$$
is invariant under rescalings $\z_{ij}\mapsto \lambda_{ij} \z_{ij}$, $\z_{ji}\mapsto \lambda_{ji} \z_{ji}$ and thus projects down to a function on $\CP^1\times \CP^1$.

We will write the vertex amplitude as an integral of the form:
$$
A_v(\Lambda\jj, \Lambda\rho)= c(\Lambda)
\int_{\SL^5} \prod_{i} dg_i \delta(g_5) \int_{(\CP^1\times \CP^1)^{10}}
\left(\prod_{i<j} M_{ij}\right) \left(\prod_{i<j} P_{ij}\right)^\Lambda,
$$
where $c(\Lambda)\sim  \Lambda^{20}$ for large $\Lambda$, 
\begin{equation}\label{eq:c-scaling}
 c(\Lambda)=\prod_{i<j}c_{ij}^\beta(\Lambda)\ \prod_{i\not=j} c_{ij}(\Lambda)
\end{equation}

\subsection{Action}
We observe, that $p^{\beta}$ is of the form
$$
p^{\beta}=\frac{\alpha}{\overline{\alpha}}, \qquad  \alpha=[\z_{ij},\z_{ji}]^{-\iu\frac{\rho_{ij}}{2}-\jj_{ij}},
$$
so we conclude $|p^{\beta}|=1$. Similarly we will prove that $|p_{ij}|\leq 1$ (lemma \ref{lm:reality2}). 

Now we define 
\begin{equation}
S=\sum_{i<j}\ln P_{ij}
\end{equation}
and 
\begin{equation}
S_{ij}=\ln p_{ij},\quad
S^{\beta}_{ij}=\ln p_{ij}^\beta.
 \end{equation}
We note that as written above, the various $S$ are multi-valued functions defined up to multiplicity of $2\pi i$, but as long as the product of the $P_{ij}$ is nonzero, we can always work in a local branch. We have
$$
S(g,\z)
=\sum_{i<j}S_{ij}(g_i,\z_{ij})+S_{ji}(g_j,\z_{ji})+S_{ij}^{\beta}(\z_{ij},\z_{ji}),
$$
and
$$
S^{\beta}_{ij}(\z_{ij},\z_{ji})=-\iu \rho_{ij}\ln |[\z_{ij},\z_{ji}]|-\jj_{ij} \ln [\z_{ij},\z_{ji}]+\jj_{ij} \ln \overline{[\z_{ij},\z_{ji}]}.
$$
The action $S_{ij}$ of the amplitude depending on $g_i$ and $\z_{ij}$ is given by
\begin{equation}
 e^{S_{ij}}=e^{S^{\rho}_{ij}}e^{S^\jj_{ij}}
\end{equation}
where
\begin{equation}
 e^{S^\rho_{ij}}=p_{ij}^\rho,\quad e^{S^j_{ij}}=p_{ij}^{\jj}
\end{equation}
Similarly we can also define auxiliary $e^{S^{aux}_{ij}}=p_{ij}^{aux}$.

The integration is restricted to the domain where for all $i\not=j$
\begin{equation}
s_{ij}\langle g_i^T \z_{ij},g_i^T \z_{ij}\rangle_{N_i^{can}}>0 
\end{equation}

\subsection{Gauge symmetries of the action}

The following transformations labeled by $g\in SL(2,\C)$, $\lambda_{ij}\in \C^*$
\begin{equation}
 g_i'=gg_i,\quad \z_{ij}'=\lambda_{ij} (g^T)^{-1} \z_{ij}
\end{equation}
preserve the non-gauge fixed action. The $g$-part of this transformation is gauge fixed by the $\delta(g_5)$ term. We can still consider variations with respect to $g_5$ of the action, they are just not independent of the others.  

The subgroup of gauge transformations with $g=1$ will be called $\CP^1$-gauge transformations.

\subsection{Critical points}

%We will now study the behaviour of the vertex amplitude when the representation labels are rescaled  $(\jj_{ij},\rho_{ij})\mapsto (\Lambda \jj_{ij},\Lambda \rho_{ij})$ and  $\Lambda \to \infty$.

According to the stationary phase method\footnote{Justification of the applicability is left for future research, as discussed in the introduction.} 
\cite{Hormander}, in the large $\Lambda$ regime the amplitude is dominated by contributions from {\it the critical points (manifolds)} where the following conditions are satisfied: the reality condition 
$$
\Re(S)=0 
$$
and the stationary point condition 
$$
\delta^\R S=0.
$$
Here $\delta^\R$ denotes the standard variation, to be distinguished from the holomorphic variation $\delta$ that we will introduce later.

We will use variational calculus (form) notation for derivatives: We denote
\begin{equation}
 \delta^{\R} S(g,\z)=\lim_{\epsilon\rightarrow 0}\frac{1}{\epsilon}(S(g(\epsilon), \z(\epsilon))- S(g,\z))
\end{equation}
where the variations are of the form:
\begin{equation}
g(\epsilon)_i=g_{i} e^{\epsilon\delta^{\R} g_{i}},\quad \z(\epsilon)_{ij}= e^{\epsilon\delta^{\R} g_{ij}^T}\z_{ij} 
\end{equation}
with $\delta^{\R} g_i, \delta^{\R} g_{ij}$  taking values in the Lie algebra of $\SL$. Let us notice that
\begin{equation}
 \delta^\R \z_{ij}=\delta^{\R} g_{ij}^T\z_{ij}\ .
\end{equation}
We will also use variations with respect to single variables that we will denote as
\begin{equation}
 \delta^\R_{g_i}S,\quad \delta^\R_{\z_{ij}}S,
\end{equation}
that are variation with single $\delta^\R g_i$ (and respectively single $\delta^\R\z_{ij}$) nonzero.

A critical point for the given boundary data (spinors $\n_{ij}$, vectors $N_i^{can}$ and the types of embedded representations) consists of a collection of 
\begin{equation}
\{ g_i, \z_{ij},\quad i\not= j\}
\end{equation}
satisfying the reality condition $\Re(S)=0$ and the stationary point conditions ($\delta^{\R}S=0$) i.e.
\begin{align}
 \forall_{ij} \delta^{\R}_{\z_{ij}}S_{ij}+\delta^{\R}_{\z_{ij}} S^\beta_{ij}=0\label{end-middle}\\
 \forall_i \sum_j \delta^{\R}_{g_i} S_{ij}=0.
\end{align}
There are gauge transformations acting on the critical points.

%% file: chapter-Reality-conditions.tex
%!TEX root = Spin_foams_with_time_like_tetrahedra.tex
\subsection{Reality condition and holomorphic derivatives}

We will consider action $S$ as a function of holomorphic $g$ and antiholomorphic $\bar{g}$ variables\footnote{
For our purpose of computing first derivatives, we can assume that all variables are group elements in $SL(2,\C)$ considering $g_{ij}$ instead of $\z_{ij}$.
}. We will prove later (see lemma \ref{lm:sij}) that
\begin{equation}
 \Re S\leq 0
\end{equation}
However this holds only when $\bar{g}$ and $g$ are complex conjugated (we will call this set real manifold). We will consider now complexified manifold where these group elements are independent. We denote holomorphic and antiholomorphic variations with respect to these group elements by $\delta_g S$ and $\delta_{\bar{g}} S$.

\begin{df}
A point satisfies the reality condition if it is in the real manifold and $\Re S=0$.
\end{df}

\begin{lm}
 If the action $S=\sum_\alpha S_\alpha$ and $\forall_\alpha \Re S_\alpha\leq 0$ on the real manifold then
 \begin{align}
  &\Re S\leq 0 \text{ on the real manifold}\\
  &\Re S=0\Leftrightarrow\forall_\alpha \Re S_\alpha=0
 \end{align}
\end{lm}

\begin{proof}
The inequality 
 \begin{equation}
  \Re S=\sum_\alpha\Re S_\alpha\leq 0
 \end{equation}
is saturated only if $\forall_\alpha \Re S_\alpha=0$
\end{proof}

\begin{lm}\label{lm:SxSxbarn}
 When reality condition are satisfied then
 \begin{equation}
  \delta_g S=-\delta_g \bar{S},\quad \delta_{\bar{g}} S=-\widebar{\delta_g S}
 \end{equation}
 when in the second equality we took $\delta \bar{g}=\widebar{\delta g}$
\end{lm}

\begin{proof}
 The real variation of $\Re S$ is zero when reality condition are satisfied (from extremality) so
 \begin{equation}
  0=\delta^\R\Re S=\delta_g \Re S+\delta_{\bar{g}} \Re S, \text{ where }\delta \bar{g}=\widebar{\delta g}
 \end{equation}
 that can be written as
 \begin{equation}
  \delta_g S+\delta_{\bar{g}} S+\delta_g \bar{S}+\delta_{\bar{g}} \bar{S}=0
 \end{equation}
 but we also have 
 \begin{equation}
  \widebar{\delta_g S}=\delta_{\bar{g}} \bar{S},\quad \delta_{\bar{g}} S=\widebar{\delta_g \bar{S}}
 \end{equation}
 The equality is thus
 \begin{equation}
  \Re (\delta_g S+\delta_g \bar{S})=0,
 \end{equation}
 From arbitrariness of $\delta g$
 \begin{equation}
  \delta_g S+\delta_g \bar{S}=0,\quad \delta_g S=-\delta_g \bar{S}
 \end{equation}
 We can also compute
 \begin{equation}
  \delta_{\bar{g}} S=\widebar{\delta_g \bar{S}}=-\widebar{\delta_g S}
 \end{equation}
 for $\delta \bar{g}=\widebar{\delta g}$.
\end{proof}

\begin{lm}
 When the reality conditions are satisfied, then the stationary point conditions are equivalent to vanishing holomorphic derivatives.
\end{lm}

\begin{proof}
It follows from lemma \ref{lm:SxSxbarn} that
 \begin{equation}
  \delta^\R S=\delta_g S+\delta_{\bar{g}} S=\delta_g S-\widebar{\delta_g S}=i\Im \delta_g S
 \end{equation}
 where we took $\delta \bar{g}=\widebar{\delta g}$. As holomorphic variables can be multiplied by $i$ we see that a vanishing holomorphic derivative is equivalent to a vanishing real variation.
\end{proof}

\subsection{Holomorphic critical point conditions}

Let us now rephrase the critical point conditions in holomorphic language:

A critical point for a given set of boundary data (spinors $\n_{ij}$, vectors $N_i^{can}$ and the type of embedded representation) consists of
a collection of 
\begin{equation}
 \{g_i, \z_{ij}\}
\end{equation}
(on the real manifold) satisfying the reality condition 
\begin{equation}
\forall_{ij}\ \Re S_{ij}=0
\end{equation}
and the stationary point conditions ($\delta S=0$)
\begin{align}
 \forall_{ij} \delta_{\z_{ij}}S_{{ij}}+\delta_{\z_{ij}} S^\beta_{ij}=0\label{end-middle-hol}\\
 \forall_i \sum_j \delta_{g_i} S_{{ij}}=0.
\end{align}
There are gauge transformations acting on the critical points.

%% file: chapter-Bivectors-spinors.tex
%!TEX root = Spin_foams_with_time_like_tetrahedra.tex
\section{Traceless matrices, spinors and bivectors}\label{sec:traceless}

In this section we will describe the connection between spinors and Lie algebra elements of $\SL$ (traceless matrices). Using this connection, we will show that critical points can be described in terms of traceless matrices satisfying certain conditions (we will call it $\SL$ solution). This will allow us later to translate it into geometric language of the Lorentz group and bivectors. We will also compute the (difference of the) phase between two critical points.

\subsection{Traceless matrices}

We will now recall some properties of spinors that we will use to translate the critical point conditions into the language of traceless matrices. Proofs are provided for the convenience of the reader in Appendix \ref{app:traceless}.

Let us assume that $\delta g$ is traceless then
\begin{equation}
 \omega \delta g^T+\delta g\omega=0
\end{equation}
for the symplectic form $\omega$ of \eqref{eq:symp}, and for spinors $\u$ and $\v$
\begin{equation}
 [\u,\delta g^T \v]=[\v,\delta g^T \u]=\frac{1}{2}\tr (\v\u^T+\u\v^T)\omega \delta g^T
\end{equation}
as $(\v\u^T-\u\v^T)\omega =[\u,\v]{\mathbb I}$.  Every traceless matrix with eigenvalues $\lambda,-\lambda\in \C$ can be written in the form
\begin{equation}
 \M=\lambda(\v\u^T+\u\v^T)\omega, \quad [\u,\v]=1,
\end{equation}
with $\u$, $\v$ uniquely determined up to common rescaling (see lemma \ref{lm:comp} in Appendix \ref{app:traceless}). This property allows us to work purely in terms of traceless matrices.

\subsection{Variations $\delta_{g_i} S_{ij}$ and reality conditions}
\label{se_derivatives}

In order to avoid overburden the notation  we will in this subsection suppress indices $ij$ and write $S_{ij}(g_i,\z_{ij})$ as $S(g,\z)$.
Let us consider normal $N^{can}$ and normalized spinor $\n$ and $\rho\in\R$, $2\j\in\Z$
\begin{equation}
 N^{can}\cdot N^{can}=\det\eta_{N^{can}}=t,\quad \langle \n,\n\rangle_{N^{can}}=\n^\dagger \eta_{N^{can}}^T \n=s,\quad s,t\in\{-1,1\}
\end{equation}
the action part of the amplitude depending on $g$ and $\z$
\begin{equation}
 e^{S}=e^{S^{\rho}}e^{S^\j}
\end{equation}
where (the integration is restricted to the domain where 
\begin{equation}
 s\langle g^T \z,g^T \z\rangle_{N^{can}}>0
\end{equation}
and 
\begin{align}
 e^{S^\rho}=\left( s\langle g^T \z,g^T \z\rangle_{N^{can}}\right)^{i\frac{\rho}{2}}\\
 e^{S^\j}=\left\{\begin{array}{ll}
                 e^{S^{aux}} & t=-1\\ \widebar{e^{S^{aux}}} & t=1
                \end{array}\right.
\end{align}
where
\begin{equation}
 e^{S^{aux}}=\left(\frac{s\langle g^T \z, \n\rangle_{N^{can}}}{(s\langle g^T\z,g^T\z\rangle_{N^{can}})^{1/2}}\right)^{t\ 2\j}
\end{equation}
In fact although there are many different cases one can prove
(see appendix \ref{app:traceless}) that on real manifold
\begin{equation}
  \Re S^\rho=0,\quad \Re S^{\j}\leq 0
\end{equation}
and the equality $\Re S^{\j}=0$ holds if and only if 
\begin{equation}
 \exists_{\xi\in\C}\ g^T\z=\xi \n
\end{equation}

\begin{lm}\label{lm:sij}
 On real manifold $\Re S\leq 0$ and under condition $\Re S=0$
 \begin{equation}
  \delta_g S=\left(i\frac{\rho}{2}+\jj\right) [\u,\delta g^T \v]=\frac{1}{2}\left(i\frac{\rho}{2}+\j\right) \tr (\u\v^T+\v\u^T)\omega \delta g^T 
 \end{equation}
 where
 \begin{equation}
  \u=s\omega \eta_{N^{can}} \bar{\n},\quad \v=\n
 \end{equation}
 and $s=\langle \n,\n\rangle_{N^{can}}\in\{-1,1\}$
\end{lm}

When reality condition is satisfied then
\begin{equation}
 \delta_{g_i} S_{{ij}}=\tr \B_{ij}\delta g_i^T
\end{equation}
where
\begin{equation}
 \B_{ij}=\frac{1}{2}\left(i\frac{\rho_{ij}}{2}+\j_{ij}\right)(\u_{ij}\v_{ij}^T+\v_{ij}\u_{ij}^T)\omega
\end{equation}
with spinors
\begin{align}
  \u_{ij}=s_{ij}\omega \eta_{N_i^{can}} \n_{ij},\quad \v_{ij}=\n_{ij}
\end{align}
and $s_{ij}=\langle \n_{ij},\n_{ij}\rangle_{N_i^{can}}$

\subsection{Variations with respect to $\z_{ij}$}

The edge action can be divided into pieces
\begin{equation}
 S_{{ij}}(g_i,\z_{ij})+S_{{ji}}(g_j,\z_{ji})+S^\beta_{ij} (\z_{ij},\z_{ji})
\end{equation}

We parametrize $\delta \z_{ij}$ by $\delta g_{ij}$ traceless as follows
\begin{equation}
 \delta \z_{ij}=\delta g_{ij}^T\ \z_{ij}
\end{equation}
all variations can be written this way but $\delta g_{ij}$ is not unique.

\begin{lm} \label{lm:zg}
$\delta_{\z_{ij}} S_{{ij}}=\delta_{g_{i}} S_{{ij}},\text{ for } \delta g_i=g_i^{-1} \delta g_{ij} g_i$
\end{lm}

\begin{proof}
 Function $S_{{ij}}$ depends only on $g_i^T\z_{ij}$ thus is invariant under variations such that $\delta (g_i^T\z_{ij})=0$.
 Examples of such variations are
 \begin{equation}
  \delta \z_{ij}=\delta g_{ij}^T\ \z_{ij}\text{ and } \delta g_i^T=-g_i^T \delta g_{ij}^T (g_i^T)^{-1}
 \end{equation}
 In this case $\delta_{\z_{ij}} S_{{ij}}+\delta_{g_{i}} S_{{ij}}=0$, so taking variation as in the thesis we get the result.
\end{proof}

We will from now on abuse notation and regard $S_{ji}^\beta=S_{ij}^\beta$ for $i<j$.

\begin{lm}\label{lm:beta}
 Let us introduce a traceless matrix $\B^\beta_{ij}$ such that
 \begin{equation}
 \delta_{\z_{ij}} S^\beta_{ij}=\tr \B^\beta_{ij}\delta g_{ij}^T 
 \end{equation}
 then
 \begin{equation}
  \B^\beta_{ij}=\frac{1}{2}\left(i\frac{\rho_{ij}}{2}+\j_{ij}\right)(\u\v^T+\v\u^T)\omega
 \end{equation}
 where
 \begin{equation}
  \u=\z_{ij},\quad \v=\frac{1}{[\z_{ij},\z_{ji}]}\z_{ji},\quad [\u,\v]=1
 \end{equation}
 Spinors $\z_{ij}$ and $\z_{ji}$ are determined up to complex scaling by
 \begin{align}
   \B^\beta_{ij}\ \z_{ji}=\frac{1}{2}\left(i\frac{\rho_{ij}}{2}+\j_{ij}\right)\ \z_{ji}\\
   \B^\beta_{ij}\ \z_{ij}=-\frac{1}{2}\left(i\frac{\rho_{ij}}{2}+\j_{ij}\right)\ \z_{ij}
 \end{align} 
 Moreover $\B^\beta_{ij}=-\B^\beta_{ji}$.
\end{lm}

\begin{proof}
We have
\begin{equation}
 \delta_{\z_{ij}} S^\beta_{ij}=-\left(i\frac{\rho_{ij}}{2}+\j_{ij}\right)\frac{[\delta g_{ij}^T\z_{ij},\z_{ji}]}{[\z_{ij},\z_{ji}]}=
 \left(i\frac{\rho_{ij}}{2}+\j_{ij}\right)\frac{[\z_{ij},\delta g_{ij}^T\z_{ji}]}{[\z_{ij},\z_{ji}]}
\end{equation}
and it can be written as
\begin{equation}
 \left(i\frac{\rho_{ij}}{2}+\j_{ij}\right)\tr \z_{ji}\frac{1}{[\z_{ij},\z_{ji}]}\z_{ij}^T\omega \delta g_{ij}^T
\end{equation}
Together with \eqref{eq:re} it can be transformed into form from the lemma. Characterization of $\z_{ij}$ and $\z_{ji}$ follows from lemma \ref{lm:sigma}.
\end{proof}

\subsection{Critical point conditions and boundary data}

The boundary data (spinors $\n_{ij}$, normal vectors  $N_i^{can}$ and types of embedded representations) can be summarized by
\begin{itemize}
 \item $N_i^{can}\in\{e_0,e_3\}$ normal vectors,
 \item $t_{i}^{can}=N_i^{can}\cdot N_i^{can}\in\{-1,1\}$,
 \item $\n_{ij}$ spinors,
 \item $s_{ij}=\langle \n_{ij},\n_{ij}\rangle_{N_i^{can}}\in\{-1,1\}$,
\end{itemize}
A critical point for a given boundary data consists of a collection (real manifold)
\begin{equation}
 g_i, \z_{ij}
\end{equation}
satisfying \eqref{end-middle-hol}
\begin{align}
 \forall_{ij} \delta_{\z_{ij}}S_{{ij}}+\delta_{\z_{ij}} S^\beta_{ij}=0\label{end-middle-hol2}\\
 \forall_i \sum_j \delta_{g_i} S_{{ij}}=0
\end{align}
and reality conditions
\begin{equation}
 g_i^T\z_{ij}=\xi_{ij}\n_{ij},\quad \xi_{ij}\in\C
\end{equation}
Condition \eqref{end-middle-hol2} is equivalent to ($i\not= j$)
\begin{equation}
 \delta_{g_i} S_{{ij}}=-\delta_{\z_{ij}} S^\beta_{ij},\quad \delta \z_{ij}=\delta g_{ij}^T\z_{ij},\ \delta g_{i}=g_i^{-1}\delta g_{ij} g_i
\end{equation}
There are gauge transformations acting on the critical points and they are parametrized by
$g\in \SL$ and $\lambda_{ij}\in \C^*$
\begin{equation}
 g_i'=gg_i,\quad \z_{ij}'=\lambda_{ij} (g^T)^{-1} \z_{ij}
\end{equation}
Gauge fixing condition $g_5=1$ fixes $\SL$ gauge transformations.

\subsubsection{$\SL$ solutions}

We can now translate critical point conditions into language of traceless matrices

\begin{df}
$\SL$ solution for the boundary data consists of 
\begin{equation}
\{ g_i\in \SL\}
\end{equation}
determining $\B_{ij},\ \B_{ij}^\beta$ traceless matrices by conditions
\begin{enumerate}
 \item $\B_{ij}=\frac{1}{2}\left(i\frac{\rho_{ij}}{2}+\j_{ij}\right)(\u_{ij}\v_{ij}^T+\v_{ij}\u_{ij}^T)\omega$
where 
\begin{align}
 \u_{ij}=s_{ij}\omega\eta_{N_i^{can}} \bar{\n}_{ij},\quad \v_{ij}=\n_{ij},\quad s_{ij}=\langle \n_{ij},\n_{ij}\rangle_{N_i^{can}}\in\{-1,1\}
\end{align}
\item $\B^\beta_{ij}=-(g_i^T)^{-1} \B_{ij} g_i^T$
\item $\B^\beta_{ij}=-\B^\beta_{ji}$
\item $\sum_j \B_{ij}=0$
\end{enumerate}
The gauge transformations are parametrized by $g\in \SL$
\begin{equation}
 g_i=gg_i,\quad \B^\beta_{ij}=(g^T)^{-1}\B^\beta_{ij}g^T
\end{equation}
\end{df}

\begin{lm}
 $\SL$ solutions are in bijective correspondence with critical points up to $\CP^1$ gauge transformations. An $\SL$ solution is determined by group elements of the critical points and then
 \begin{equation}
  \tr \B_{ij}\delta g_i^T=\delta_{g_i} S_{{ij}},\quad \tr \B_{ij}^\beta\delta g_{ij}^T=\delta_{g_{ij}} S_{ij}^\beta,
 \end{equation}
The $\SL$ gauge transformations are acting the same way on both sides of the correspondence.
\end{lm}

\begin{proof}
Critical point determines $\SL$ geometric solution determining traceless matrices from $\delta_{\z_{ij}} S^\beta_{ij}$ and $\delta_{g_{i}} S_{{ij}}$. They satisfies all assumptions of $\SL$ solution.
 
 In the opposite direction we need to find $\z_{ij}$ such that traceless matrices $\B_{ij}$ and $\B^\beta_{ij}$ are such that
 \begin{equation}
  \tr \B_{ij}\delta g_i^T=\delta_{g_i} S_{{ij}},\quad \tr \B_{ij}^\beta\delta g_{ij}^T=\delta_{g_{ij}} S_{ij}^\beta,
 \end{equation}
 This can be determined using lemma \ref{lm:sigma} and \ref{lm:beta}. In fact $\z_{ij}$ and $\z_{ji}$ need to be eigenvectors of $\B^\beta_{ij}$ (or equivalently  $\B^\beta_{ji}$) 
  \begin{align}
   \B^\beta_{ij}\ \z_{ji}=\frac{1}{2}\left(i\frac{\rho_{ij}}{2}+\j_{ij}\right)\ \z_{ji}\\
   \B^\beta_{ij}\ \z_{ij}=-\frac{1}{2}\left(i\frac{\rho_{ij}}{2}+\j_{ij}\right)\ \z_{ij}
 \end{align}
 Matrix  $\B^\beta_{ij}$ has exactly such eigenvalues because $\B^\beta_{ij}=-(g_i^T)^{-1} \B_{ij} g_i^T$ and $\B_{ij}$ has such eigenvalues. This determines $\z_{ij}$ up to $\CP^1$ gauge transformations.
 
 From equality $\B^\beta_{ij}=-(g_i^T)^{-1} \B_{ij} g_i^T$ we get
 \begin{align}
 \frac{1}{2}\left(i\frac{\rho_{ij}}{2}+\j_{ij}\right) \frac{1}{[\z_{ij},\z_{ji}]}(\z_{ij}\z_{ji}^T+\z_{ji}\z_{ij}^T)\omega=\nonumber\\
  -\frac{1}{2}\left(i\frac{\rho_{ij}}{2}+\j_{ij}\right) (g_i^T)^{-1}(\u_{ij}\v_{ij}^T+\v_{ij}\u_{ij}^T)\underbrace{\omega g_i^T}_{=g_i^{-1}\omega}
 \end{align}
thus 
\begin{equation}\label{eq:comp}
\left(\z_{ij}\frac{1}{[\z_{ij},\z_{ji}]}\z_{ji}^T+\frac{1}{[\z_{ij},\z_{ji}]}\z_{ji}\z_{ij}^T\right)\omega=
  -(\tilde{\u}_{ij}\tilde{\v}_{ij}^T+\tilde{\v}_{ij}\tilde{\u}_{ij}^T)\omega
\end{equation}
where
\begin{equation}
 \tilde{\v}_{ij}=(g_i^T)^{-1} \v_{ij},\quad \tilde{\u}_{ij}=(g_i^T)^{-1} \u_{ij},
\end{equation}
By lemma \ref{lm:comp}
\begin{equation}
 \z_{ij}=\xi_{ij}(g_i^T)^{-1} \v_{ij}=\xi_{ij}(g_i^T)^{-1} \n_{ij},\quad \xi_{ij}\in\C
\end{equation}
and this is just reality condition. We see now that with our choice of $\z_{ij}$
\begin{equation}
 \tr \B_{ij}^\beta\delta g_{ij}^T=\delta_{\z_{ij}} S_{ij}^\beta,\quad \delta \z_{ij}=\delta g_{ij}^T\ \z_{ij}
\end{equation}
and from reality condition (by lemmas \ref{lm:sij})
\begin{equation}
 \tr \B_{ij}\delta g_i^T=\delta_{g_i} S_{{ij}}
\end{equation}
The remaining conditions of critical point can be  written in terms of matrices $\B_{ij}$ and $\B^\beta_{ij}$ and are the same as conditions for $\SL$ solution.

The $\z_{ij}$ are determined up to $\CP^1$ gauge transformations.
\end{proof}

\subsection{Determination of the phase}

Let us compute the value of the edge part of the action in the critical point
We can choose a gauge for $\z_{ij}$ such that
\begin{equation}
 \forall_{ij} g_i^T \z_{ij}=\n_{ij}
\end{equation}
In such a case
\begin{equation}
 S_{{ij}}=0
\end{equation}
The only contribution to the action comes from the $\beta$ amplitude
\begin{equation}
 e^{S^\beta_{ij}}=[\z_{ij},\z_{ji}]^{-\left(i\frac{\rho_{ij}}{2}+\j_{ij}\right)}
 \widebar{[\z_{ij},\z_{ji}]}^{-\left(i\frac{\rho_{ij}}{2}-\j_{ij}\right)}
\end{equation}
From equality of traceless matrices from lemma \ref{lm:comp} and \eqref{eq:comp}  (taking into account gauge fixing)
\begin{equation}
 \z_{ij}=(g_i^T)^{-1} \v_{ij}\text{ and }
 \z_{ji}=[\z_{ij},\z_{ji}](g_i^T)^{-1} \u_{ij}
\end{equation}
but similarly
\begin{equation}
 \z_{ji}=(g_j^T)^{-1} \v_{ji}\text{ and }
 \z_{ij}=[\z_{ji},\z_{ij}](g_j^T)^{-1} \u_{ji}
\end{equation}
This gives equalities (where we introduced $\Lambda_{ij}=[\z_{ij},\z_{ji}]$)
\begin{align}
 (g_j^T)^{-1} \v_{ji}=\Lambda_{ij}(g_i^T)^{-1} \u_{ij}\ \Longrightarrow\  \v_{ji}=\Lambda_{ij}(g_j^T)(g_i^T)^{-1} \u_{ij}\\
 (g_i^T)^{-1} \v_{ij}=-\Lambda_{ij}(g_j^T)^{-1} \u_{ji}\ \Longrightarrow\  \u_{ji}=-\Lambda_{ij}^{-1}(g_j^T)(g_i^T)^{-1} \v_{ij}\\
 \Lambda_{ij}=[(g_i^T)^{-1} \v_{ij},(g_j^T)^{-1} \v_{ji}]=[(g_j^T)(g_i^T)^{-1}\v_{ij}, \v_{ji}]
\end{align}
The phase amplitude is then
\begin{equation}
 e^{-i(\rho_{ij}\ln |\Lambda_{ij}|+2\j_{ij}\Arg\Lambda_{ij})}
\end{equation}

\subsubsection{Difference of the phase amplitude between two critical points}

Changing phases of coherent states changes also the total phase. However difference of the phases of two different critical points is invariant under such transformation. Our goal is to determine this phase for two critical points
\begin{equation}
 (g_i,\z_{ij}) \text{ and } (\tilde{g}_i,\tilde{\z}_{ij})
\end{equation}
As $\u_{ij}$ and $\v_{ij}$ are determined by boundary data they are the same in both critical points thus
\begin{align}
 \v_{ji}=\Lambda_{ij}(g_j^T)(g_i^T)^{-1} \u_{ij},\quad &
 \u_{ji}=-\Lambda_{ij}^{-1}(g_j^T)(g_i^T)^{-1} \v_{ij}\\
 \v_{ji}=\tilde{\Lambda}_{ij}(\tilde{g}_j^T)(\tilde{g}_i^T)^{-1} \u_{ij},\quad &
 \u_{ji}=-\tilde{\Lambda}_{ij}^{-1}(\tilde{g}_j^T)(\tilde{g}_i^T)^{-1} \v_{ij}
\end{align}
So we have ($\delta_{ij}=\Lambda_{ij}\tilde{\Lambda}_{ij}^{-1}$)
\begin{align}
(\tilde{g}_j^T)(\tilde{g}_i^T)^{-1}\left((g_j^T)(g_i^T)^{-1}\right)^{-1} \v_{ji}&=\delta_{ij}  \v_{ji}\\
(\tilde{g}_j^T)(\tilde{g}_i^T)^{-1}\left((g_j^T)(g_i^T)^{-1}\right)^{-1} \u_{ji}&= \delta_{ij}^{-1}\u_{ji}
\end{align}
Using $\z_{ji}=(g_j^T)^{-1} \v_{ji}$ and $\z_{ij}\approx(g_j^T)^{-1} \u_{ji}$ we can write
\begin{align}
\Delta_{ij}\ \z_{ji}&=\delta_{ij}  \z_{ji}\\
\Delta_{ij}\ \z_{ij}&= \delta_{ij}^{-1}\z_{ij} 
\end{align}
where we introduce 
\begin{equation}\label{eq:Delta}
 \Delta_{ij}=(g_j^T)^{-1}(\tilde{g}_j^T)(\tilde{g}_i^T)^{-1}(g_i^T)
\end{equation}
We also have
\begin{equation}
 e^{\tilde{S}_{ij}^\beta-S_{ij}^\beta}=
 e^{i(\rho_{ij}\ln |\delta_{ij}|+2\j_{ij}\Arg\delta_{ij})}
\end{equation}
We have for $[\u,\v]=1$
\begin{align}
 &e^{\ln|\delta| (-i)i(\u\v^T+\v\u^T)\omega+\Arg\delta i(\u\v^T+\v\u^T)\omega}\u= \delta^{-1} \u\\
 &e^{\ln|\delta| (-i)i(\u\v^T+\v\u^T)\omega+\Arg\delta i(\u\v^T+\v\u^T)\omega}\v= \delta \v
\end{align}
Using fact that
\begin{equation}
 \frac{2\gamma}{\gamma-i}\B^\beta_{ij}=\frac{\rho_{ij}}{2}\ i\frac{1}{[\z_{ij},\z_{ji}]}(\z_{ij}\z_{ji}^T+\z_{ji}\z_{ij}^T)\omega
\end{equation}
and the properties of $\Delta_{ij}$ we see that
\begin{equation}\label{eq:diffrence-phases}
 \Delta_{ij}=e^{\Arg\delta_{ij}\frac{2}{\rho_{ij}}\tilde{\M}_{ij}-i\ln|\delta_{ij}|\frac{2}{\rho_{ij}}\tilde{\M}_{ij}}
\end{equation}
where we introduced $\tilde{\M}_{ij}=\frac{2\gamma}{\gamma-i}\B^\beta_{ij}$.

\begin{lm}\label{lm:difference-phases}
 The phase difference between critical points defined by two $\SL$ solutions $\{g_i\}$ and $\{\tilde{g}_i\}$ is equal
 \begin{equation}
\Delta S= \tilde{S}-S= i\sum_{i<j} \rho_{ij}r_{ij}+2\j_{ij}\phi_{ij}\text{ mod } 2\pi i
 \end{equation}
where $r_{ij}$ and $\phi_{ij}$ are real numbers determined by
\begin{equation}
 g_i \tilde{g}_i^{-1}\tilde{g}_j g_j^{-1}=e^{\phi_{ij}\M_{ij}-ir_{ij}\M_{ij}}
\end{equation}
where $\M_{ij}=\frac{2}{\rho_{ij}} \frac{2\gamma}{\gamma-i}(\B^\beta_{ij})^T$.
\end{lm}

\begin{proof}
 It is enough to transpose \eqref{eq:diffrence-phases} and \eqref{eq:Delta}.
\end{proof}

\section{The geometric solutions}\label{sec:stationary}

We will translate our $\SL$ description into $\SO(1,3)$ language. We will also extract necessary geometric boundary data from the boundary data given by spinors $\n_{ij}$.
This section is devoted to the relation between the two descriptions.

\subsection{Spin structures}

Let us notice that EPRL construction makes sense only if certain integrability condition holds
\begin{equation}
 \forall_i \sum_{j\not=i} \jj_{ij} \text{ is integer,} 
\end{equation}
because only then the intertwiners space (defined in the distributional sense) is nonempty. We assume this condition for every boundary tetrahedron. The vertex integral and the action (defined up to $2\pi i$) are then invariant under the following transformations
\begin{equation}\label{eq:spin-gauge}
 g_i\rightarrow s_ig_i,\quad s_i=\pm 1
\end{equation}
The transformations of this kind will be called spin structure transformations.

\subsection{Bivectors and traceless matrices}
\label{sec:bivectors}

Generators of $\SO_+(1,3)$ are  matrices in $M^4$ that are antisymmetric after lowering an index. They can be identified with bivectors. Generators of $\SL$ are traceless matrices. The isomorphism between these Lie algebras is given on simple bivectors by
\begin{equation}
\pi\colon so(1,3)\ni B=a\wedge b\longrightarrow \B=\frac{1}{4}(\eta_a\hat{\eta}_b-\eta_b\hat{\eta}_a)\in sl(2,\C).
\end{equation}
With the standard choice of the orientation, Hodge ${\Hodge}$ operation corresponds to multiplication by $i$ of the traceless matrix.
We also have for two bivectors $B_1$ and $B_2$
\begin{equation}
 {B_1}\cdot B_2=-2\Re \tr \B_1\B_2.
\end{equation}
The bivector is spacelike if $B\cdot B>0$ and mixed if $B\cdot B<0$ and ${\Hodge}B\cdot {\Hodge}B=-B\cdot B$. The image of traceless matrix with purely imaginary eigenvalues is spacelike. 

We assume that the $\SL$ representations satisfies (spacelike faces conditions)
\begin{equation}
 \rho_{ij}=2\gamma \j_{ij}.
\end{equation}
Then traceless matrices $\B_{ij}$ have the property that
\begin{equation}
 \frac{2\gamma}{\gamma-i}\B_{ij}=i\frac{1}{2}\rho_{ij}(\u_{ij}\v_{ij}^T+\v_{ij}\u_{ij}^T)\omega
\end{equation}

\subsubsection{Characterization of the bivectors} \label{sec:3d-character}

We can characterize $\pi^{-1}\left( -\frac{2\gamma}{\gamma-i}\B_{ij}^T\right)$ (see lemma \ref{lm:char} in the Appendix)

\begin{lm}
We have
\begin{equation}
 \pi(B_{ij}')=-\frac{2\gamma}{\gamma-i}\B_{ij}^T
\end{equation}
where $B_{ij}'= {\Hodge}({v}_{ij}\wedge N^{can}_i)={\Hodge}(l_{ij}\wedge N^{can}_i)$ and the null vector $l_{ij}$ is given by
\begin{equation}
 l_{ij}^\mu=s_{ij}\rho_{ij}\n_{ij}^T {\sigma}^\mu \bar{\n}_{ij}
\end{equation}
It is future directed if $s_{ij}=1$ and past directed if $s_{ij}=-1$.
The vector $v_{ij}$ is determined by 
\begin{equation}
 l_{ij}={v}_{ij}+cN_i^{can},\quad {v}_{ij}\perp N_i^{can},\quad c=\frac{l_{ij}\cdot N_i^{can}}{N_i^{can}\cdot N_i^{can}}
\end{equation}
\end{lm}

We can now characterize bivectors in terms of 3 dimensional geometry in the space perpendicular to $N_i^{can}$.

\begin{df}\label{df:vij-vec}
The geometric boundary data are sets of vectors 
\begin{equation}
 v_{ij}\perp N_i^{can}
\end{equation}
with norm $v_{ij}\cdot v_{ij}=-t_i^{can} \rho_{ij}^2$
that are obtained form spinors $\n_{ij}$ as a projection onto space orthogonal to $N_i^{can}$ of the null vectors $l_{ij}$ defined by
\begin{equation}
 l_{ij}^\mu=s_{ij}\rho_{ij}\n_{ij}^T {\sigma}^\mu \bar{\n}_{ij}
\end{equation}
\end{df}

Let us notice what follows from the previous subsection:
\begin{itemize}
 \item For $t_i^{can}=1$ ($s_{ij}=1$) ${v}_{ij}$ is spacelike,
 \item For $t_i^{can}=-1$  and $s_{ij}=1$, ${v}_{ij}$ is timelike future directed,
 \item For $t_i^{can}=-1$  and $s_{ij}=-1$, ${v}_{ij}$ is timelike past directed,
\end{itemize}

\subsection{Geometric solutions}

The inversion $I\in \SO(1,3)$ is defined by
\begin{equation}
 \forall_v\ Iv=-v
\end{equation}
It does not belong to $\SO_+(1,3)$.

Let us introduce notion of geometric solution

\begin{df}
The $\SO(1,3)$ geometric solution is a collection 
\begin{equation}
 \{G_i\in \SO(1,3)\}_{i=0,\ldots 4}
\end{equation}
such that bivectors
\begin{equation}
 B_{ij}={\Hodge} v_{ij}\wedge N_i^{can},\quad B_{ij}^{\G}=G_i({\Hodge} v_{ij}\wedge N_i^{can})\quad i\not=j
\end{equation}
with ${v}_{ij}$ defined by the boundary data (definition \ref{df:vij-vec})
satisfy
\begin{align}
 &\forall_{i\not=j}B_{ij}^{\G}=-B_{ji}^{\G},\\
 &\forall_i \sum_{j\not=i} B_{ij}=0\label{eq:closure}
\end{align}
Two geometric solutions $\{G_i\}$, $\{G_i'\}$ are gauge equivalent if there exists $G\in \SO(1,3)$ and $s_i\in\{0,1\}$ such that
\begin{equation}
 \forall_i\quad G_i'=GG_iI^{s_i}
\end{equation}
Gauge transformations 
\begin{equation}
 \forall_i\quad G_i'=G_iI^{s_i}
\end{equation}
are called inversion gauge transformation.
\end{df}

Let us notice that necessary condition \eqref{eq:closure} for a critical point is in fact a condition for the boundary data. It is called \emph{closure condition} \cite{Frank2} \footnote{This is condition for non-decaying for $\Lambda\rightarrow \infty$ of the invariant in the case of $N^{can}=e_0$. Arguably it is also condition in the case of $N^{can}=e_3$ for non-decaying of the invariant defined in the distributional sense. The issue deserve separate treatment and it will not be address in the present paper.}
\begin{equation}\label{df:closure}
\forall_i\ \sum_{j\not=i} {v}_{ij}=0
\end{equation}
We will always assume that boundary data satisfies the closure condition.

Let us notice that from \ref{sec:3d-character} we know that
\begin{equation}\label{eq:B-spacelike}
 B_{ij}={\Hodge}(v_{ij}\wedge N_i^{can})
\end{equation}
is a spacelike bivector.

\begin{lm}\label{lm:map}
 There exists bijection between $\SL$ solutions up to spin structure transformations \eqref{eq:spin-gauge} for the given boundary data satisfying the closure condition and $\SO(1,3)$ geometric solutions $\{G_i\}$ up to inversion gauge transformations for the corresponding geometric boundary data. The map from $\SL$ solutions is given by
 \begin{equation}
    G_i=\pi(g_i)
 \end{equation}
and then also
 \begin{align}
  B_{ij}&=-\pi\left(\frac{2\gamma}{\gamma-i}\B_{ij}^T\right)\\
  B_{ij}^\G&=\pi\left(\frac{2\gamma}{\gamma-i}(\B^\beta_{ij})^T\right)
 \end{align}
The $\SL$ gauge transformations correspond to $\SO_+(1,3)$ gauge transformations.
 \end{lm}

\begin{proof}
 From any $\SL$ solution we produce in this way $\SO(1,3)$ geometric solution. The map identifies only points that differ by spin structure transformations. Two gauge equivalent $\SL$ solutions are mapped into gauge equivalent geometric solutions. We need only to show that in every class up to gauge inversion transformations there exists an element that is the image of $\SL$ solution.
 
 Let us choose an $\SO(1,3)$ geometric solution. By inversion gauge transformation we can assume that all $G_i\in \SO_+(1,3)$. The preimages of such group elements constitute $\SL$ solution.
\end{proof}

\noindent\emph{ Comment:} Transposition is a price one need to pay for following notation from \cite{Frank2}.

%% file: chapter-Geometric-reconstruction.tex
%!TEX root = Spin_foams_with_time_like_tetrahedra.tex
\section{Geometric reconstruction}\label{sec:geometric}

In the previous section we determined that critical points modulo gauge transformations are in one to one correspondence with $\SO(1,3)$ geometric solutions. In this section we will classify the latter\footnote{The classification was done in \cite{Barrett2000, Barrett-Crane, Barrett2003}, see also \cite{FrankEPRL, FK}, but in a bit different set-up.}.  Geometric $\SO(1,3)$ solutions divide into non-degenerate and degenerate ones. 
With the first class we can associate non-degenerate Lorentzian simplices. However with the pair of degenerate geometric solutions we will (in section \ref{sec:other}) associate simplex in other than Lorentzian signature.
For this reason we want to provide classification in arbitrary signature.

\subsection{Notation}

We will denote all operations in this arbitrary metric by underline i.e. (Hodge star $\Hodgu$, scalar product $\cdou$ and contractions with use of the metric $\lrcorneu$ and $\llcorneu$).

We can introduce reflections with respect to the normalized (to $\pm 1$ vector $N$)
\begin{equation}
 (R_{N})^\mu_\nu={\mathbb I}^\mu_\nu-\frac{2N^\mu\ N_\nu}{N\cdou N}\in \OO=\OO(p,q)
\end{equation}
where we lowered index with use of the metric. Notice that $R_{N}^2={\mathbb I}$. 

We can also introduce inversion 
\begin{equation}
 Iv=-v
\end{equation}
Depending on the signature inversion belongs ($\OO(2,2)$ and $\OO(4)$) or does not ($\OO(1,3)$) to the connected component of identity. It is however always in special orthogonal subgroup in dimension $4$. The reflections do not belong to special orthogonal subgroup thus neither to connected component of identity. We denote special orthogonal subgroup by $\SO$ and its connected component of identity by $\SO_+$.

\subsection{Geometric solution}

We will first define geometric version of the boundary data

\begin{df}
The ( $\SO$ geometric) boundary data is a collection
\begin{equation}
 \{N_i^{can}\in\R^4,\quad v_{ij}\in \R^4\colon \forall_i\ \sum_{j\not=i} v_{ij}=0, v_{ij}\perpu N_i^{can}\}
\end{equation}
where $N_i^{can}$ is a canonical normal, $N_i^{can}\cdou N_i^{can}=\T_i^{can}\in\{-1,1\}$. We will say that the boundary data is non-degenerate if for every $i$, every $3$ out of $4$ vectors $v_{ij}$ are linearly independent.
\end{df}

Set of canonical normals is such that every normalized vector $N$, $N\cdou N=\pm 1$ can be rotated by an element of $\SO$ to exactly one of canonical normals\footnote{For $\SO(1,3)$ it can be chosen as $\{e_0,e_3\}$, for $\SO(4)$ as $\{e_0\}$. For $\SO(2,2)$ we can choose $\{e_3,e_1\}$ if we assume that they have different norms, but in the asymptotic analysis only $e_3$ will play a role.}.

\begin{df}
The geometric $\SO$ solution for the geometric boundary data is a collection 
\begin{equation}
 \{G_i\in \SO\}_{i=0,\ldots 4}
\end{equation}
such that bivectors
\begin{equation}
 B_{ij}={\Hodgu} v_{ij}\wedge N_i^{can},\quad B_{ij}^\G=G_i({\Hodgu} v_{ij}\wedge N_i^{can})\quad i\not=j
\end{equation}
satisfy
\begin{equation}
 B_{ij}^\G=-B_{ji}^\G
\end{equation}
Two solutions $\{G_i\}$, $\{G_i'\}$are gauge equivalent if there exists $G\in \SO$ and $s_i\in\{0,1\}$ such that
\begin{equation}
 \forall_i G_i'=GI^{s_i}G_i.
\end{equation}
\end{df}

We can associate with this data normals $N_i^\G=G_iN_i^{can}$.

\begin{df}
 Geometric solution is non-degenerate if every four out of five $N_i^\G$ span the whole $\R^4$. It is degenerate if this condition is not satisfied.
\end{df}

\subsection{Geometric bivectors}\label{sec:bivectors2}

For convenience of the reader we provide in Appendix \ref{sec:bivectors22} a construction of bivectors related to faces of the $4$ simplex that we will denote by $B_{ij}^\geom$  and call {\it geometric bivectors}. The construction works also for degenerate cases and depends on the order of vertices. 

If the $4$ simplex is nondegenerate we can introduce outer directed normal $N_i^\geom$ (see \ref{sec:bivectors2})and certain numbers $W_i^\geom$ related to volumes of tetrahedra (see Appendix \ref{sec:bivectors22}) such that
\begin{equation}
 \sum_i W_i^{\geom} N_i^{\geom}=0,\quad -t_i W_i^\geom>0.
\end{equation}
In fact,
\begin{itemize}
 \item $B_{ij}^{\geom}=-B_{ji}^{\geom}$,
 \item $B_{ij}^{\geom}=-\frac{1}{Vol^{\geom}} W_i^{\geom}W_j^{\geom} {\Hodgu}(N_j^{\geom}\wedge N_i^{\geom})$,
 \item $B_{ij}^{\geom}\lrcorneu N_i^{\geom}=0$.
\end{itemize}
For a spacelike face its area is equal to 
\begin{equation}
A_{ij}^\geom= \frac{1}{2}|B_{ij}^{\geom}|.
\end{equation}
and $Vol^\geom$ is volume of the simplex form (see Appendix \ref{sec:bivectors22})\footnote{It differs by $4!$ from the volume of the simplex.}.

\subsection{Reconstruction of normals from the bivectors}

Suppose that we have bunch of bivectors $B_{ij}^\G$ that come from some non-degenerate geometric solution. We can reconstruct from them $N_i$ up to a sign.  We assume that $v_{ij}$ for the fixed $j$ span the whole space perpendicular to $N_i^{can}$ (for example the boundary data is non-degenerate). 

\begin{lm}
Let us assume that we have a geometric solution $\{G_i\}$ for the non-degenerate geometric boundary data. Then the following are equivalent for the chosen $i$ and vector $N$
\begin{itemize}
 \item $\forall_{j\not=i} B_{ij}^\G\lrcorneu N=0, \quad N\cdou N=\pm 1$,
 \item $N=\pm N_i^\G$.
\end{itemize}
\end{lm}

\begin{proof}
 We know that
 \begin{equation}
  \forall_{j\not=i} B_{ij}^\G\lrcorneu N_i^\G=0
 \end{equation}
If there is independent vector $N$ satisfying the same equation, then
\begin{equation}
 \exists_{\lambda_{ij}},\quad B_{ij}^\G=\lambda_{ij} {\Hodgu}(N\wedge N_i^\G).
\end{equation}
This contradicts non-degeneracy of the boundary data because every $3$ out of $4$ $B_{ij}^\G=G_i B_{ij}$ should be independent. The vector $N$ needs to be proportional to $N_i^\G$ and from normalization of $N_i^\G$ it follows that $N=\pm N_i^\G$.
\end{proof}

\subsection{Reconstruction of bivectors from the knowledge of $\pm N_i$}

We will now reconstruct bivectors from normals $N_i$. In fact out theorem works in any dimension in arbitrary non-degenerate signature.

\begin{thm}\label{thm:deterB}
Let us assume that in $\R^n$ we have normalized to $\pm 1$ vectors $N_i$, $i=0\ldots n$, such that
(nondegeneracy) any $n$ out of $n+1$ vectors $N_i$ span the whole $\R^n$. Then there exist $n-2$ vectors
 \begin{equation}
 B_{ij}' \ i\not=j, i=0,\ldots, n
 \end{equation}
 such that
\begin{itemize}
\item for every $i$
 \begin{equation}\label{eq:simplicity_and_closure}
  B_{ij}'\lrcorneu N_i=0,\quad \sum_{j\not=i} B_{ij}'=0,
 \end{equation}
 \item for all $i\not=j$, 
 \begin{equation}\label{eq:orientation}
 B_{ij}'=-B_{ji}'.
 \end{equation}
\end{itemize}
A solution to equations \eqref{eq:simplicity_and_closure} and \eqref{eq:orientation} is given by:
\begin{equation}
 B_{ij}'=W_iW_j\ {\Hodgu}N_j\wedge N_i,
\end{equation}
where constants $W_i\in \R$ are nonzero solutions of
\begin{equation}
 \sum_i W_iN_i=0.
\end{equation}
For any other solution $B_{ij}''$ there exists a constant $\lambda\in \R$  such that
\begin{equation}
 B_{ij}''=\lambda B_{ij}'
\end{equation}
The solution is independent of changing some of $N_i$ by sign.
\end{thm}

\begin{proof}
Let us first prove uniqueness of such solution up to scaling. Let us assume such $B_{ij}'$ are given.
There is exactly one solution (up to a scaling by real constant) of the equation
\begin{equation}
 \sum_i W_iN_i=0,\quad W_i\in\R
\end{equation}
as there are $n+1$ vectors in $n$ dimensional space and every $n$ out of $n+1$ are independent. The constant $W_i$ are all nonzero (in the nontrivial solution) and in fact they will turn out to be proportional to signed volumes of the tetrahedra of the $4$ simplex with normal $N_i$ \footnote{This is called Minkowski theorem \cite{Minkowski} see also \cite{Minkowski2, Bianchi2010}, but in arbitrary signature and for simplices. We do not know any reference for general Minkowski theorem in other than euclidean signature.}.

Let us notice that $B_{ij}'$ are simple $n-2$ bivectors because they are annihilated by two independent normals
\begin{equation}
 B_{ij}'\lrcorneu N_i=0\quad B_{ij}'\lrcorneu N_j=-B_{ji}'\lrcorneu N_j=0
\end{equation}
so there need to exists constant $\lambda_{ij}$ such that
\begin{equation}
 B_{ij}'=\lambda_{ij} {\Hodgu}N_j\wedge N_i,\quad \lambda_{ij}\not=0
\end{equation}
We have for every $i$
\begin{equation}
 0=\sum_{j\not=i} B_{ij}'={\Hodgu}(\sum_{j\not=i}  \lambda_{ij} N_j)\wedge N_i
\end{equation}
so we have
\begin{equation}
\sum_{j\not=i} \lambda_{ij} N_j=-\lambda_i N_i,
\end{equation}
for some $\lambda_i\in\R$. This equation has a unique up to a constant solution so
\begin{equation}
 \lambda_{ij}=\frac{W_j}{W_i}\lambda_i
\end{equation}
From the symmetry 
\begin{equation}
 \lambda_{ij} {\Hodgu}N_j\wedge N_i=B_{ij}'=-B_{ji}'=-\lambda_{ji} {\Hodgu}N_i\wedge N_j=\lambda_{ji}{\Hodgu}N_j\wedge N_i
\end{equation}
so we have $\lambda_{ij}=\lambda_{ji}$ and
\begin{equation}
 \frac{W_j}{W_i}\lambda_i=\frac{W_i}{W_j}\lambda_j\Rightarrow \lambda_i=\lambda W_i^2
\end{equation}
for some constant $\lambda$ and finally
\begin{equation}
 B_{ij}'=\lambda W_iW_j {\Hodgu} N_j\wedge N_i
\end{equation}
This shows uniqueness up to a scaling. To show existence it is enough to check that such constructed forms satisfy requirements (that is just reversing all arguments).

As changing $N_i$ by sign also change $W_i$ by the same sign, the $B_{ij}'$ are independent of the choice of sign of normal vectors.
\end{proof}

Bivectors $B_{ij}^\G$ satisfies requirements for normals $N_i^{\G}$.

\subsection{Nondegenerate bivectors and $4$-simplex}

We can now prove that non-degenerate geometric solution determines $4$-simplex uniquely up to shifts and inversion.

\begin{thm}\label{thm:reconstr}
Given a non-degenerate geometric $\SO$ solution $\{G_i\}$ there exist exactly two  $4$-simplices (defined up to shift $\R^4$ transformations) such that
 \begin{equation}
  B_{ij}^{\geom}=r B_{ij}^\G
 \end{equation}
 where $r=\pm 1$. They are related by inversion transformation $I$ and for both the sign $r$ is the same.
\end{thm}

\begin{proof}
 We will first prove that there exists such $4$ simplex. 
 We take any $5$ planes orthogonal to $N_i^{\G}$. They cut out a $4$ simplex $\Delta'$. This simplex is unique up to shifts and scaling by a real number (changing the size and applying the inverse).
 
 The bivectors of any of these $4$-simplices $B_{ij}^{\geom'}$ satisfy from reconstruction from normals (see Theorem \ref{thm:deterB}):
 \begin{equation}
  B_{ij}^\G=c B_{ij}^{\geom'},\quad c\in\R^*
 \end{equation}
 Under scaling transformation (by real number $\lambda$, so also under inverse) the bivector changes by $\lambda^2$. There exist exactly two scalings $\pm \sqrt{|c|}$ that brings the bivectors to
 \begin{equation}
  B_{ij}^\G=r B_{ij}^{\geom}\quad r=\pm 1
 \end{equation}
The sign cannot be changed but it depends on the choice of orientation.

Uniqueness: From bivectors we can reconstruct $\pm N_i^{\G}$ (sign ambiguity). For any choice of the signs we reobtain the same $4$ simplices. 
\end{proof}

The sign $r$ seems to be an additional data in the reconstruction. 

\begin{df}\label{df:r}
We call the constant $r=\pm 1$ from the reconstruction for geometric solution, geometric Pleba{\'n}ski orientation.
\end{df}

Constant $r$ relates chosen orientation of $\R^4$ with the orientation defined by order of tetrahedra.

We have
\begin{equation}\label{eq:BGgeom}
 B_{ij}^\G=-\frac{1}{Vol^{\geom}}rW_i^{\geom}W_j^{\geom} {\Hodgu}(N_j^{\geom}\wedge N_i^{\geom})
\end{equation}
where $Vol^{\geom}$ is $4!$ volume of the $4$ simplex (see \ref{sec:bivectors2}).

\subsection{Uniqueness of Gram matrix and reconstruction}\label{sec:Gram}

For the non-degenerate geometric solution $\{G_i\}$ the edge lengths of the tetrahedron $i$ in the $4$ simplex $\Delta$ can be reconstructed from bivectors $B_{ij}^{\geom}=rB_{ij}^\G$, $j\not=i$. Let us now consider single $i$-th tetrahedron. As $G_i$ is a rotation the shape of the tetrahedron with bivectors $B_{ij}^\G=G_i{\Hodgu}(v_{ij}\wedge N_i^{can})$
is the same as the shape of the tetrahedron with bivectors 
\begin{equation}
 B_{ij}=\Hodgu (v_{ij}\wedge N_i^{can})\ ,
\end{equation}
which are determined by the geometric boundary data. 

\begin{lm}
If the geometric boundary data is non-degenerate then for every $i$ there exists a unique up to inversion and translations tetrahedron with face bivectors
\begin{equation}\label{eq:tetrahedron-bivectors}
 B_{ij}=r_i\Hodgu (v_{ij}\wedge N_i^{can})\ .
\end{equation}
in the subspace $N_i^{can\perp}$ with $r_i=\pm 1$.
\end{lm}

\begin{proof}
Let us fix $i$. We cut a tetrahedron with planes perpendicular to $v_{ij}$ in $N_i^{can\perp}$ in generic position. Its bivectors $B_{ij}'$ are proportional to $\Hodgu (v_{ij}\wedge N_i^{can})$ thus
\begin{equation}
 B_{ij}'=\lambda_{ij}'\Hodgu (v_{ij}\wedge N_i^{can})
\end{equation}
From the closure condition we know
\begin{equation}
 \sum_{j\not=i} B_{ij}'=\sum_{j\not=i} \lambda_{ij}'\Hodgu (v_{ij}\wedge N_i^{can})=
 \Hodgu \left(\sum_{j\not=i} \lambda_{ij}'v_{ij}\right)\wedge N_i^{can}
\end{equation}
As $v_{ij}$ are perpendicular to $N_i^{can}$ from non-degeneracy
\begin{equation}
 \sum_{j\not=i} \lambda_{ij}'v_{ij}=0\Rightarrow \exists_{\lambda}\colon \lambda_{ij}'=\lambda
\end{equation}
By rescaling of the tetrahedron we can get $\lambda=\pm 1$.
\end{proof}

This determines edge lengths uniquely as functions of $v_{ij}$. Let us denote the signed square lengths of the edge 
\begin{equation}
 l^{i\ 2}_{jk}\text{ between faces } (ij) \text{ and } (ik) \text{ of the tetrahedron } i.
\end{equation}
This numbers are defined for $i,j,k$ pairwise different and are symmetric in $j,k$.

\begin{df} \label{df:lengths}
The geometric boundary data satisfies lengths matching condition if $l^{k\ 2}_{ij}$ is symmetric in all its indices. 
\end{df}

The lengths matching condition is necessary for existence of non-degenerate geometric solution.

If lengths matching condition is satisfied we define \emph{signed square lengths}
\begin{equation}
 l_{ml}^2=l^{k\ 2}_{ij},\quad \text{ for } m,l \text{ the remaining missing indices different from } i,j,k
\end{equation}
These lengths determines $4$ simplex unique up to orthogonal transformation and shifts (see \cite{Dittrich2008, Dittrich2010} for description of the matching conditions).

\begin{thm}\label{Gram-matrxi}
For the signed square lengths $l_{ij}^2$ from non-degenerate geometric boundary data satisfying lengths matching condition we introduce lengths Gram matrix
of the $4$ simplex
\begin{equation}\label{eq:Gram}
 G^l=\left(\begin{array}{ccccc}
  0 & 1&1&\cdots& 1\\
  1 & 0 &l_{01}^2 &\cdots & l_{04}^2\\
  1 & l_{10}^2 & 0 &\cdots &l_{14}^2\\
  \vdots &\vdots &\vdots &\ddots &\vdots\\
  1 & l_{40}^2 & l_{41}^2 &\cdots &0
 \end{array}\right)
\end{equation}
Let us denote the signature of $G^l$ by $(\underbrace{p+1}_{+},\underbrace{q+1}_{-},\underbrace{n}_{0})$ 
\begin{itemize}
 \item If $n=0$ then there exists a unique up to $\OO(p,q)\rtimes \R^4$ transformations non-degenerate $4$ simplex in the spacetime with signature $(p,q)$ with these lengths. There are two inequivalent $4$-simplices up to $\SO(p,q)\rtimes \R^4$ transformations
 \item If $n>0$ then there exists a unique up to $\OO(p,q)\rtimes \R^{p+q}$ transformations degenerate $4$ simplex in the signature $(p,q)$ with these lengths.
\end{itemize}
\end{thm}

\begin{proof}
It is more convenient to work with the matrix with elements
\begin{equation}
 M_{ij}=\frac{1}{2}(l_{i0}^2+l_{j0}^2-l_{ij}^2),\quad i,j\in\{1,\ldots 4\}
\end{equation}
that should correspond to the matrix of suitable scalar products.
By a change of basis we can transform $G^l$ to the block diagonal form where one block is $\left(\begin{array}{cc} 0  & 1\\1&0\end{array}\right)$ and the second block is $M$. The signature of $M$ is thus $(p,q,n)$. There exist matrix $R$ with $4$ columns and $p+q$ rows with full range such that
\begin{equation}
 M=R^T\eta R,\quad \eta=\text{diag}(\underbrace{1,\cdots,1}_p,\underbrace{-1,\cdots,-1}_q)\label{eq:N-M}
\end{equation}
because $M$ is symmetric with given signature. The $4$ simplex can be constructed as follows. Choose arbitrary $x_0$ and define $x_i$ for $i\not=0$ by
\begin{equation}
 R_i^\mu=x_i^\mu-x_0^\mu
\end{equation}
This $4$ simplex has a prescribed lengths. Let us now compare two different $4$ simplices with vertices $x'_i$ and $x_i$. Both $R'$ and $R$ need to satisfy \eqref{eq:N-M}. Since $R$ and $R'$ have the same kernel (kernel of $M$) and full range there exists $G\in \GL(p+q)$ such that
\begin{equation}
 R'=GR
\end{equation}
and it satisfies
\begin{equation}
 G^T\eta G=\eta\ \Longrightarrow G\in \OO(p,q)
\end{equation}
We have thus
\begin{equation}
 x_i'=Gx_i+v,\text{ where } v=Ox_0'-x_0
\end{equation}
as $v$ and $G\in \OO(p,q)$ are arbitrary we obtained uniqueness of the solution up to desired transformation.
\end{proof}

Suppose that $p'\geq p$ and $q'\geq q$ then we can affinely embed spacetime of signature $(p,q)$ into spacetime of signature $(p',q')$. Even if $n>0$ we can reconstruct degenerate $4$-simplex in $4$ dimensions up to $\OO(p',q')\rtimes \R^4$ transformations. However signature $(p',q')$ is not unique.

If we introduce normals to tetrahedra of two such $4$ simplices $\{N_i\}$ and $\{N_i'\}$ then there exists $G\in \OO$ and $s_i\in\{0,1\}$ such that
\begin{equation}
 N_i'=(-1)^{s_i}GN_i=GI^{s_i}N_i
\end{equation}
that are exactly gauge transformations of the $\SO$ geometric solution.

\subsection{Geometric rotations $G_i^{\geom}$}

Suppose that we have a non-degenerate boundary data and the Gram matrix \eqref{eq:Gram} is non-degenerate.
From the lengths Gram matrix we can reconstruct geometric non-degenerate $4$ simplex (up to $\OO\ltimes \R^4$ transformations). 
Let us choose one of the simplices and compute geometric bivectors $B_{ij}^{\geom}$ and normals $N_i^{\geom}$.

For normals $N_i^{\geom}$ we will introduce vectors $j\not=i$
\begin{equation}
 v_{ij}^{\geom}=-\frac{1}{Vol^{\geom}}\left(W_i^{\geom}W_j^{\geom} N_j^{\geom}-\frac{W_i^{\geom}W_j^{\geom} \ N_j^{\geom}\cdou N_i^{\geom}}{ N_i^{\geom}\cdou N_i^{\geom}}N_i^{\geom}\right),\quad v_{ij}^{\geom}\perpu N_i^{\geom}
\end{equation}
that are normals to the faces of $i$-th tetrahedron recovered from geometric bivectors. We have
\begin{equation}
 B^{\geom}_{ij}={\Hodgu}(v_{ij}^{\geom}\wedge N_i^{\geom})
\end{equation}

\begin{lm}
 If the lengths matching condition is satisfied then
\begin{equation}
 v_{ij}^{\geom}\cdou{} v_{ik}^{\geom}=v_{ij}\cdou{} v_{ik}
\end{equation} 
\end{lm}

\begin{proof}
This is equivalent to
\begin{equation}
 B^{\geom}_{ij}\cdou B^{\geom}_{ik}=B_{ij}\cdou B_{ik}
\end{equation}
Both bivectors are bivectors of either reconstructed $i$th tetrahedron in $4$ simplex or reconstructed boundary tetrahedron in the space perpendicular to $N_i^{can}$. Similarly as for $4$ simplex we can prove that the two tetrahedra differs by rotation $G\in \OO$. The scalar products are thus preserved\footnote{These scalar product can be computed by a version of Cayley-Menger determinant \cite{Berger}.}.
\end{proof}

We can introduce group elements $G_i^{\geom}$ for any $i$ by conditions
\begin{equation}
 G_i^{\geom}N_i^{can}=N_i^{\geom},\quad  \forall_{j\not=i}\ G_i^{\geom}v_{ij}=v_{ij}^{\geom}
\end{equation}
There are $5$ conditions but only $4$ are independent (closure conditions are the same for $v^{\geom}$ and $v$ vectors).

\begin{lm}
 Elements $G_i^{\geom}\in \OO$.
\end{lm}

\begin{proof}
 Both $N_i^{\geom}$ and $N_i^{can}$ are of the same type normalized, and perpendicular to the rest of the vectors. The shapes of tetrahedra are also the same so scalar product between vectors are preserved. This is however the definition of being orthogonal.
\end{proof}

\subsection{Relation of $G_i$ to $G_i^{\geom}$}

We would like to compare $G_i$ from the definition of geometric solution with $G_i^{\geom}$ obtained from $B_{ij}^{\geom}=(-1)^sB_{ij}^\G$ (where $r=(-1)^s$ for $s\in\{0,1\}$). We know that there exist $s_i\in\{0,1\}$ such that
\begin{equation}
 N_i^\G=(-1)^{s_i} N_i^{\geom},
\end{equation}
and 
\begin{align*}
 &{\Hodgu}(G_iv_{ij}\wedge N_i^\G)=B_{ij}^\G=(-1)^s B_{ij}^{\geom}=\\
 &={\Hodgu}((-1)^{s} v_{ij}^{\geom}\wedge N_i^{\geom})={\Hodgu}((-1)^{s+s_i} v_{ij}^{\geom}\wedge N_i^\G)
\end{align*}
so because both $v_{ij}^{\geom}$ and $G_iv_{ij}$ are orthogonal to $N_i^\G$ we have
\begin{equation}
 G_iv_{ij}=(-1)^{s+s_i}v_{ij}^{\geom},\quad G_iN_i^{can}=(-1)^{s_i}N_i^{\geom}
\end{equation}
thus
\begin{equation}\label{eq:G-Ggeom}
 G_i=G_i^{\geom} I^{s_i} (IR_{N^{can}_i})^{s}
\end{equation}

\begin{lm}\label{lm:SO-geom}
Geometric $\SO$ solutions are in 1-1 correspondence with reconstructed from lengths $4$-simplices and choices 
$s\in\{0,1\}$ and $s_i\in \{0,1\}$ such that
\begin{equation}\label{eq:det}
 \forall_i\ \det G_i^{\geom}=(-1)^{s}.
\end{equation}
The relation is given by
\begin{equation}
 G_i=G_i^{\geom} I^{s_i} (IR_{N^{can}_i})^{s}.
\end{equation}
\end{lm}

\begin{proof}
 The condition for $\{G_i\}$ to be a geometric solution is $\det G_i=1$. This is equivalent by \eqref{eq:G-Ggeom} to
 \begin{equation}
  1=\det G_i^{\geom} I^{s_i} (IR_{N^{can}_i})^{s}=(-1)^{s}\det G_i^{\geom} 
 \end{equation}
\end{proof}

Let us notice that as there is only one reconstructed $4$ simplex up to rotation from $\OO$ thus two geometric rotations are always related by 
\begin{equation}
 G_i^{\geom'}=GG_i^{\geom},\quad G\in \OO
\end{equation}
thus 
\begin{equation}
 \forall_i \frac{\det G_i^{\geom'}}{\det G_i^{\geom}}=\det G
\end{equation}
and we can introduce a definition

\begin{df}\label{df:orient}
Suppose that the geometric boundary data satisfies the lengths matching condition. We say that it satisfies orientations matching condition if for any (and thus for all) reconstructed $4$ simplices
\begin{equation}
 \exists_{r\in\{-1,1\}}\ \forall_i \det G_i^{\geom}=r
\end{equation}
\end{df}

Let us notice that after we choose reconstructed $4$ simplex ($\OO$ transformation), the choice of $s_i$ is arbitrary and corresponds to involution gauge transformations. The value of $s$ is fixed by
\begin{equation}
 r=(-1)^s
\end{equation}
and it is Pleba\'nski orientation.

\begin{thm}\label{thm:nondeg-SO}
The non-degenerate $\SO$ geometric solution exists if and only if the geometric boundary data satisfies the lengths and orientations matching conditions. If we have one gauge equivalence class of geometric solutions $\{G_i\}$ then the representative of the second one is given as follows
\begin{equation}\label{eq:GRGR}
 \tilde{G}_i=R_{e_\alpha}G_iR_{N_i^{can}}
\end{equation}
where $e_\alpha$ is any normalized to $\pm 1$ vector. The second class corresponds to reflected $4$ simplex. Any non-degenerate solution is gauge equivalent to one of these two.
\end{thm}

\begin{proof}
The identification comes from lemma \ref{lm:SO-geom} and the description of the gauge transformations. There are two reconstructed $4$-simplices up to $\SO$ rotations. The choice of $s_i$ corresponds to involution gauge transformations.

Given such simplex from one class the representative of the second class can be obtained by applying a reflection 
\begin{equation}
 B_{ij}^{\tG}=R_{e_\alpha} (B_{ij}^{\G}),\quad s'=s+1.
\end{equation}
Thus for geometric $\SO$ solutions
\begin{equation}
 \tilde{G}_i=R_{e_\alpha}\ G_i\ (IR_{N^{can}_i}),
\end{equation}
where $e_\alpha$ is normalized. The equation \eqref{eq:GRGR} follows by applying an inversion gauge transformation.
\end{proof}

\subsection{Classification of geometric solutions}

We will consider only the case of non-degenerate boundary data that is that for every $i$, every $3$ out of $4$ vectors $v_{ij}$ are independent.
General classification was done in \cite{Barrett2003}.

\begin{lm}\label{lm-ijk}
For a geometric solution $\{G_i\}$ for a non-degenerate geometric boundary data for any different $i,j,k$ one of the following holds
\begin{itemize}
 \item[a)] $N_i^\G=\pm N_k^\G$ and $N_j^\G=\pm N_k^\G$
 \item[b)] $N_i^\G\not=\pm N_j^\G$
\end{itemize}
\end{lm}

\begin{proof}
The conditions are exclusive. Suppose neither holds, then $N_i^\G=\pm N_j^\G$ and their are not parallel to $N_k^\G$. In this case
\begin{equation}
 G_kv_{ki}\approx G_kv_{kj},
\end{equation}
but this contradicts non-degeneracy of the boundary data. Thus exactly one of the two conditions must be satisfied.
\end{proof}

\begin{lm}\label{lm:degen-parallel}
 There are two exclusive possibilities for solution $\{G_i\}$ for a non-degenerate geometric boundary data:
 \begin{itemize}
  \item[a)] All $N_i^\G$ are parallel.
  \item[b)] Geometric solution is non-degenerate.
 \end{itemize}
\end{lm}

\begin{proof}
From lemma \ref{lm-ijk} we can conclude that either all $N_i^\G$ are parallel or they are pairwise independent
\begin{equation}
 \forall_{i\not=j} N_i^\G\not=\pm N_j^\G
\end{equation}
Let us consider the second case. We will prove that there exists only one (up to a scaling) solution $W_i$ of
\begin{equation}
 \sum_i W_i N_i^\G=0,
\end{equation}
and for nontrivial solution all $W_i\not=0$. As there are $5$ vectors in $4$ dimensional space at least one solution exists.
We know that
\begin{equation}
 G_iv_{ij}\wedge N_i^\G=\lambda_{ij}  N_j^\G\wedge N_i^\G, \quad \lambda_{ij}\not=0
\end{equation}
as $N_i^\G$ is independent of $N_j^\G$.

For any $i$
\begin{equation}
 0= \sum_j W_j N_j^\G\wedge N_i^\G=\left(\sum_{j\not=i} \frac{W_j}{\lambda_{ij}}G_iv_{ij}\right)\wedge N_i^\G
\end{equation}
as $G_iv_{ij}$ are perpendicular to $N_i^\G$
\begin{equation}
 0=\sum_j \frac{W_j}{\lambda_{ij}}G_iv_{ij}\ \Longrightarrow\ \frac{W_j}{\lambda_{ij}}=\frac{W_k}{\lambda_{ik}}
\end{equation}
from non-degeneracy of $v_{ij}$ as for any $i$ $\sum_{j\not= i} G_iv_{ij}=0$. The ratio of $W_j$ to $W_k$ is fixed and nonzero for $j,k\not=i$. However choice of $i$ is arbitrary so the solution is unique up to a constant and with all $W_i\not=0$.

This is equivalent to the solution being non-degenerate.
\end{proof}

%% file: chapter-Other-signatures.tex
%!TEX root = Spin_foams_with_time_like_tetrahedra.tex
\section{Other signature solutions}\label{sec:other}

Let us notice that from Lemma \ref{lm:degen-parallel} it follows that in the case when $N_i^{can}$ are of different types and the boundary data is non-degenerate then degenerate geometric solutions cannot occur. The case of all $N_i^{can}$ timelike was describe in \cite{Frank2, Frank3}. In this case, if lengths and orientations matching conditions are satisfied then either the Gram matrix is degenerate and geometric solution corresponds to degenerate $4$-simplex or critical points occur in pairs (there are exactly two geometric solutions) and one can associate with them an Euclidean $4$-simplex. The difference of the phases is proportional to Regge action again, but the proportionality constant is different.

Our goal is to provide uniform treatment of both cases $N^{can}=e_0$ and $N^{can}=e_3$ where 
\begin{equation}
 \forall_i\ N_i^{can}=N^{can}
\end{equation}

\subsection{Vector geometries}

Let us denote
\begin{equation}
 V=\{v\in M^4\colon v\perp N^{can}\}
\end{equation}
We can consider subgroup of $\SO(1,3)$ that preserves $N^{can}$
\begin{equation}
 \SO(V)=\{G\in \SO(1,3)\colon GN^{can}=N^{can}\}
\end{equation}
Let us recall a notion of vector geometry introduced in \cite{Frank2}.

\begin{df}\label{def:vector-geometry}
 For the geometric $\SO(1,3)$ non-degenerate boundary data $v_{ij}$ (satisfying closure condition) we call a vector geometry  a collection
 \begin{equation}
 \{ G_i\in \SO(V)\}
 \end{equation}
 such that
 \begin{equation}
  v_{ij}^\G=-v_{ji}^\G
 \end{equation}
 where $v_{ij}^\G=G_iv_{ij}$,
modulo gauge transformations
\begin{equation}
 G_i\rightarrow GG_i,\quad G\in \SO(V)
\end{equation}
\end{df}

We have

\begin{lm}\label{lm:vector-geom}
 There is 1-1 correspondence between $\SO(1,3)$ geometric degenerate solutions and vector geometries (up to gauge equivalence on both sides).
\end{lm}

\begin{proof}
As all $N_i^\G$ are parallel there exists $G\in \SO(1,3)$ such that $GN_i^\G=(-1)^{s_i}N^{can}$, where $s_i\in\{0,1\}$. We can apply rotation and inversion gauge transformations to get $G_iN^{can}=N^{can}$. We can thus assume that $N_i^{\G}=N^{can}$. The remaining gauge freedom is as in vector geometries.

In this situation $G_iN^{can}=N^{can}$ so $G_i\in \SO(V)$
 \begin{equation}
  B_{ij}^\G={\Hodge}G_i(v_{ij}\wedge N^{can})={\Hodge} (G_iv_{ij}\wedge N^{can})
 \end{equation}
and we can define $v_{ij}^\G=G_iv_{ij}\in V$. One can check that all conditions for geometric solutions are equivalent to conditions for vector geometries.
\end{proof}

\subsection{Other signature solutions}

In this section we will relate pairs of degenerate solutions with non-degenerate simplex but of different than Lorentzian signature. We describe it in unified language applicable to both $N^{can}=e_0$ and $N^{can}=e_3$.

Let us introduce auxiliary space ${M^4}'$ that differs from Minkowski space $M^4$ by flipping the norm of $N^{can}$
\begin{equation}
 g_{\mu\nu}'=g_{\mu\nu}-2\frac{N^{can}_\mu N^{can}_\nu}{N^{can}\cdot N^{can}}
\end{equation}
where we used $g_{\mu\nu}$ for lowering indices. 
We will use prime to distinguish operations related to this metric (Hodge star $\Hodgp$, scalar product $\cdop$, contraction with use of the metric $\lrcornep$).
We introduce
\begin{equation}
 t^{can'}=N^{can}\cdop N^{can}=-N^{can}\cdot N^{can}\in\{-1,1\}.
\end{equation}
Let us notice that restricted to $V$ both scalar products coincide, thus $V$ can be regarded as subspace of both $M^4$ as well ${M^4}'$. For vectors in $V$ we can use exchangeably both scalar products.
The Hodge ${\Hodgp}$ operation satisfies in ${M^4}'$
\begin{equation}
 {\Hodgp}^2=1
\end{equation}
and inversion $I\in \SO({M^4}')$.
Let us introduce
\begin{equation}
 \Phi^\pm\colon \Lambda^2 {M^4}'\rightarrow V,\quad \Phi^\pm(B)=t^{can'}(\pm B+ {\Hodgp}B)\lrcornep N^{can}
\end{equation}
where we regard $V$ as a subspace of ${M^4}'$.
Let us notice also that for $v\in V\subset {M^4}'$
\begin{equation}\label{eq:v-B}
  \Phi^\pm({\Hodgp}v\wedge N^{can})=v
\end{equation}
The map $\tilde{\Phi}=(\Phi^+,\Phi^-)\colon \Lambda^2 {M^4}'\rightarrow V\oplus V$ is an isomorphism
\begin{equation}\label{eq:phi-minus}
 \tilde{\Phi}^{-1}(v^+,v^-)=\frac{1}{2}[(v^+-v^-)\wedge N^{can}+{\Hodgp}(v^++v^-)\wedge N^{can}]
\end{equation}
We can use this isomorphism to transform the action of $\SO({M^4}')$ from bivectors into $V\oplus V$. Let us notice that as $\SO({M^4}')$ preserves the decomposition into self-dual and anti-self-dual forms, the action is diagonal and we can define 
\begin{equation}
 \SO({M^4}')\ni G\rightarrow (\PHi^+(G),\PHi^-(G))
\end{equation}

\begin{lm}
$\PHi^\pm(G)\in O(V)$ 
\end{lm}

\begin{proof}
The action on $\Lambda^2{M^4}'$ preserves the scalar product 
\begin{equation}
 B\cdop B={\Hodgp}(B\wedge {\Hodgp}B)
\end{equation}
It can be translated into $V\oplus V$ as follows. We need to compute
\begin{equation}
 \tilde{\Phi}^{-1}(v^+,v^-)\cdot \tilde{\Phi}^{-1}(v^+,v^-)
\end{equation}
but this is by \eqref{eq:phi-minus} after cancellations
\begin{align}
 &\frac{1}{2} {\Hodgp}[v^+\wedge N^{can}\wedge {\Hodgp}(v^+\wedge N^{can})+v^-\wedge N^{can}\wedge {\Hodgp}(v^-\wedge N^{can})]=\nonumber\\
 &\frac{1}{2}t^{can'}(v^+\cdot v^++v^-\cdot v^-)\label{eq:VV}
\end{align}
Thus
\begin{equation}
 \tilde{\Phi}^{-1}(v^+,v^-)\cdot \tilde{\Phi}^{-1}(v^+,v^-)=
 \frac{1}{2}t^{can'}(v^+\cdot v^++v^-\cdot v^-)
\end{equation}
So $\PHi^\pm(G)$ preserves the scalar product on $V$.
\end{proof}

Let us notice that the condition for simplicity of a bivector $B$
\begin{equation}
 0={\Hodgp}(B\wedge B)=\frac{1}{2}t^{can'}(v^+\cdot v^+-v^-\cdot v^-)\label{eq:simVV}
\end{equation}
is equivalent to $|v^+|^2=|v^-|^2$.

\begin{lm}\label{lm-exact}
 The following exact sequence holds
  \begin{equation}
  0\rightarrow \{1,I\}\rightarrow \SO({M^4}')\rightarrow^{\tilde{\PHi}} \SO(V)\times \SO(V)\rightarrow^{\tilde{\sgn}} Z_2
 \end{equation}
 where we defined
 \begin{align}
  \tilde{\PHi}(G)&=(\PHi^+(G),\PHi^-(G)),\\
  \tilde{\sgn}(G^+,G^-)&=\sgn(G^+)\sgn(G^-),
 \end{align}
and
 \begin{equation}
   \sgn(G^\pm)=\left\{\begin{array}{ll}
                1 & N^{can}=e_0\\
                1 & N^{can}=e_3\text{ and } G^\pm \text{ is time direction preserving},\\
                -1 & N^{can}=e_3\text{ and } G^\pm \text{ is time direction reversing}.
               \end{array}\right.
 \end{equation}
\end{lm}

\begin{proof}
Let us consider $N^{can}=e_3$.
 The only group elements preserving all bivectors are $1$ and $I$. We need to find now the image of $\SO({M^4}')$. The image of a connected component $\SO_+({M^4}')$ is $\SO_+(V)\times \SO_+(V)$ (dimensions of groups match). We need to determine the image of the group element
 \begin{equation}
  G=R_{e_1}R_{e_3}
 \end{equation}
as it is the generator of $\SO({M^4}')/\SO_+({M^4}')$. We have $G(e_3)=-e_3$ and $\PHi^\pm(G)$ are equal and
\begin{equation}
 \PHi^\pm(G)(e_0)=-e_0,\ \PHi^\pm(G)(e_1)=e_1,\ \PHi^\pm(G)(e_2)=-e_2,
\end{equation}
so $\PHi^\pm(G)\in \SO(V)$, $\sgn(\PHi^\pm(G))=-1$ and the image is in the kernel of $\tilde{\sgn}$ as $\sgn(\PHi^+(G))\sgn(\PHi^-(G))=1$.

In the case of $N^{can}=e_0$ 
\begin{equation}
 0\rightarrow \{1,I\}\rightarrow \SO({M^4}')\rightarrow^{\tilde{\PHi}} \SO(V)\times \SO(V)\rightarrow 0 
\end{equation}
by similar reasoning as before.
\end{proof}

\begin{lm}\label{lm-sign-vec}
 Let us suppose that we have two vector geometries $\{G_i\}$ and $\{G_i'\}$ then
 \begin{equation}
  \sgn G_i \sgn G_i'
 \end{equation}
is equal for every $i$.
\end{lm}

\begin{proof} 
It is a trivial statement for $N^{can}=e_0$. Let us consider $N^{can}=e_3$.

It is enough to show that 
 \begin{equation}
  s=\sgn G_i \sgn G_j\text{ is equal to } s'=\sgn G_i' \sgn G_j'
 \end{equation}
We know that 
\begin{equation}
 v_{ij}^\G=-v_{ji}^\G,\quad 
\end{equation}
so if one is future directed then the second is past directed. Group elements changing future into past directed vectors have $\sgn=-1$ thus
\begin{equation}
 s=\left\{\begin{array}{ll}
           1 & v_{ij} \text{ and } v_{ji} \text{ has opposite time direction}\\
           -1 & v_{ij} \text{ and } v_{ji} \text{ has the same time direction}
          \end{array}\right.
\end{equation}
for $s'$ it is the same conditions thus $s=s'$.
\end{proof}

\begin{lm}\label{lm:pm-G}
 The two statement are equivalent for $G\in \SO({M^4}')$
 \begin{itemize}
  \item $G N^{can}=(-1)^s N^{can}$, $s\in\{0,1\}$
  \item $\PHi^+(G)=\PHi^-(G)$
 \end{itemize}
 and then for $GI^s$ regarded as an element of $\SO(V)$
 \begin{equation}
  \PHi^\pm(G)=GI^s
 \end{equation}
\end{lm}

\begin{proof}
 In one direction: follows from the definition of $\Phi^\pm$ and \eqref{eq:v-B} that 
 \begin{equation}
  \PHi^\pm(G)(v)=\PHi^\pm(G)\Phi^\pm({\Hodgp}(v\wedge N^{can})=\Phi^\pm( {\Hodgp}(Gv\wedge \underbrace{GN^{can}}_{=(-1)^sN^{can}}))=
  (-1)^s Gv
 \end{equation}
so $\PHi^+(G)=\PHi^-(G)$.
 
If $\PHi^\pm(G)=\tilde{G}\in \SO(V)$ then
 \begin{equation}
   G=\tilde{G} I^{s}\in \SO({M^4}')
 \end{equation}
 where  $s\in\{0,1\}$.
\end{proof}

\subsection{Correspondence}

Let us notice that $\SO(1,3)$ geometric boundary data with all $N_i^{can}=N^{can}$ can be regarded as $\SO({M^4}')$ geometric boundary data. We will call it \emph{flipped geometric boundary data}.

\begin{thm}\label{thm:deg-geometric}
 There is a 1-1 correspondence between
 \begin{itemize}
  \item ordered pairs of two non-gauge equivalent vector geometries,
  \item geometric $\SO({M^4}')$ non-degenerate solutions (up to the inversion gauge transformations) for flipped geometric boundary data.
 \end{itemize}
 From the $\SO({M^4}')$ solution $\{G_i\}$ the two vector geometries $\{G_i^\pm\}$ are obtained as
 \begin{equation}
  G_i^\pm=\PHi^\pm(G_i)
 \end{equation}
 and $\Phi^\pm(B_{ij}^{\G})=v_{ij}^{{\{G^\pm\}}}$.
 \end{thm}

\begin{proof}
Let us consider geometric $\SO({M^4}')$ solution.
We have
\begin{equation}
 \Phi^\pm(B_{ij})=v_{ij},\quad \Phi^\pm(B_{ij}^\G)=G_i^\pm v_{ij}
\end{equation}
so the following conditions are equivalent
\begin{equation}
 B_{ij}^\G=-B_{ji}^\G\Longleftrightarrow v_{ij}^{\Gp}=-v_{ji}^{\Gp}\text{ and } v_{ij}^{\Gm}=-v_{ji}^{\Gm}
\end{equation}
Let us now suppose that we have two vector geometries $G_i^+$ and $G_i^-$. From lemma \ref{lm-sign-vec}, $s=\sgn G_i^+\sgn G_i^-$ is the same for all $i$. By gauge transforming $G_i^-$ by $G$ such that $\sgn G=s$ we can obtain situation when
\begin{equation}
 \forall_i \sgn G_i^+\sgn G_i^-=1
\end{equation}
so there exist (unique up to $I^{s_i}$) elements $G_i\in \SO({M^4}')$ that constitute an $\SO({M^4}')$ solution. 

Gauge equivalent $\SO({M^4}')$ geometric solutions are obtained from gauge equivalent pairs of  vector geometries. The only thing that is left is to prove that the $\SO({M^4}')$ geometric solution is non-degenerate exactly when two vector geometries are not gauge equivalent.

Let us assume that the geometric solution is degenerate. As boundary data is non-degenerate, the geometric solution is degenerate when all $N_i^\G$ are parallel. By gauge transformation (including inversion gauge transformation) we can assume that $N_i^\G=N^{can}$ then
\begin{equation}
 G_i^+=G_i^-
\end{equation}
so vector geometries were gauge equivalent.

Other way around, if two vector geometries are equivalent then by gauge transformation we can assume $G_i^+=G_i^-$ and $N_i^\G=G_iN^{can}=(-1)^{s_i}N^{can}$ by lemma \ref{lm:pm-G}.
\end{proof}

\begin{lm}\label{lm:deg-split}
Suppose that we have a non-degenerate geometric boundary data satisfying lengths and orientations matching conditions.
There is a 1-1  correspondence between gauge equivalent classes of geometric non-degenerate $\SO({M^4}')$ solutions and reconstructed $4$ simplices in ${M^4}'$ (up to shifts and $\SO({M^4}')$ rotations). Gauge equivalent classes of non-degenerate $\SO({M^4}')$ geometric solutions occur in pairs related to reflected $4$-simplices and their representatives are related by
\begin{equation}
 G_i'=R_{N^{can}} G_i R_{N^{can}}
\end{equation}
\end{lm}

Let us notice the following identity for $G\in \SO({M^4}')$
\begin{equation}
 \PHi^\pm(R_{N^{can}} G R_{N^{can}})=\PHi^\mp(G)\label{eq:RSO}
\end{equation}
Two non-equivalent geometric solutions thus satisfy
\begin{equation}
 \PHi^\pm(G_i')=\PHi^\mp(G_i)
\end{equation}

\subsubsection{Orientations}

We will now describe orientations matching condition in terms of self-dual and anti-self-dual forms. Let us introduce
\begin{equation}\label{two-vgeom}
 v_{ij}^{\geom\pm}=\Phi^\pm(B_{ij}^\geom)
 %t^{can}(\pm B_{ij}^{\geom}+{\Hodge}B_{ij}^{\geom})\llcorner N^{can}
\end{equation}
From simplicity of $B_{ij}^{\geom}$ and \eqref{eq:VV} and \eqref{eq:simVV}
\begin{equation}
 v_{ij}\cdot v_{ik}=v_{ij}^{\geom\pm}\cdot v_{ik}^{\geom\pm}
\end{equation}
We can introduce $G_i^{\geom\pm}\in O(V)$ by conditions
\begin{equation}
 \forall_{j\not=i}\ G_i^{\geom\pm} v_{ij}=v_{ij}^{\geom\pm}
\end{equation}

\begin{lm}\label{lm:orient-deg}
 $\det G_i^{\geom}=\det G_i^{\geom\pm}$
\end{lm}

\begin{proof}
 We can write $G_i^{\geom}=\tilde{G}_i (IR_{N^{can}})^{s_i}$ where $(-1)^{s_i}=\det G_i^{\geom}$ and $\tilde{G}_i\in \SO({M^4}')$. Let us notice that
 \begin{equation}
  \tilde{G}_i^{-1}N_i^{\geom}=N^{can}
 \end{equation}
 We have
 \begin{align*}
  &\PHi^\pm(\tilde{G}_i^{-1})G_i^{\geom\pm}v_{ij}=\PHi^\pm(\tilde{G}^{-1})v_{ij}^{\geom\pm}=\\
  &=t^{can'}(\pm \tilde{G}_i^{-1}{\Hodgp}(v_{ij}^{\geom}\wedge N_i^{\geom})+\tilde{G}_i^{-1}(v_{ij}^{\geom}\wedge N_i^{\geom}))\lrcornep N^{can}=\\
  &=t^{can'}(\pm {\Hodgp}(\tilde{G}_i^{-1}v_{ij}^{\geom}\wedge N^{can})+\tilde{G}_i^{-1}v_{ij}^{\geom}\wedge N^{can})\lrcornep N^{can}=\\
  &=\tilde{G}_i^{-1}v_{ij}^{\geom}=\tilde{G}_i^{-1}G_i^{\geom}v_{ij}
 \end{align*}
and as also
\begin{equation}
 \tilde{G}_i^{-1}G_i^{\geom}N^{can}=N^{can}
\end{equation}
we see that
\begin{equation}
 \det G_i^{\geom\pm}=\det \PHi^\pm(\tilde{G}_i^{-1})G_i^{\geom\pm}=\det \tilde{G}^{-1}G_i^{\geom}=\det G_i^{\geom}
\end{equation}
because $\tilde{G}_i\in \SO({M^4}')$ and $\tilde{G}_i^\pm\in \SO(V)$.
\end{proof}

%% file: chapter-Classification.tex
%!TEX root = Spin_foams_with_time_like_tetrahedra.tex
\section{Classification of solutions}\label{sec:classification}

In this section we will classify solutions and determine under which circumstances which geometric $\SO(1,3)$ solution can occur. We assume that the boundary data is non-degenerate, that is, that for every $i$, every $3$ out of the $4$ vectors $v_{ij}$ are independent. Using the correspondence between the geometric $\SO(1,3)$ solutions and critical points, this gives, in fact, a classification of the latter.

\subsection{Non-degenerate simplices}

Let us recall that non-degenerate $\SO(1,3)$ geometric solution can occur only if the boundary data is non-degenerate.

\begin{thm}\label{thm-1}
 Let us assume that the geometric boundary data is non-degenerate. Then the following are equivalent:
 \begin{itemize}
  \item lengths and orientations matching conditions are satisfied and the reconstructed $4$-simplex is non-degenerate Lorentzian,
  \item there exists a non-degenerate $\SO(1,3)$ geometric solution,
  \item there exist exactly two classes of gauge equivalent non-degenerate $\SO(1,3)$ geometric  solutions. We can choose their representatives $\{G_i\}$ and $\{\tilde{G}_i\}$ to be related by
 \begin{equation}
 \tilde{G}_i=R_{e_\alpha}G_iR_{N_i^{can}},\quad e_\alpha \text{ normalized}.
\end{equation}  
 \end{itemize}
\end{thm}

\begin{proof}
 Follows from theorem \ref{thm:nondeg-SO}.
\end{proof}

If we have one $\SO(1,3)$ geometric solution with additional choice of the gauge such that $G_i\in \SO_+(1,3)$, then the representative of the second gauge equivalence class is given as follows:
\begin{equation}
 \tilde{G}_i=R_{e_0}G_iR_{N_i^{can}}I^{r_i}\in \SO_+(1,3),
\end{equation}
where we added
\begin{equation}
 r_i=\left\{\begin{array}{ll}
             0 & \text{ when } N^{can}_i=e_0,\\ 1 & \text{ when } N^{can}_i=e_3,
            \end{array}\right.
\end{equation}
such that also the $\tilde{G}_i$  are in the connected component of the identity.

\begin{thm}\label{thm-2}
 Let us assume that the geometric boundary data is non-degenerate. Then the following are equivalent:
 \begin{enumerate}
  \item\label{one} lengths and orientations matching conditions are satisfied, and the reconstructed $4$-simplex is non-degenerate of split or Euclidean signature,
  \item \label{two} there exists a non-degenerate $\SO({M^4}')$ geometric solution for flipped geometric boundary data,
  \item\label{three} there exist exactly two gauge equivalence classes of non-degenerate $\SO({M^4}')$ geometric solutions for flipped geometric boundary data. We can choose  their representatives $\{G_i\}$ and $\{\tilde{G}_i\}$ to be related by
 \begin{equation}
 \tilde{G}_i=R_{e_\alpha}G_iR_{N^{can}}, \quad e_\alpha \text{ normalized},
\end{equation}  
\item\label{four} there exist at least two gauge equivalence classes of vector geometries $\{G_i^+\}$ and $\{G_i^-\}$,
\item\label{five} there exist exactly two gauge equivalence classes of vector geometries.
 \end{enumerate}
 The vector geometries $\{G_i^+\}$ and $\{G_i^-\}$ are obtained from geometric $\SO({M^4}')$ solution $\{G_i\}$ by
 \begin{equation}
  G_i^\pm=\PHi^\pm(G_i)
 \end{equation}
\end{thm}

\begin{proof}
 The equivalences $\ref{one}.\Leftrightarrow \ref{two}.\Leftrightarrow \ref{three}.$ follow from lemma \ref{lm:deg-split}. $\ref{three}.\Rightarrow \ref{five}.$: By theorem \ref{thm:deg-geometric} from an ordered pair of non-equivalent vector geometries one obtains one non-degenerate geometric solution. If there existed a third class of vector geometries, then taking all 6 possible ordered pairs, one would obtain 6 different geometric solutions. This contradicts \ref{three}. $\ref{five}\Rightarrow \ref{four}$ is trivial. $\ref{four}\Rightarrow \ref{two}$ is just an application of theorem \ref{thm:deg-geometric}.
\end{proof}

\subsection{Degenerate geometric solutions}

In this section we will prove several results about vector geometries under assumption that the lengths matching condition is satisfied.

\subsubsection{Possible signatures}

Let us now suppose that lengths matching condition is satisfied and we have a vector geometry $\{G_i\in \SO(V)\}$. We would like to determine possible signatures of $4$ simplices reconstructed from the Gram matrix. Let us denote reconstructed spacetime by $\underline{M^4}$.

For that let us choose an edge $34$ of the reconstructed $4$-simplex. As all faces are spacelike, the edge needs also to be spacelike. Let us denote its vector by $e_{34}^{\geom}$. The space perpendicular to $e_{34}^{\geom}$ can be identified with the subspace $W$ of bivectors of the form
\begin{equation}
 W=\{B\colon \exists n\perp e_{34}^{\geom}, B=n\wedge e_{34}^{\geom}\},\quad B^1\cdou B^{2}=- n^1 \cdou  n^{2}
\end{equation}
This space is spanned by geometric bivectors of the faces $(kl)$, $k,l\in\{0,1,2\}$ sharing the edge $e_{34}^{\geom}$
\begin{equation}
 W=\text{span}\{B_{kl}^{\geom},\quad k<l,\quad k,l\in \{0,1,2\}\}
\end{equation}
As every two such faces share a tetrahedron we can compute the scalar product between their bivectors without actual knowledge about the reconstruction, for example (as all faces are spacelike)
\begin{equation}
 -n_{kl}^\geom\cdou n_{kl'}^\geom=B_{kl}^{\geom}\cdou B_{kl'}^{\geom}=B_{kl}\cdou B_{kl'}.
\end{equation}
This can be also computed in the Lorentzian signature
\begin{equation}
B_{kl}\cdou B_{kl'}= B_{kl}\cdot B_{kl'}=-{\Hodge}B_{kl}^{\geom}\cdot {\Hodge}B_{kl'}^{\geom}=-N^{can}\cdot N^{can}v_{kl}^{\geom}\cdot v_{kl'}^{\geom}= -t^{can} v_{kl}^\G\cdot v_{kl'}^\G
\end{equation}
where $t^{can}=N^{can}\cdot N^{can}$.

Let us now notice that vectors $e_{3i}^{\geom}$ can be expressed as a linear combinations of $n_{kl}^{\geom}$ for $k,l\notin \{3,4\}$ and $e_{34}^{\geom}$. Therefore the vectors
\begin{equation}
 n_{kl}^{\geom},\ (kl)\in\{(01),(12),(20)\},\ e_{34}^{\geom}
\end{equation}
span the whole space of the $4$ simplex as vectors $\{e_{3i}^{\geom}\colon i\in\{0,1,2,4\}\}$ do.

Let us now consider the matrix of scalar products of $e_{34}^{\geom}$ and vectors $n_{kl}^{\geom}$, $(kl)\in\{(01),(12),(20)\}$: 
\begin{equation}\label{eq:scalar}
G'= \left(\begin{array}{c|c}
        n_{kl}^{\geom}\cdou n_{k'l'}^{\geom} & 0\\
        \hline
        0 & e_{34}^{\geom}\cdou e_{34}^{\geom}
       \end{array}\right).
\end{equation}
The submatrix $n_{kl}^{\geom}\cdou n_{k'l'}^{\geom}$ has the same signature as $t^{can} v_{kl}^\G\cdot v_{k'l'}^\G$. 

\begin{lm}
If lengths matching condition is satisfied, the reconstructed $4$ simplex is non-degenerate and there exists a vector geometry then the signature of the reconstructed $4$-simplex is either $++--$ or $----$.
\end{lm}

\begin{proof}
The matrix $v_{kl}^\G\cdot v_{k'l'}^\G$ ($(kl)\in\{(01),(02),(12)\}$ has rank $3$ thus is non-degenerate and has the same signature as a matrix of scalar products in $V$. Let us consider two cases:
\begin{itemize}
 \item $N^{can}=e_0$. Since $t^{can}=1$ the signature of $G'$ is $(---)$ plus $(-)$
 \item $N^{can}=e_3$. Since $t^{can}=-1$ the signature of $G'$ is $(-++)$ plus $(-)$.
\end{itemize} 
\end{proof}

Let us introduce for any pairwise different $i,j,k,\in\{0,1\ldots 4\}$ a matrix
\begin{equation}\label{matrix-tildeG}
  \tilde{G}_{ijk}=\left(\begin{array}{ccc}
                   v_{ij}\cdot v_{ij} & -v_{ji}\cdot v_{jk} & -v_{ij}\cdot v_{ik}\\
                   -v_{jk}\cdot v_{ji} & v_{jk}\cdot v_{jk} & -v_{kj}\cdot v_{ki}\\
                   -v_{ik}\cdot v_{ij} & -v_{ki}\cdot v_{kj} & v_{ki}\cdot v_{ki}
                  \end{array}\right).
\end{equation}

\begin{lm}\label{lm:gram-deg}
If lengths matching condition is satisfied and there exists a vector geometry then the following are equivalent
\begin{itemize}
 \item the reconstructed $4$-simplex is degenerate
 \item for any pairwise different $i,j,k,\in\{0,1\ldots 4\}$ the matrix $\tilde{G}_{ijk}$ from \eqref{matrix-tildeG} satisfies
 \begin{equation}
 \det \tilde{G}_{ijk}=0
 \end{equation}
\end{itemize}
\end{lm}

\begin{proof} 
The Matrix $G'$ from \eqref{eq:scalar} is degenerate exactly when the matrix of the scalar products $v_{kl}^\G\cdot v_{k'l'}^\G$, $(kl)\in\{(01),(12),(20)\}$ is degenerate. The latter is equal to $\tilde{G}_{012}$. As choice of indices is arbitrary we obtain equivalence.
\end{proof}

\subsubsection{Possible degenerate points}

We will now prove a version of Pleba\'nski classification (see \cite{plebanski, Urbantke1984, Bengtsson1995, Reisenberger2016,Baez}).

\begin{lm}\label{lm:two-deg}
Let us assume that we have non-degenerate geometric boundary data for $N_i^{can}=N^{can}$.
There exist at most two (up to an overall rotation by $G\in \OO(V)$) sets of vectors $\{w_{ij}\}_{i\not=j,i,j\in\{0,\ldots 4\}}$ in $V$ satisfying
 \begin{enumerate}
  \item $\forall_{i\not=j}\ w_{ij}=-w_{ji}$
  \item $\forall_i\ \sum_{j\not=i} w_{ij}=0$
  \item $\forall_{i,j,k} w_{ij}\cdot w_{ik}=v_{ij}\cdot v_{ik}$
 \end{enumerate}
If for any pairwise different $i,j,k,\in\{0,1\ldots 4\}$ the matrix from \eqref{matrix-tildeG} satisfies
 \begin{equation}
  \det \tilde{G}_{ijk}=0
 \end{equation}
then there exists at most one solution.
\end{lm}

\begin{proof}
 We will first prove that there are at most two matrices $G_{ij,kl}$ (indices are pairs $i \not=j$ and $k\not=l$) satisfying
 \begin{equation}
  G_{ij,ik}=v_{ij}\cdot v_{ik},\quad G_{ij,kl}=-G_{ji,kl}=-G_{ij,lk}=G_{kl,ij},\quad 
  \sum_{j\not= i} G_{ij,kl}=0
 \end{equation}
 Matrix $G_{ij,kl}$ should be the matrix of scalar product $w_{ij}\cdot w_{kl}$. If there are at most two such matrices there will be at most two sets of vectors $\{w_{ij}\}$ up to rotations from $\OO(V)$.

The only undetermined entries are $G_{ij,kl}$ where all indices are different. Let us assume that $i,j,k,l$ are different and $m$ is the lacking index. We know that
\begin{equation}
 0=\sum_{n\not=i} G_{in,kl}=\overbrace{\sum_{n\in\{k,l\}} G_{in,kl}}^{\text{expressed by }v_{ab}\cdot v_{ac}} +G_{ij,kl}+G_{im,kl}.
\end{equation}
So $G_{ij,kl}$ can be expressed by $G_{im,kl}$ and $v_{ab}\cdot v_{ac}$. Similarly we can replace by $m$ any other index. As all permutations are generated by transpositions $(mi)$, $(mj)$, $(mk)$, $(ml)$, we see that we can express any $G_{ij,kl}$ with different indices as a linear combination of $G_{01,23}$ and $v_{ab}\cdot v_{ac}$. Let us denote $G_{01,23}=\alpha$.

As $w_{ij}$ are vectors in $V$, every determinant of a $4$ by $4$ minor of $G_{ij,kl}$ need to vanish. Let us take the determinant of
\begin{equation}
\left(\begin{array}{l|l}
 G_{0i,0j} & G_{0i,23}\\
 \hline
 G_{23,0j} & G_{23,23}
 \end{array}
 \right),\quad i,j\in\{1,2,3\}
\end{equation}
The only unknown entry is $\alpha=G_{01,23}=G_{23,01}$ and it appears twice in the matrix. The coefficients in the contribution of $\alpha^2$ to the determinant is $-(G_{02,03}^2-G_{02,02}G_{03,03})\not=0$ as $v_{02}$ and $v_{03}$ are not parallel. The determinant is thus quadratic polynomial and there are at most two solutions in $\alpha$ to the equation 
\begin{equation}
\det\left(\begin{array}{l|l}
 G_{0i,0j} & G_{0i,23}\\
 \hline
 G_{23,0j} & G_{23,23}
 \end{array}
 \right)=0.
\end{equation}

If $G_{ij,kl}$ is a matrix of scalar product then its range of  is $3$ dimensional (non-degeneracy of boundary data) and the signature is the same as signature of $V$. There exist vectors $M_{ij}^\mu\in V$ such that
\begin{equation}
 G_{ij,kl}=M_{ij}^\mu\eta_{\mu\nu} M_{kl}^\nu,\ \eta_{\mu\nu}\text{ scalar product in } V.
\end{equation}
If there exists another set of vectors ${M'_{ij}}^\mu$ then there exists $G\in \OO(V)$ such that
\begin{equation}
 M'_{ij}=GM_{ij}.
\end{equation}
This proves that there are at most two solutions up to rotations.

If $\det \tilde{G}_{012}=0$ then $w_{01}$, $w_{12}$, $w_{20}$ are not independent (if they were the signature of $\tilde{G}_{012}$ would be the same as $V$). Thus
\begin{align}
 0=\det \left(\begin{array}{ccc}
                   G_{01,12} & G_{01,20}& G_{01,23}\\
                   G_{12,12} & G_{12,20} & G_{12,23}\\
                   G_{20,12} & G_{20,20} & G_{20,23}
                  \end{array}\right)
 =\det \left(\begin{array}{ccc}
                   -v_{10}\cdot v_{12} & -v_{01}\cdot v_{02} & G_{01,23}\\
                   v_{21}\cdot v_{21} & -v_{21}\cdot v_{20} & -v_{21}\cdot v_{23}\\
                   -v_{20}\cdot v_{21} & v_{20}\cdot v_{20} & v_{20}\cdot v_{23}
                  \end{array}\right).                  
\end{align}
As $v_{21}$ and $v_{20}$ are independent (boundary non-degeneracy condition) this gives linear equation for $G_{01,23}$ as 
$(v_{21}\cdot v_{21})(v_{20}\cdot v_{20})-(v_{20}\cdot v_{21})^2\not=0$. So there exists at most one matrix of scalar products.
\end{proof}

\begin{lm}\label{lm:only-two}
There are at most two vector geometries up to gauge transformations. If 
for any pairwise different $i,j,k,\in\{0,1\ldots 4\}$ the matrix $\tilde{G}_{ijk}$ from \eqref{matrix-tildeG} satisfies
 \begin{equation}
 \det \tilde{G}_{ijk}=0
 \end{equation}
then there exists at most one.
\end{lm}

\begin{proof}
For the sets of vectors $\{w_{ij}\}$ satisfying assumptions of lemma \ref{lm:two-deg} we can introduce $G_i^{\{w\}}$ defined by (for non-degenerate boundary data)
\begin{equation}
 G_i^{\{w\}}v_{ij}=w_{ij}.
\end{equation}
We have $G_i^{\{w\}}\in \OO(V)$. Vectors $\{w_{ij}\}$ are defined by vector geometry if and only if $\det G_i^{\{w\}}=1$ for all $i$.
There exist at most two up to an overall $\OO(V)$ transformation sets of vectors from lemma \ref{lm:two-deg}. From two sets differing by rotations not from $\SO(V)$ at most one has $\forall_i\ \det G_i^{\{w\}}=1$, so there are at most two vector geometries up to gauge transformations.

If $\det \tilde{G}_{ijk}=0$ then there exists at most one set of possible vectors $\{w_{ij}\}$ up to rotations from $\OO(V)$. As before only one can have $\forall_i\ \det G_i^{\{w\}}=1$.
\end{proof}

\subsection{Degenerate $\SO(1,3)$ geometric solutions}

\begin{lm}
Let us assume that the lengths matching condition is satisfied, the geometric boundary data is non-degenerate and there exists a vector geometry $\{\tilde{G}_i\}$ and if in addition
\begin{itemize}
 \item the reconstructed $4$-simplex is non-degenerate in ${M^4}'$ (flipped with respect to $N^{can}$)  then orientations matching condition is satisfied and there exist exactly two gauge inequivalent classes of vector geometries,
 \item the reconstructed $4$-simplex is degenerate  then orientations matching condition is satisfied and there exists exactly one gauge equivalent class of vector geometries
\end{itemize}
\end{lm}

\begin{proof}
If lengths matching condition is satisfied, the reconstructed $4$ simplex is non-degenerate and there is one vector geometry then we can consider $3$ sets of vectors satisfying assumptions of lemma \ref{lm:two-deg}
\begin{equation}
 v_{ij}^{\geom\pm},\ v_{ij}^{\tilde{G}}
\end{equation}
and $v_{ij}^{\geom+}$ and $v_{ij}^{\geom-}$ from equation \eqref{two-vgeom} are not related by a rotation\footnote{If $v_{ij}^{\geom+}=Gv_{ij}^{\geom-}$ then $G=G_i^{\geom+}(G_i^{\geom-})^{-1}\in \SO(V)$ ($\det G_i^{\geom+}=\det G_i^{\geom-}$) for non-degenerate boundary data. In such case all bivectors of the reconstructed $4$-simplex are annihilated by $\PHi^-(G)N^{can}$, which contradicts its non-degeneracy.}.

From lemma \ref{lm:two-deg} we know that there exists $s\in\{+,-\}$ and $G\in \OO(V)$ such that
\begin{equation}
 v_{ij}^{\geom\ s}=G v_{ij}^{\tilde{G}}
\end{equation}
thus by lemma \ref{lm:orient-deg}
\begin{equation}
 G_i^{\geom\ s}=G\tilde{G}_i ,\quad \det G_i^{\geom}=\det G\tilde{G}_i=r,
\end{equation}
where we introduced $r=\det G$. As a result, the orientations matching condition is satisfied and by theorem \ref{thm-2} there exist two gauge inequivalent vector geometries. By lemma \ref{lm:only-two} there exists no other gauge class of vector geometry.

If the reconstructed $4$ simplex is degenerate then by lemma \ref{lm:gram-deg} $\det \tilde{G}_{012}=0$ and by lemma \ref{lm:two-deg} there exists at most one solution $w_{ij}$ (up to rotation from $\OO(V)$). We know one solution $v_{ij}^{\tilde{G}}$ and the second follows from the geometric construction of $v_{ij}^{\geom+}$ from \eqref{two-vgeom} ( it follows that they differ by $G\in \OO(V)$)
\begin{equation}
 r=\det G=\det G\tilde{G}_i=\det G_i^{\geom+}=\det G_i^{\geom}
\end{equation}
Thus orientations matching condition is satisfied. By lemma \ref{lm:only-two} there exists no other vector geometry.
\end{proof}

\subsection{Classification of solutions}

Degenerate and non-degenerate solutions cannot occur at the same time. If lengths matching condition is satisfied, but orientations matching condition is not satisfied then we cannot have any solution.

When lengths matching condition is not satisfied we can have neither non-degenerate $\SO(1,3)$ geometric solution (reconstructed $+---$ simplex satisfies lengths matching condition) nor two degenerate solutions (reconstructed $++--$ or $----$ simplex need to satisfy lengths matching condition). There still might exist a single vector geometry in this case.

By lemma \ref{lm:map} this classification applies to real critical points of the action.

%% file: chapter-Difference-of-phases.tex
%!TEX root = Spin_foams_with_time_like_tetrahedra.tex
\section{Phase difference}\label{sec:phase}

In this section we will give an interpretation in terms of Regge geometries of the difference of the phases between two critical points. The overall phase can be changed by adjusting the phases of the coherent states. Therefore, in the study of the asymptotic behaviour of the vertex amplitude the phase difference $\Delta S$ is of main interest. It is defined up to $2\pi i$ because it appears in the exponent. In Section \ref{sec:da} we will show that it agrees with the expected Regge term $\Delta S^\geom$ up to $\pi \iu$. We will improve this result in Section \ref{sec:pi} by using certain deformation argument and we will show that this agreement holds exactly, i.e. up to $2\pi \iu$:
$$
\Delta S =\Delta S^\Delta \mod 2\pi \iu.
$$

\subsection{The Regge term $\Delta S^\geom$}\label{sec:Regge}

Regge calculus \cite{Regge} was devised as an extension of gravitational action to certain distributional metrics that are flat everywhere except on the faces of the simplicial decomposition. On these faces, the geometry is not smooth anymore,  and deficit angles appear. Faces are assumed to have flat inner simplicial geometry. The action is a sum of contributions from single simplices given by $$S_{\Delta}=\sum_{i<j} A_{ij} \theta_{ij},$$
where $A_{ij}$ is the area of the triangle $(ij)$ and $\theta_{ij}$ is the dihedral angle between the tetrahedra $i$ and $j$ at the triangle $(ij)$. However, apart from the Euclidean case, the definition of the dihedral angle is nontrivial. For example, in the Lorentzian theory one need to consider separately the cases when the tetrahedra form a thin or thick wedge at the triangle, respectively, \cite{BarrettFoxon, Frank2} (see figure \ref{fig:dihedral}). 

In this paper we limit ourselves to the case where the faces are spacelike, but the normals to tetrahedra can be timelike or spacelike. We basically follow the definition of the dihedral angles in the relevant cases from \cite{Suarez}\footnote{Our sign is opposite to \cite{Suarez} and we subtract $\pi$ in the case of Euclidean and split signatures. It also differs from \cite{Brewin}. Our dihedral angles are always real opposed to \cite{Sorkin}. }. We argue that this is the right definition because the Schl\"afli identity holds \cite{Suarez, Rivin}
 \begin{equation}
  0=\sum_{i<j} A_{ij}\delta \theta_{ij},
 \end{equation}
where $\delta \theta_{ij}$ are variations of dihedral angles under changes of the shape of the $4$-simplex. This identity is crucial for Regge calculus \cite{Regge}.

The dihedral angles are defined as follows:
\begin{itemize}
 \item Euclidean $(----)$ or split signature $(++--)$ and $N_i^{can}=N_{j}^{can}$. The dihedral angle is a unique angle $\theta_{ij}$ such that
 $$
  \cos\theta_{ij}=t^{can'}N_i^{\geom}\cdot' N_j^{\geom},\quad t^{can'}\theta_{ij}\in(0,\pi)
 $$
 \item Lorentzian signature $(+---)$ and $N_i^{\geom}\cdot N_j^{\geom}>0$ and $N_i^{can}=N_{j}^{can}$. The dihedral angle $\theta_{ij}>0$ is the unique angle such that
$$
N_i^{\geom}\cdot N_j^{\geom} = \cosh\theta_{ij} 
$$
This case contains thick wedge for both normals timelike and thin wedge for both normals spacelike.
 \item Lorentzian signature $(+---)$ and $N_i^{\geom}\cdot N_j^{\geom}<0$ and $N_i^{can}=N_{j}^{can}$. The dihedral angle $\theta_{ij}<0$ is the unique angle such that
$$
N_i^{\geom}\cdot N_j^{\geom} = -\cosh\theta_{ij}
$$
This case contains thick wedge for both normals spacelike and thin wedge for both normals timelike.
\item Lorentzian signature $(+---)$ and $N_i^{can}\neq N_{j}^{can}$. The dihedral angle $\theta_{ij}$ is the unique angle such that
$$
N_i^{\geom}\cdot N_j^{\geom} = \sinh\theta_{ij}.
$$
\end{itemize}

\begin{figure}\label{fig:dihedral}
 \includegraphics[width=\textwidth]{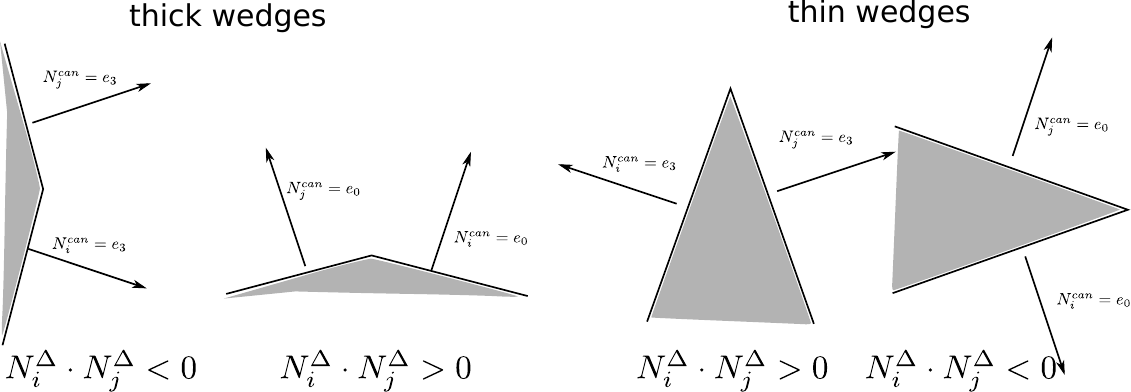}
 \caption{Thick and thin wedges.}
\end{figure}

The area of the triangle $(ij)$ is equal to
$$
A_{ij}=\frac{1}{2} |B_{ij}^{\geom}|.
$$
By the reconstruction theorem (Theorem \ref{thm:reconstr}) we know that 
$$
B_{ij}^{\geom}=\pm B_{ij}^{G}.
$$
Since $|B_{ij}^{G}|^2=\rho_{ij}^2$, it immediately follows that the Regge action for the geometries reconstructed from the geometric solution $\{G_i\}$  becomes
$$
S_{\Delta}=\frac{1}{2}\sum_{i<j} \rho_{ij} \theta_{ij}. 
$$

We define the geometric difference of the phase $\Delta S^{\geom}$ to be equal to
\begin{equation}
\Delta S^\geom:=2 i r S_{\Delta}
\left\{\begin{array}{ll}
           1 & \text{reconstructed simplex is Lorentzian,}\\
          \gamma^{-1}& \text{reconstructed simplex in other signature,}
          \end{array}\right.
\end{equation}
where $r$ is the Plebański orientation (definition \ref{df:r}).

\subsection{Revisiting the phase difference}

Let us recall that the difference of the phases between two critical points is given in terms of its $\SL$ solutions $\{g_i\}$ and $\{\tilde{g}_i\}$ by lemma \ref{lm:difference-phases} as
\begin{equation}
\Delta S=\tilde{S}-S=i\sum_{i<j} \rho_{ij} r_{ij}+2\j_{ij}\phi_{ij}
\end{equation}
where
\begin{equation}
 g_i \tilde{g}_i^{-1}\tilde{g}_j g_j^{-1}=e^{\phi_{ij}\M_{ij}-ir_{ij}\M_{ij}},
\end{equation}
where $\M_{ij}=\frac{2}{\rho_{ij}} \frac{2\gamma}{\gamma-i}(\B^\beta_{ij})^T$. Taking the image in $\SO_+(1,3)$ (see lemma \ref{lm:map}) we obtain
\begin{align}
 &G_i\tilde{G}_i^{-1}\tilde{G}_jG_j^{-1}=
 e^{2\phi_{ij}\frac{1}{\rho_{ij}}B_{ij}^\G-2r_{ij}\frac{1}{\rho_{ij}}{\Hodge}B_{ij}^\G},
\end{align}
as $\pi(\M_{ij})=\frac{2}{\rho_{ij}}B_{ij}^\G$.

Let us notice that $B_{ij}^\G\cdot B_{ij}^\G=\rho_{ij}^2$, and that $\frac{1}{\rho_{ij}}B_{ij}^\G$ generates rotation with period $2\pi$.
The difference of the contribution to the phase from the link action is
\begin{equation}
 i(\rho_{ij}r_{ij}+2\j_{ij}\phi_{ij}).
\end{equation}

\subsection{Determination of $r_{ij}$ and $\phi_{ij}$}\label{sec:phase_determination}

In this section we will determine $r_{ij}$ and $\phi_{ij}$ in a geometric way. The price we pay is an ambiguity of $\pi$ in the phase that will be fixed by a deformation argument in section \ref{sec:pi}.

\begin{lm}\label{lm-geom-phase-ndeg}
The contributions $\phi_{ij}$ and $r_{ij}$ to the difference in phase for two critical points corresponding to two non-degenerate $\SO(1,3)$ solutions $\{G_i\}$ and $\{\tilde{G}_i\}$ satisfy
\begin{equation}\label{eq:IRRpd}
  I^sR_{N_i^\G}R_{N_j^\G}=e^{-2r_{ij}\frac{1}{\rho_{ij}}{\Hodge}B_{ij}^\G+2\phi_{ij} \frac{1}{\rho_{ij}}B_{ij}^\G},
\end{equation}
where $s=0$ if both canonical normals are of the same type, and $s=1$ if they are of different types.
\end{lm}

\begin{proof}
We can compute
 \begin{align}
  G_i\tilde{G}_i^{-1}\tilde{G}_jG_j^{-1}&=G_i I^{r_i} R_{N_i^{can}} G_i^{-1} R_{e_0} R_{e_0} G_j R_{N_j^{can}} I^{r_j}G_j^{-1}=\nonumber\\
 &=I^{r_i+r_j} R_{N_i^\G}R_{N_j^\G},
 \end{align}
where $r_i\in\{0,1\}$ are such that
\begin{equation}
 \tilde{G}_i=R_{e_0}G_iR_{N_i^{can}} I^{r_i}\in \SO_+(1,3).
\end{equation}
We see that $s=r_i+r_j=0$ modulo $2$ if and only if both canonical normals are the same.
\end{proof}

We will now obtain a similar result for two degenerate $\SO(1,3)$ geometric solutions (vector geometries).

\begin{lm}\label{lm-translate}
Let $G_i,G_j\in \SO({M^4}')$ be such that
\begin{equation}
 \PHi^\pm(G_i)=G_i^\pm,\ \PHi^\pm(G_j)=G_j^\pm.
\end{equation}
Let $N_i^\G=G_i N^{can}$ and $N_j^\G=G_j N^{can}$. Then
\begin{equation}
 \PHi^\pm(R_{N_i^\G} R_{N_j^\G})=(G_i^\pm) (G_i^\mp)^{-1} (G_j^\mp)(G_j^\pm)^{-1}.
\end{equation}
\end{lm}

\begin{proof}
 From the equality \eqref{eq:RSO} we have
 \begin{align}
 &\PHi^\pm(G_i R_{N^{can}} G_i^{-1}  G_j R_{N^{can}} G_j^{-1})=\nonumber\\
&=\PHi^\pm(G_i (R_{N^{can}} G_i^{-1} R_{N^{can}}) (R_{N^{can}} G_j R_{N^{can}}) G_j^{-1})=\nonumber\\
&=(G_i^\pm) (G_i^\mp)^{-1} (G_j^\mp)(G_j^\pm)^{-1}.
 \end{align} 
That prove the equality.
\end{proof}

\begin{lm}\label{lm:B-split}
 Suppose that we have a bivector $B$ in $\Lambda^2{M^4}'$ then
 \begin{equation}\label{eq:eB}
  \PHi^\pm(e^B)=e^{B^\pm},
 \end{equation}
where $B^\pm\in so(V)\approx \Lambda^2 V$ are determined by the equality
\begin{equation}
\Phi^\pm(B^\pm)=\Phi^\pm(B)\in V,
\end{equation}
and we embed bivectors from $V$ into ${M^4}'$.
\end{lm}

\begin{proof}
We can decompose $B$ into its self-dual/anti-self-dual part, $B=B_++B_-$. Group elements generated by them commute, thus
\begin{equation}
 \PHi^\pm(e^B)=\PHi^\pm(e^{B_+})\PHi^\pm(e^{B_-})=\PHi^\pm(e^{B_\pm}).
\end{equation}
As $B_\pm$ are independent and $B^+$ is determined by $B_+$ ($B^-$ is determined by $B_-$),
it is enough to check what the image of $\PHi^\pm(e^{B_\pm})$ is in terms of $\Phi^\pm(B_\pm)=\Phi^\pm(B)$. 
Let us assume that $B$ preserves $N^{can}$, then by lemma \ref{lm:pm-G}
\begin{equation}
 e^{B^\pm}=e^B\in \SO(V),
\end{equation}
where we identify subgroups $\SO(V)$ from Minkowski and from ${M^4}'$. In this special case
\begin{equation}
t^{can'}( \pm B^{\pm}+{\Hodgp}B^\pm)\lrcornep N^{can}=\Phi^\pm(B)=\Phi^\pm(B_\pm),
\end{equation}
as $B^\pm$ preserves $N^{can}$. Let us notice that as $B_+=B_-$ is arbitrary in this case so we have
\begin{equation}
 t^{can'}(\pm B^{\pm}+{\Hodgp}B^\pm)\lrcornep N^{can}=\Phi^\pm(B_\pm)=\Phi^\pm(B)
\end{equation}
for arbitrary $B$.
\end{proof}

We have

\begin{lm}\label{lm-geom-phase-deg}
The contributions $\phi_{ij}$ and $r_{ij}$ for the difference in phase between two critical points corresponding to two vector geometries  $\{G_i^+\}$ and $\{G_i^-\}$ satisfies an identity written in terms of $\SO({M^4}')$ geometric solutions $G_i$ and  $\tilde{G}_i$ for flipped geometric boundary data, where
\begin{equation}
 \PHi^\pm(G_i)=G_i^\pm,\quad \tilde{G}_i=R_{N^{can}}G_iR_{N^{can}}.
\end{equation}
Namely
\begin{equation}\label{eq:RRdeg}
  R_{N_i^\G}R_{N_j^\G}=e^{2\phi_{ij} \frac{1}{\rho_{ij}}{\Hodgp}B_{ij}^\G},\quad r_{ij}=0
\end{equation}
where $B_{ij}^\G$ is the bivector from the $\SO({M^4}')$ geometric solution.
\end{lm}

\begin{proof}
From lemma \ref{lm-translate} and \ref{lm:difference-phases}  we know that (regarding $\SO(V)$ as subgroup of $\SO(1,3)$)
\begin{equation}
  \PHi^\pm(R_{N_i^\G}R_{N_j^\G})=(G_i^\pm) (G_i^\mp)^{-1} (G_j^\mp)(G_j^\pm)^{-1}=e^{\mp r_{ij}\frac{2}{\rho_{ij}}{\Hodge}B_{ij}^{\{G^\pm\}} \pm \phi_{ij} \frac{2}{\rho_{ij}}B_{ij}^{\{G^\pm\}}},
\end{equation}
as $+$ counts difference from $G^+$ to $G^-$ and $-$ in opposite direction.
 
If $r_{ij}\not=0$ then right hand side would not be in $\SO(V)\subset \SO(1,3)$. However every element $G_i^\pm$ belongs to this subgroup. Thus $r_{ij}=0$.
 
Regarding $B^{\{G^\pm\}}_{ij}$ as bivectors in ${M^4}'$, we have with use of theorem \ref{thm:deg-geometric}
\begin{equation}
 \Phi^\pm(\Hodgp B_{ij}^{\G})=\pm \Phi^\pm(B_{ij}^{\G})=\pm v_{ij}^{\{G^\pm\}}=\pm\Phi^\pm(B_{ij}^{\{G^\pm\}}).
\end{equation}
Thus by lemma \ref{lm:B-split} we have
\begin{equation}
 \PHi^\pm\left(e^{2\phi_{ij} \frac{1}{\rho_{ij}}{\Hodgp}B_{ij}^\G}\right)=e^{\pm 2\phi_{ij} 
 \frac{1}{\rho_{ij}}B_{ij}^{\{G^\pm\}}}.
\end{equation}
Comparing the images of two group elements under $\tilde{\Phi}$ we obtain their equality up to $I$
\begin{equation}
  I^sR_{N_i^\G}R_{N_j^\G}=e^{2\phi_{ij} \frac{1}{\rho_{ij}}{\Hodgp}B_{ij}^\G}.
\end{equation}
As ${\Hodgp}B_{ij}^\G$ is simple, we see that $s=0$.
\end{proof}

\subsection{Geometric difference of the phase modulo $\pi$}\label{sec:da}

We know that
\begin{equation}
 I^s R_{N_i^\G}R_{N_j^\G}=I^s R_{N_i^{\geom}}R_{N_j^{\geom}}
\end{equation}
This is a group element that appears in lemmas \ref{lm-geom-phase-ndeg} and \ref{lm-geom-phase-deg}.

Let us recall that ${\Hodgu}(N_i^{\geom}\wedge N_j^{\geom})$ is spacelike. We have (see appendix \ref{app:NN}):
\begin{itemize}
\item For $N_i^{can}=N_j^{can}$
\begin{equation}\label{eq:phasedifferenceSameN}
 R_{N_i^\G}R_{N_j^\G}=e^{2\T_j^{can}\T_i^{can'}\theta_{ij}\frac{N_i^{\geom}\wedge N_j^{\geom}}{|{\Hodgu}N_i^{\geom}\wedge N_j^{\geom}|}}.
\end{equation}
\item For $N_i^{can}\not= N_j^{can}$, (in Lorentzian signature)
\begin{equation}\label{eq:phasedifferenceLorentziandifferent}
 IR_{N_i^\G}R_{N_j^\G}=e^{2\T_j^{can}\T_i^{can}\theta_{ij}\frac{N_i^{\geom}\wedge N_j^{\geom}}{|\Hodgu N_i^{\geom}\wedge N_j^{\geom}|}+\pi{\Hodgu}\frac{N_i^{\geom}\wedge N_j^{\geom}}{|{\Hodgu}N_i^{\geom}\wedge N_j^{\geom}|}}.
\end{equation}
\end{itemize}
Moreover
\begin{equation}
 |B_{ij}^\G|^2=\rho_{ij}^2,
\end{equation}
so the rotation generated by either
\begin{equation}
 \frac{1}{\rho_{ij}}B_{ij}^\G,\quad\text{ or  }\quad
 \frac{1}{\rho_{ij}}{\Hodgp}B_{ij}^\G \text{ in flipped signature},
\end{equation}
has period $2\pi$.

From the geometric reconstruction theorem we know that
\begin{equation}
 B_{ij}^\G=-\frac{1}{Vol^{\geom}}r W_i^{\geom}W_j^{\geom}{\Hodgu}(N_j^{\geom}\wedge N_i^{\geom}),
\end{equation}
and the sine law shows that $\left|\frac{1}{Vol^{\geom}}r W_i^{\geom}W_j^{\geom}\right||{\Hodgu}N_j^{\geom}\wedge N_i^{\geom}|=\rho_{ij}$. So
\begin{equation}
 \frac{1}{\rho_{ij}}B_{ij}^\G=\sigma_{ij}\frac{1}{|{\Hodgu}N_i^{\geom}\wedge N_j^{\geom}|}{\Hodgu}(N_j^{\geom}\wedge N_i^{\geom}),
\end{equation}
where
\begin{equation}
 \sigma_{ij}=-r\ \text{sign}\ W_i^{\geom}W_j^{\geom}.
\end{equation}
Let us notice (see \eqref{eq:N-geom}) that
\begin{equation}
 \sigma_{ij}=-r \T^{can}_i\T^{can}_j=\left\{\begin{array}{ll}
                      -r & N_i^{can}=N_j^{can}\\
                      r & N_i^{can}\not=N_j^{can}
                     \end{array}\right.
\end{equation}
By comparing \eqref{eq:phasedifferenceSameN}, \eqref{eq:phasedifferenceLorentziandifferent} with \eqref{eq:IRRpd} and \eqref{eq:RRdeg} we get the following
\begin{itemize}
 \item For Lorentzian signature (non-degenerate solutions) and normals of the same type
 \begin{equation}
  2\phi_{ij}=0\text{ mod } 2\pi,\quad 2r_{ij}=2 r\theta_{ij}
 \end{equation}
\item For Lorentzian signature (non-degenerate solutions) and normals of different types
 \begin{equation}\label{eq:pi2}
  2\phi_{ij}=\pi\text{ mod } 2\pi,\quad 2r_{ij}=2r\theta_{ij}
 \end{equation}
\item For other signature solutions (degenerate solutions)
 \begin{equation}
  2\phi_{ij}=2r\theta_{ij}\text{ mod } 2\pi,\quad 2r_{ij}=0
 \end{equation}
\end{itemize}

We will now consider contributions of $\frac{\pi}{2}$ appearing in the action from \eqref{eq:pi2}.

\begin{lm}\label{lm:pi-two}
 The following holds for given boundary data:
 \begin{equation}
  \sum_{\substack{i<j\colon 2j_{ij}\text{odd}\\N_i^{can}\not=N_j^{can}}} s_{ij}\frac{\pi}{2}\in \pi\Z,\qquad
  \sum_{\substack{i<j\\ N_i^{can}\not=N_j^{can}}} s_{ij}2\j_{ij}\frac{\pi}{2}\in \pi\Z,
 \end{equation}
where $s_{ij}\in\{-1,1\}$ are arbitrary.
\end{lm}

\begin{proof}
Let us consider the sub-graph of the spin network consisting of those links that have half-integer spins. As this is a graph with all the vertices of even valence, there exists Euler cycles, i.e., cycles such that every link of half-integer spin belongs to the cycle exactly once. 

 Let us count number of changes of the type of normals from consecutive vertices on one of those cycles. As it is a cycle, the number is even. However as we go through all links exactly once this is the number of links $ij$ with $N_i^{can}\not=N_j^{can}$. Hence
 \begin{equation}
   \sum_{\substack{i<j\colon 2j_{ij}\text{odd}\\N_i^{can}\not=N_j^{can}}} s_{ij}\frac{\pi}{2}\in \pi\Z,
 \end{equation}
 as the sum is over even number of $\frac{\pi}{2}$.
 Let us denote by $[j]$ the largest integer smaller than $j$. Definitely
 \begin{equation}
  \sum_{\substack{i<j\\ N_i^{can}\not=N_j^{can}}} s_{ij}2[\j_{ij}]\frac{\pi}{2}\in \pi\Z.
 \end{equation}
Summing the two equalities, we obtain the desired result.
\end{proof}

As we determined $S$ modulo $\pi$ we can skip the $\frac{\pi}{2}$ terms coming from \eqref{eq:pi2} by lemma \ref{lm:pi-two} and write
\begin{equation}
 \Delta S=\Delta S^{\geom} \text{ mod } \pi i,
\end{equation}
where
\begin{equation}
\Delta S^{\geom}=ir\left\{\begin{array}{ll}
                  \sum_{i<j}\rho_{ij}\theta_{ij} & \text{ Lorentzian signature,}\\
                  \sum_{i<j}2\j_{ij}\theta_{ij} & \text{ other signatures.}
                 \end{array}\right.
\end{equation}

\subsection{Deformation argument to fix the remaining ambiguity}\label{sec:pi}

We only need to determine the remaining ambiguity of $\pi$ in the action. This however depends on the choice of $\SL$ lifts, and it needs to be done consistently for the whole $4$-simplex. Taking into account that the only source of ambiguity are the contributions $\phi_{ij}$ for half-integer spins, we have
\begin{equation}
 \Delta S-\Delta S^{\geom}=i\sum_{i<j\colon 2\j_{ij}\text{ odd}}\begin{cases}	
\phi_{ij}& \text{(non-deg.\ Lorentzian solutions)}\\
\phi_{ij}-r\theta_{ij}&\text{(two degenerate solutions)}
 \end{cases}
\end{equation}

\subsubsection{Lorentzian signature case}

\begin{lm}\label{lm:def-nondeg}
Suppose that we have non-degenerate geometric boundary data $v_{ij}^0$ and $N_i^{can0}$ with spins $\j_{ij}^0$ with non-degenerate
 $\SO(1,3)$ geometric solutions $\{G_i^0\}$. Let us assume that there exists a continuous path
 \begin{equation}
  G_i(t),\quad v_{ij}(t),\quad N_i^{can}(t)=N_i^{can0},
 \end{equation}
 such that:
 \begin{itemize}
  \item For all $i\not= j$
  \begin{equation}
   G_i^0=G_i(0),\quad v_{ij}^0=v_{ij}(0),
  \end{equation}
  \item for all $t\in[0,1]$, $\{G_i(t)\}$ is a $\SO(1,3)$ geometric solution for $v_{ij}(t)$ boundary data,
  \item for all $t\not=1$ the boundary data $v_{ij}(t)$ is non-degenerate,
  \item for all $i\not= j$ $v_{ij}(1)\not=0$,
  \item for all $t\not=1$ solution $\{G_i(t)\}$ is non-degenerate, and 
  \item for $t=1$ solution $\{G_i(t)\}$ and $\{\tilde{G}_i(t)=R_{e_\alpha} G_i(t) R_{N^{can}_i}\}$ are gauge equivalent.
 \end{itemize}
Then
\begin{equation}
 \Delta S^0=\Delta S^{\geom0}\text{ mod } 2\pi i.
\end{equation}
\end{lm}

\begin{proof}
 The function
 \begin{equation}
  f(t)=\sum_{i<j\colon 2\j_{ij}^0\text{ odd}}\phi_{ij}(t)\text{ mod } 2\pi
 \end{equation}
takes values in $\{0,\pi\}$  by lemma \ref{lm:pi-two} and is changing continuously if we compute differences between two solutions $\{G_i(t)\}$ and $\{\tilde{G}_i(t)=R_{e_\alpha} G_i(t) R_{N^{can}_i}\}$. Its value needs to be constant.  
We need to show $\lim_{t\rightarrow 1} f(t)=0$. As in the limit $\{G_i(1)\}$ and $\{\tilde{G}_i(1)\}$ are gauge equivalent
thus there exists $G\in \SO(1,3)$ such that $\tilde{G}_i=GG_i$ for all $i$.

Let us introduce lifts $g_i$ and $\tilde{g}_i$ and $g$ of these elements. There exists $s_i\in\{0,1\}$ such that
\begin{equation}
 \tilde{g}_i=(-1)^{s_i} gg_i,
\end{equation}
and then
\begin{equation}
 (-1)^{s_i+s_j}=g_i \tilde{g}_i^{-1}\tilde{g}_j g_j^{-1}=e^{\phi_{ij}\M_{ij}-ir_{ij}\M_{ij}}.
\end{equation}
Thus from $\M_{ij}\not=0$ it follows that $\phi_{ij}(1)=(s_i+s_j)\pi$ mod $2\pi$. We have
\begin{equation}
 \Delta S^0-\Delta S^{\geom0}=i\sum_{i<j\colon 2\j_{ij}^0\text{ odd}}(s_i+s_j)\pi\text{ mod }2\pi i%\qedhere
\end{equation}
We can consider an Euler cycle in the subgraph consisting of edges with odd spin. For such a cycle 
\begin{equation}
 \sum_{i<j\colon (ij)\in \text{cycle}}(s_i+s_j)\pi=\sum_{i\in \text{cycle}} 2s_i\pi=0\text{ mod } 2\pi.
\end{equation}
As the subgraph of half-integer spins has even-valent nodes, we can decompose it into Euler cycles thus $\Delta S^0=\Delta S^{\geom0}$ mod $2\pi i$.
\end{proof}

Let us deform our boundary data by deforming solution as follows: We choose a spacelike plane described by a simple, normalized bivector $V$ in generic position i.e.
\begin{equation}\label{eq:generic}
 \forall_{i\not=j} V\wedge {\Hodge}B_{ij}^\G\not=0.
\end{equation}
We now contract directions in ${\Hodge}V$ (perpendicular to directions in $V$). All the time $B_{ij}^{\{G(t)\}}\not=0$ and the solution $\{G_i(t)\}$ obtained by reconstruction from the bivectors is non-degenerate. Let us now consider the limit where we shrink $*V$ to zero. We will denote this time as $t=1$. Due to \eqref{eq:generic}, the limits $\lim_{t\rightarrow 1}B_{ij}^{\{G(t)\}}$ exist and are nonzero. The shrinking has a dual action on the geometric normals (their directions in $*V$ expands but we need to apply normalization).
As the geometric normal vectors $N_i^{\geom}(t)$ do not lie in the $V$ plane (see \eqref{eq:generic}), they also have a limit and it is equal to their normalized components lying in the plane $\Hodge V$. By suitable definition of $v_{ij}(t)$ we can assume that the limits $G_i(1)=\lim_{t\rightarrow 1} G_i^{\geom}(t)$ exist.
The $4$-simplex is now highly degenerate, contained in a $2d$ plane. All bivectors are proportional to $V$. 

So we have for any non-degenerate boundary data by lemma \ref{lm:def-nondeg}
\begin{equation}
\Delta S=\Delta S^{\geom} \text{ mod } 2\pi i.
\end{equation}

\subsubsection{The case of other signatures}

The difference $\Delta S$ is well-defined if we fix between which two degenerate points we need to compute the difference of the phase.

\begin{lm}\label{lm:def-deg}
Let us consider a set of non-degenerate geometric boundary data, with spins $\j_{ij}^0$ and all canonical normals $N_i^{can0}=N^{can}$.
Let us suppose that there are two non-equivalent vector geometries for this boundary data. We associate with them a non-degenerate $\SO({M^4}')$ geometric solution $\{G_i^0\}$ with boundary data $v_{ij}^0$. Let us suppose that there exists a continuous path
 \begin{equation}
  G_i(t),\quad v_{ij}(t),\quad N_i^{can}(t)=N^{can},
 \end{equation}
 such that 
 \begin{itemize}
  \item one has
  \begin{equation}
   G_i^0=G_i(0),\quad v_{ij}^0=v_{ij}(0),\quad k\in\{0,1\},
  \end{equation}
  \item for all $t\in[0,1]$, $\{G_i(t)\}$ is a $\SO({M^4}')$ geometric solution for $v_{ij}(t)$ boundary data,
  \item for all $t\in[0,1]$ the boundary data $v_{ij}(t)$ is non-degenerate,
  \item for $t\in[0,1)$ the $\SO({M^4}')$ geometric solution $\{G_i(t)\}$ is non-degenerate,
  \item $\{G_i(1)\}$ is a degenerate $\SO({M^4}')$ geometric solution.
 \end{itemize}
Then
\begin{equation}
\Delta S^0=\Delta S^{\geom0} \text{ mod } 2\pi i.
\end{equation}
\end{lm}

\begin{proof}
 The function
 \begin{equation}
  f(t)=\sum_{i<j\colon 2j_{ij}^0\text{ odd}}\phi_{ij}(t)-r\theta_{ij}(t)\text{ mod } 2\pi
 \end{equation}
takes values in $\{0,\pi\}$ and is changing continuously if we compute differences between two critical points determined by vector geometries $\{G^\pm_i(t)\}$ obtained from a $\SO({M^4}')$ geometric solution $\{G_i(t)\}$. Thus it is constant. For simplicity let us work in a suitable gauge, such that
\begin{equation}
 G_i^+(1)=G_i^{-}(1)\in \SO_+(1,3),\quad N_i^{\{G(1)\}}=(-1)^{r_i}N^{can},\quad r_i\in\{0,1\}.
\end{equation}
We have for the lifts  $g_i^+(1)=(-1)^{s_i} g_i^{-}(1)$. Similar considerations using Euler cycles as for Lorentzian signature show that
\begin{equation}
 \sum_{i<j\colon 2\j_{ij}^0\text{ odd}}\phi_{ij}(1)=0\text{ mod } 2\pi.
\end{equation}
Also, as
\begin{equation}
 \theta_{ij}(1)=(r_i+r_j)\pi \text{ mod } 2\pi, 
\end{equation}
an argument with an Euler cycle shows that
\begin{equation}
 \sum_{i<j\colon 2\j_{ij}^0\text{ odd}}\theta_{ij}(1)=0\text{ mod } 2\pi.
\end{equation}
Thus $f(1)=0$ and
\begin{equation}
\Delta S^0-\Delta S^{\geom0}=if(0)=if(1)=0 \text{ mod } 2\pi i.%\qedhere
\end{equation}
\end{proof}

Let us deform the boundary data as follows: We choose $N$ of the same type as $N^{can}$, in generic position, i.e., such that
\begin{equation}
 \forall_{ij}\ N\wedge B_{ij}^\G\not=0.
\end{equation}
We contract in the direction of $N$ in ${M^4}'$. During contraction we have a continuous path of non-degenerate $\SO({M^4}')$ geometric solutions (with non-degenerate boundary data). At the end we obtain a degenerate $4$-simplex for non-degenerate boundary data. By lemma \ref{lm:def-deg} we have
\begin{equation}
\Delta S=\Delta S^{\geom}\text{ mod } 2\pi i.
\end{equation}

\subsection{Summary}\label{sub:summary}

Let us now choose all outer pointing normals $N_i^{\geom}$.
We obtained the following formula for the difference of the phase between two critical points
\begin{equation}
\Delta S=ir\sum_{i<j}\rho_{ij}\theta_{ij}
\left\{\begin{array}{ll}
           1 & \text{reconstructed simplex is Lorentzian,}\\
          \gamma^{-1}& \text{reconstructed simplex in other signature,}
          \end{array}\right.
\end{equation}
where  $r$ is the Plebanski orientation (definition \ref{df:r}), and $\theta_{ij}$ are generalized dihedral angles (see \ref{sec:da}) to be reconstruct as follows:
\begin{equation}
 R_{N_i^{\geom}}R_{N_j^{\geom}}=O e^{\frac{2 \T_j^{can}\T_i^{can}\theta_{ij}}{|\Hodgu N_j^{\geom}\wedge N_i^{\geom}|}N_i^{\geom}\wedge N_j^{\geom}}
\end{equation}
where
\begin{equation}
 O=\left\{\begin{array}{ll}
           \text{ inversion in the plane } N_i^{\geom}\wedge N_j^{\geom} & \text{ for }N_i^{can}\not=N_j^{can},\\
           \text{ identity } & \text{ for } N_i^{can}=N_j^{can}.
          \end{array}\right.
\end{equation}

%% file: chapter-Hessian.tex
%!TEX root = Spin_foams_with_time_like_tetrahedra.tex
\section{Computation of the Hessian}

In this chapter we will finish our analysis by computing the scaling property of the measure factor in the formal application of stationary phase approximation. We assume that after taking into account gauge transformations the remaining Hessian is non-degenerate. The contribution to the stationary phase approximation of the amplitude $A$ from the Hessian is
\begin{equation}
 \sqrt{\det\mathcal H}^{-1/2}.
\end{equation}
The integrand depends on the following variables
\begin{itemize}
 \item $\z_{ij}$ for $i\not=j$, giving $80$ real variables,
 \item $g_i\in \SL$, giving $30$ real variables. 
\end{itemize}
But these variables are subject to gauge transformations, which reduces the effective number of nontrivial variables: 
\begin{itemize}
 \item $\z_{ij}\rightarrow \lambda_{ij}\z_{ij}$, giving $40$ real gauge parameters,
 \item $\SL$ gauge gives $6$ real gauge parameters.
\end{itemize}
This gives $64$ nontrivial variables, and the scaling
\begin{equation}
 \sqrt{\det\mathcal H}^{-1/2}=C\Lambda^{-32}.
\end{equation}
Together with the scaling $\Lambda^{20}$ (see \eqref{eq:c-scaling}) of $c(\Lambda)$ this gives the scaling of the measure factor,
\begin{equation}
 \Lambda^{-32}\Lambda^{20}=\Lambda^{-12}.
\end{equation}
Let us notice that $\mathcal H$ does not depend on the choice of the phase of the coherent states (phase of $\n_{ij}$). Its determinant can be expressed by geometric quantities like $v_{ij}$ and $N_i^{\G}$, $B_{ij}^\G$ in a covariant way. It is thus a geometric quantity depending on lengths and orientations.

%% file: chapter-Conclusions.tex
%!TEX root = Spin_foams_with_time_like_tetrahedra.tex
\section{Outlook}
In this paper we performed the complete analysis of the critical points in the extended EPRL setting. 
Besides the important open questions discussed in section \ref{sec:conjectures}, there are several points which merit further analysis:
 \begin{itemize}
 \item Extension of our results to the case of surfaces of mixed signature. A spin foam model for surfaces of this kind was introduced in \cite{Conrady2010b} based on coherent state techniques. There is no known EPRL-like construction. We suspect that there might be some nongeometric contributions to the asymptotic formula in this case.
 \item The determinant of the Hessian is a complex number, thus it contributes to the phase. It is an important part of the asymptotics and can be described in geometric terms if the boundary state is geometric \cite{PR, Roberts, Kaminski2013}. The spread of the coherent states has no obvious geometric meaning but it influences the measure factors. In the Euclidean EPRL model the values of the Hessian at two different critical points are equal thus, their phase can be cancelled by proper choice of the phase of coherent states. It is tempting to conjecture that similar statement is true for the Lorentzian models. However nothing of this kind was proven even in the standard EPRL setup.
 \item A determination of the absolute value of the determinant of the Hessian is interesting for the following reason. In order to evaluate the spin foam amplitudes of more complicated triangulations, one needs to take a product over many vertex amplitudes. In the semiclassical limit one can hope that one can replace the vertex amplitudes in the product by their asymptotic forms.  In such a situation, the phase is exactly the Regge action \cite{Regge} of discrete gravity (with some problems regarding orientations \cite{Engle2015}) and the measure of the path integral is obtained from the absolute value of the determinant of the Hessian. However, it seems that in the case of current spin foam models the amplitude of the whole foam cannot be obtained by this approximation \cite{Hellmann2013, Hellmann-Kaminskishort} (but see also \cite{Han2013} for possible resolutions).
 \item Another open problem is the extension of our result to the case of non-vanishing cosmological constant. The asymptotic analysis of a corresponding version of the EPRL model was given in \cite{Haggard2014,Haggard-tetrahedra}. However, moving away from Euclidean signature even the formulation of the model becomes very formal. There is no satisfactory proposal for the intertwiners representing timelike tetrahedra for the situation with cosmological constant (see however \cite{Han2010, Fairbairn2010} for possible, alternative rigorous deformations). 
\end{itemize}

\section*{Acknowledgements}
We would like to thank Bianca Dittrich, Sebastian Steinhaus and Hal Haggard for interactions at the early stage of this work. We would like to express our gratitude for Jeff Hnybida for fruitful discussions about the topic. We also thank the anonymous referees for their valuable comments on a previous version of this work. This article is partially based upon the work from COST Action MP1405 QSPACE, supported by COST (European Cooperation in Science and Technology). WK thanks the Institute for Quantum Gravity at the Friedrich-Alexander Universit\"at Erlangen-N\"urnberg (FAU) for hospitality. This work was also partially supported by the Polish National Science Centre grant No. 2011/02/A/ST2/00300. 

%% file: Appendix.tex
%!TEX root = Spin_foams_with_time_like_tetrahedra.tex
\appendix
\section*{Appendices}
\addcontentsline{toc}{section}{Appendices}
\renewcommand{\thesubsection}{\Alph{subsection}}

\subsection{Notation summary}

Notation summary 
\begin{itemize}
 \item $i,j,k$ nodes (tetrahedra) numbers (indices),
 \item spacetime indices $\mu,\nu$,
 \item $\v,\u,\r,\z$ spinors, $\z_{ij}$ spinors labelled by pairs of indices (tetrahedra), $\n_{ij}$ boundary data spinors,
 \item $\B,\M$  traceless matrices $sl(2,\C)$, $\B_{ij}$ traceless matrices in $sl(2,\C)$ labelled by two tetrahedra,
 \item $\eta_N=N^\mu\sigma_\mu$, $\hat{\eta}_N=N^\mu\hat{\sigma}_\mu$,
 \item $\langle \u,\v\rangle_N=\u^\dagger \eta_N^T\v$ Hermitian scalar product defined by normal $N$,
 \item $[\u,\v]=\u^T\omega \v$,
 \item $\omega$ symplectic form for spinors,
 \item $N$ normal vectors, $N_i^\geom$ outer pointing normal vector to $i$-th tetrahedron, $N_i^{can}$ canonical normal vector either $e_0=(1,0,0,0)$ or $e_3=(0,0,0,1)$, $N_i^\G$ normal vector to $i$-th tetrahedron obtained from geometric solution,
 \item $v,l$ vectors ($l$ null vector), $v_{ij}$ boundary data vectors,
 \item $\cdot$ scalar product,
 \item $R_N$ reflection with respect to normal $N$,
 \item $I$ inversion,
 \item $g\in \SL$ group elements, 
 \item $B_{ij}$ bivectors, $B_{ij}^{\G}$ bivectors from the geometric solution, $B_{ij}^{\geom}$ geometric bivectors from the simplex,
 \item $M^4$ Minkowski spacetime,
 \item ${M^4}'$ spacetime with flipped $N^{can}$ (split or euclidean signature),
 \item $r$ Plebanski orientation (relating orientation of $M^4$ or ${M^4}'$ to the combinatoric orientation of the simplex),
 \item $G\in \OO(M^4)$ or $\OO({M^4}')$ ($\SO(M^4)$, $\SO({M^4}')$),
 \item we denote $\SO_+(M^4)=\SO_+(1,3)$ connected component of $\SO(M^4)=\SO(1,3)$,
 \item $V$ space perpendicular to $N^{can}$ in Minkowski (also embeddable in ${M^4}'$),
 \item $\Phi^\pm$ maps from bivectors in ${M^4}'$ into $V$ (self-dual and anti-self-dual forms), $\PHi^\pm$ corresponding maps from $\SO({M^4}')$ into $\SO(V)$,
 \item Hodge star ${\Hodge}$,
 \item if we are working explicitly in the flipped spacetime then we use $\Hodgp$, $\cdop$ and for contraction with use of the metric $\llcornep$ and $\lrcornep$,
 \item if we are working  in arbitrary signature spacetime then we use $\Hodgu$, $\cdou$, $\llcorneu$ and $\lrcorneu$,
 \item representation labels ($\jj$, $\rho$) for $\SL$ group, $\j_{ij}$ and $\rho_{ij}$ representation labels for edge connecting tetrahedron $i$ with $j$,
 \item $\lambda$ number,
 \item $\C^*$ invertible complex numbers, $\R^*$ invertible real numbers
 \item $\Lambda$ integer scaling of spins,
 \item $\theta_{ij}$ dihedral angle between tetrahedron $i$ and $j$,
 \item $A_{ij}$ area of the face between tetrahedron $i$ and $j$,
 \item $S$ action, $S_{ij}$, $S_{ij}^\beta$ parts of the action.
\end{itemize}

\subsection{Conventions}
\label{sec:conventions}

We are using abstract definition of Grassmann algebra $\bigwedge X$ and $\wedge$ by its universal properties. For $b\in X^*$ we define left ($\llcorner$) and right ($\lrcorner$) antiderivatives of order $-1$ 
\begin{equation}
 b\llcorner w,\quad w\lrcorner b,\quad w\in \bigwedge X
\end{equation}
by its action on $X\subset \bigwedge X$
\begin{equation}
 b\llcorner x= x\lrcorner b=(b,x),\quad x\in X
\end{equation}
Suppose that we have a metric $\underline{g}_{\mu\nu}$.
The norm of $k$ vectors is given by
\begin{equation}
 a_1\wedge \ldots \wedge a_k\cdou b_1\wedge \ldots \wedge b_k=\sum_{\sigma\in S_k} (-1)^{\sgn\sigma} \prod_{i=1}^k a_i\cdou b_{\sigma(i)}
\end{equation}
The Hodge dual is defined by
\begin{equation}
 \Hodgu(a_1\wedge \ldots \wedge a_k)=a_k\llcorneu a_{k-1}\llcorneu \cdots a_1\llcorneu \Omega
\end{equation}
where $\Omega$ is the normalized volume $n$-vector (choice of the orientation) and $\llcorneu$ is the contraction with dualized vector by the scalar product.

\subsection{Restriction of representations of $\SL$}\label{app:decomp}

In this appendix we will collect results of \cite{Bargmann} and \cite{Conrady2010} about restriction of the irreducible unitary representations of $\SL$ to the subgroup $\St(N^{can})$. We introduce spinors
\begin{equation}
 \n_0=\bvec{1}{0},\quad \n_1=\bvec{0}{1}
\end{equation}
Let us consider generator of rotation around $z$ axis. We can introduce bases of eigenfunction in the representations of $\SU(1,1)$ and $\SU(2)$.
On the groups $\St(N^{can})$ we consider right regular representations. 

\subsubsection{The embedding map in the case of $\St(N^{can})=\SU(2)$}

In this case following \cite{EPRL} we consider embedding of the spin $\j$ representation. We can realize this representation as functions on $\SU(2)$ given by matrix elements of the representation in $L_z$ eigenbasis. The embedding of the functions
$$
\Psi_{\j m}(u)=\sqrt{2\j+1}D^{\j}_{\j m}(u),\quad u\in \SU(2)
$$
into the representation space of unitary irreducible representation $(\j,\rho=2\gamma \j)$ is given by
$$
F_{\j m}(\z)=\frac{1}{\sqrt{\pi}}\braket{\z}{\z}_{{N^{can}}}^{\iu \frac{\rho}{2}-1} \Psi_{\j m}(u(\z)),
$$
where ${N^{can}}=e_0$ and
$$
u(\z)=\frac{1}{\sqrt{\braket{\z}{\z}_{{N^{can}}}}}
\begin{pmatrix}
z_0& z_1\\ -\overline{z}_1 & \overline{z}_0
\end{pmatrix},\quad \z=\bvec{z_0}{z_1}
$$
and 
$$
\braket{\u}{\v}_{{N^{can}}}=\overline{u}_0 v_0 + \overline{u}_1 v_1.
$$
Using the explicit expression for the representation matrices of the $\SU(2)$ group we write $F$ explicitly: 
$$
F_{\j m}(\z)=\sqrt{\frac{\Gamma(2\j+2)}{2\pi\Gamma(\j+m+1) \Gamma(\j-m+1)}}\braket{\z}{\z}_{N^{can}}^{\iu \frac{\rho}{2}-1-\j}  \braket{\n_0}{\z}_{N^{can}}^{\j+m} \braket{\n_1}{\z}_{N^{can}}^{\j-m},
$$

\subsubsection{The embedding map in the case of $\St({N^{can}})=\SU(1,1)$}

Following Hnybida and Conrady \cite{Conrady2010, Conrady2010a} we consider embeddings of the basis functions of the $\pm$ discrete series unitary irreducible representations of spin $\j$ of the $\SU(1,1)$ group. Again we can realize them as a functions on the group
\begin{align}
\Psi^+_{\j m}(v)&=\sqrt{2\j-1} D^{+,\j}_{\j m}(v),\quad m\geq \j\\
\Psi^-_{\j m}(v)&=\sqrt{2\j-1} D^{-,\j}_{-\j m}(v),\quad m\leq -\j
\end{align}
Their image in the representation space of unitary irreducible representation $(\j,\rho=2\gamma \j)$ is given by
$$
F^{\tau}_{jm}(\z)=\frac{1}{\sqrt{\pi}}\theta(\tau\braket{\z}{\z}_{{N^{can}}})\braket{\z}{\z}_{{N^{can}}}^{\iu \frac{\rho}{2}-1} \Psi^\tau_{\j m}(v^\tau(\z)),
$$
where $\tau=\pm 1$, ${N^{can}}=e_3$,
\begin{align}
v^+(\z)&=\frac{1}{\sqrt{\braket{\z}{\z}}_{{N^{can}}}}\begin{pmatrix}z_0& z_1\\ \overline{z}_1 & \overline{z}_0\end{pmatrix},\quad \z=\bvec{z_0}{z_1} \\
v^-(\z)&=\frac{1}{\sqrt{-\braket{\z}{\z}_{{N^{can}}}}}\begin{pmatrix}\overline{z}_1& \overline{z}_0\\ z_0 & z_1\end{pmatrix},
\end{align}
and
$$
\braket{\u}{\v}_{{N^{can}}}:= \overline{u}_0 v_0 - \overline{u}_1 v_1
$$
is the $\SU(1,1)$-invariant scalar product. 
Using the explicit expression for the representation matrices of the $\SU(1,1)$ group \cite{Bargmann} we write the distributions $F^\pm$ explicitly: 
\begin{align}
 F^{\tau}_{\j m}(\z)=&\sqrt{\frac{\Gamma(\tau m+\j)}{\pi\Gamma(2\j-1) \Gamma(\tau m-\j+1)}}\theta(\tau\braket{\z}{\z}_{N^{can}})\nonumber\\
 &\braket{\z}{\z}_{N^{can}}^{\iu \frac{\rho}{2}-1+\j}  \braket{\z}{\n_0}_{N^{can}}^{-\j-m} (\tau \braket{\z}{\n_1}_{N^{can}})^{-\j+m},
\end{align}

\subsubsection{Coherent states}

All three families of representations have certain common feature. In the $L_z$ eigenbasis
\begin{align}
\begin{array}{lll}
 \Psi_{\j m},& m\in\{-\j,-\j+1,\ldots \j\} & \text{ for } \SU(2)\\
 \Psi^+_{\j m},& m\in\{\j,\j+1,\ldots \} & \text{ for } \SU(1,1)\\
 \Psi^-_{\j m},& m\in\{\ldots,-\j-1,-\j\} & \text{ for } \SU(1,1)
\end{array}
\end{align}
we can consider extremal eigenfunctions. These are basic coherent states \cite{Perelomov}. They usefulness for asymptotic analysis comes from their simple form (see \cite{Bargmann, Conrady2010a})
\begin{itemize}
 \item For ${N^{can}}=e_0$
 \begin{equation}
 F_{\j\j}(\z)=\sqrt{\frac{2\j+1}{2\pi}} \braket{\z}{\z}_{{N^{can}}}^{\iu \frac{\rho}{2}-1-\j} \braket{\n_0}{\z}_{{N^{can}}}^{2\j}  
 \end{equation}
 \item For ${N^{can}}=e_3$
 \begin{align}
 F^{+}_{\j\j}(\z)&=\sqrt{\frac{2\j-1}{\pi}}\theta(\braket{\z}{\z}_{{N^{can}}}) \braket{\z}{\z}_{{N^{can}}}^{\iu \frac{\rho}{2}-1+\j} \braket{\z}{\n_0}_{{N^{can}}}^{-2\j},\\
F^{-}_{\j-\j}(\z)&=\sqrt{\frac{2\j-1}{\pi}}\theta(-\braket{\z}{\z}_{{N^{can}}})(-\braket{\z}{\z}_{{N^{can}}})^{\iu \frac{\rho}{2}-1+\j} (-\braket{\z}{\n_1}_{{N^{can}}})^{-2\j}
 \end{align}
\end{itemize}
where $\n_0=\bvec{1}{0}$ and $\n_1=\bvec{0}{1}$.

All other coherent states are obtained through transformation of these basic coherent states by group action of $\St({N^{can}})$. Let us notice that the group action preserves $\braket{\cdot}{\cdot}_{N^{can}}$ and we can move group elements from $\z$ into $\n_i$ obtaining
\begin{align}
 \Psi^{\n}(\z)&=F_{\j\j}(u^T\z),&\quad \n=(u^T)^{-1}\n_0,\\
 \Psi^{+,\n^+}(\z)&=F^+_{\j\j}(u^T\z),&\quad \n^+=(u^T)^{-1}\n_0,\\
 \Psi^{-,\n^-}(\z)&=F^-_{\j-\j}(u^T\z),&\quad \n^-=(u^T)^{-1}\n_1
\end{align}
We have the following classification:
\begin{itemize}
 \item Every spinor $\n_+$ such that $\langle \n_+,\n_+\rangle_{{N^{can}}}=1$ with ${N^{can}}=e_3$ is obtained as $\n_+=(u^T)^{-1}\n_0$ for $u\in \SU(1,1)$
 \item Every spinor $\n_-$ such that $\langle \n_-,\n_-\rangle_{{N^{can}}}=-1$ with ${N^{can}}=e_3$ is obtained as $\n_-=(u^T)^{-1}\n_1$ for $u\in \SU(1,1)$
 \item Every spinor $\n$ such that $\langle \n,\n\rangle_{{N^{can}}}=1$ with ${N^{can}}=e_0$ is obtained as $\n=(u^T)^{-1}\n_0$ for $u\in \SU(2)$
\end{itemize}
This allow us to write coherent states as follows
\begin{align}
\Psi^{+,\n^+}(\z)&=\theta(\braket{\z}{\z}_{{N^{can}}})\sqrt{\frac{2\j-1}{\pi}}\braket{\z}{\z}_{{N^{can}}}^{\iu \frac{\rho}{2}-1+\j}\braket{\z}{\n^+}_{{N^{can}}}^{-2\j},\\
\Psi^{-,\n^-}(\z)&=\theta(-\braket{\z}{\z}_{{N^{can}}})\sqrt{\frac{2\j-1}{\pi}}(-\braket{\z}{\z}_{{N^{can}}})^{\iu \frac{\rho}{2}-1+\j}(-\braket{\z}{\n^-}_{{N^{can}}})^{-2j}. 
\end{align}
with ${N^{can}}=e_3$ and arbitrary
\begin{equation}
 \langle \n^+,\n^+\rangle_{{N^{can}}}=1,\quad \langle \n^-,\n^-\rangle_{{N^{can}}}=-1
\end{equation}
and for $\SU(2)$ embedding
$$
\Psi^{\n}(\z)=\sqrt{\frac{2\j+1}{2\pi}}\braket{\z}{\z}_{{N^{can}}}^{\iu \frac{\rho}{2}-1-\j}\braket{\n}{\z}_{{N^{can}}}^{2\j}.
$$
with ${N^{can}}=e_0$ and arbitrary $\langle \n,\n\rangle_{{N^{can}}}=1$.

\subsection{Traceless matrices}\label{app:traceless}

In this section we will describe relation between traceless matrices in two complex dimensions and spinors.
Let us assume that $\delta g$ is traceless then
\begin{equation}
 \omega \delta g^T+\delta g\omega=0
\end{equation}
thus for spinors $\u$ and $\v$
\begin{equation}
 [\u,\delta g^T \v]=\u^T\omega \delta g^T \v=-\v^T\delta g \omega \u=\v^T\omega \delta g^T \u=[\v,\delta g^T \u]
\end{equation}
where we transposed the whole formula in the middle equality.

\begin{lm}\label{lm:id2}
 Let us assume that $[\u,\v]=1$ then
 \begin{equation}
  (\v\u^T-\u\v^T)\omega ={\mathbb I}
 \end{equation}
\end{lm}

\begin{proof}
 We apply left hand side to $\v$ and $\u$
 \begin{align}
  (\v\u^T-\u\v^T)\omega \v=\v\underbrace{\u^T\omega \v}_{=1}- \u\underbrace{\v^T\omega \v}_{=0}=\v\\
  (\v\u^T-\u\v^T)\omega \u=\v\underbrace{\u^T\omega \u}_{=0}- \u\underbrace{\v^T\omega \u}_{=-1}=\u
 \end{align}
so as $\u$ and $\v$  are independent ($[u,v]\not=0$) and span the whole space of spinors we show that it is identity operator.
\end{proof}

We can write ($\delta g$ traceless $[\u,\v]=1$)
\begin{equation}\label{eq:re2}
[\u,\delta g^T\v]= \tr \v\u^T\omega \delta g^T=\tr \left(\v\u^T\omega -\frac{1}{2}{\mathbb I}\right)\delta g^T=\frac{1}{2}\tr (\v\u^T+\u\v^T)\omega \delta g^T
\end{equation}
In the last expression $\frac{1}{2}(\v\u^T+\u\v^T)\omega$ is traceless.

\begin{lm}\label{lm:sigma}
 Let us assume that $[\u,\v]=1$ then
 \begin{equation}
  (\v\u^T-\u\v^T)\omega ={\mathbb I}
 \end{equation}
Furthermore the matrix
 \begin{equation}
 \M=(\v\u^T+\u\v^T)\omega
\end{equation}
is traceless and
\begin{equation}
 \M\v=\v,\quad \M\u=-\u
\end{equation}
are two eigenvectors. Every traceless matrix with eigenvalues $\pm 1$ is of this form.
\end{lm}

\begin{proof}
 We can compute
 \begin{equation}
  \tr(\v\u^T+\u\v^T)\omega=[\u,\v]+[\v,\u]=0
 \end{equation}
and
\begin{align}
  (\v\u^T+\u\v^T)\omega \v=\v\underbrace{\u^T\omega \v}_{=1}+ \u\underbrace{\v^T\omega \v}_{=0}=\v\\
  (\v\u^T+\u\v^T)\omega \u=\v\underbrace{\u^T\omega \u}_{=0}+ \u\underbrace{\v^T\omega \u}_{=-1}=-\u
 \end{align}
Matrix $\M$ with eigenvectors $\u$ and $\v$ such that $\M\u=-\u$ and $\M\v=\v$ is of this form.
\end{proof}

We can write (for $\delta g$ traceless and $[\u,\v]=1$)
\begin{equation}\label{eq:re}
[\u,\delta g^T\v]=\frac{1}{2}\tr (\v\u^T+\u\v^T)\omega \delta g^T
\end{equation}

\begin{lm}\label{lm:comp}
 Suppose that spinors $\u,\v,\tilde{\u},\tilde{\v}$ satisfy
 \begin{equation}
  [\u,\v]=1,\quad [\tilde{\u},\tilde{\v}]=1
 \end{equation}
and
\begin{itemize}
 \item for all traceless matrices $\delta g$
\begin{equation}
\tr (\v\u^T+\u\v^T)\omega \delta g^T=\tr (\tilde{\v}\tilde{\u}^T+\tilde{\u}\tilde{\v}^T)\omega \delta g^T
\end{equation}
then there exists $\lambda\in \C^*$ such that
\begin{equation}
 \tilde{\u}=\lambda \u\quad \tilde{\v}=\lambda^{-1} \v
\end{equation}
\item and for all traceless matrices $\delta g$
\begin{equation}
 \tr (\v\u^T+\u\v^T)\omega \delta g^T=-\tr (\tilde{\v}\tilde{\u}^T+\tilde{\u}\tilde{\v}^T)\omega \delta g^T
\end{equation}
then there exists $\lambda\in \C^*$ such that
\begin{equation}
 \tilde{\u}=\lambda \v\quad \tilde{\v}=-\lambda^{-1} \u
\end{equation}
\end{itemize}
\end{lm}

\begin{proof}
 From lemma \ref{lm:sigma} matrices 
 \begin{equation}
  \M=\frac{1}{2}(\v\u^T+\u\v^T)\omega,\quad \tilde{\M}=\frac{1}{2} (\tilde{\v}\tilde{\u}^T+\tilde{\u}\tilde{\v}^T)\omega
 \end{equation}
are traceless thus $M=\tilde{M}$. The result now follows from uniqueness of eigenvectors up to scaling and lemma \ref{lm:sigma}.
 
 As for the second part we can write
 \begin{equation}
  -\frac{1}{2}\tr (\tilde{\v}\tilde{\u}^T+\tilde{\u}\tilde{\v}^T)\omega =
  \frac{1}{2}\tr (\hat{\v}\hat{\u}^T+\hat{\u}\hat{\v}^T)\omega 
 \end{equation}
 where $\hat{\u}=\tilde{\v}$ and $\hat{\v}=-\tilde{\u}$. Notice that $[\hat{\u},\hat{\v}]=1$ and we can apply previous part.
\end{proof}

\subsubsection{Spinors $\u$ and $\v$}

We introduce spinors
\begin{equation}
 \u=s\omega \eta_{N^{can}} \bar{\n},\quad \v=\n
\end{equation}
that satisfies
\begin{equation}
 [\u,\v]=s\n^\dagger \eta_{N^{can}}^T\omega^T\omega \n=s\langle \n,\n\rangle_{N^{can}}=1
\end{equation}

\begin{lm}\label{lm:u-r}
 For any spinor $\r$ holds
 \begin{equation}
  \langle \n,\r\rangle_{N^{can}}=s[\u,\r],\quad [\n,\r]=-st\langle \u,\r\rangle_{N^{can}}
 \end{equation}
\end{lm}

\begin{proof}
The first equality follows from
 \begin{equation}
  s[\u,\r]=\bar{\n}^T\eta_{N^{can}}^T\omega^T\omega  \r=\langle \n,\r\rangle_{N^{can}}
 \end{equation}
 The second by $\eta_{N^{can}}\omega \eta_{N^{can}}^T=t\omega$ 
\begin{equation}
 [\n,\r]=\n^T\omega \r=t\n^T\eta_{N^{can}}\omega\eta_{N^{can}}^T \r=-st\u^\dagger \eta_{N^{can}}^T \r=-st\langle \u,\r\rangle_{N^{can}}
\end{equation}
where we used $\eta_{N^{can}}^\dagger=\eta_{N^{can}}$ and $\omega^\dagger=-\omega$.
\end{proof}

\begin{lm}\label{lm:nn-uu}
 We have equality
 \begin{equation}
  \eta_{N^{can}}^T=\eta_{N^{can}}^T(s\n\n^\dagger +st\u\u^\dagger)\eta_{N^{can}}^T
 \end{equation}
\end{lm}

\begin{proof}
We will prove that
\begin{equation}
 {\mathbb I}=(s\n\n^\dagger +st\u\u^\dagger)\eta_{N^{can}}^T
\end{equation}
It is enough to check that action of both sides coincides on spinors $\v=\n$ and $\u$. We have by lemma \ref{lm:u-r}
\begin{equation}
 (s\n\n^\dagger +st\u\u^\dagger)\eta_{N^{can}}^T\n=s\n\underbrace{\langle \n,\n\rangle_{N^{can}}}_{=s}+st\u\underbrace{\langle \u,\n\rangle_{N^{can}}}_{=-st[\n,\n]=0}=\n
\end{equation}
and similarly
\begin{equation}
 (s\n\n^\dagger +st\u\u^\dagger)\eta_{N^{can}}^T\u=s\n\underbrace{\langle \n,\u\rangle_{N^{can}}}_{=s[\u,\u]=0}+st\u\underbrace{\langle \u,\u\rangle_{N^{can}}}_{=-st[\v,\u]=st}=\u
\end{equation}
\end{proof}

Consider $S^\rho, S^{\j}, S^{aux}$ defined in section \ref{se_derivatives}.
\begin{lm}\label{lm:reality2}
 On real manifold
 \begin{equation}
  \Re S^\rho=0,\quad \Re S^{\j}\leq 0
 \end{equation}
and the equality $\Re S^{\j}=0$ holds if and only if 
\begin{equation}
 \exists_{\xi\in\C}\ g^T\z=\xi \n
\end{equation}
\end{lm}

\begin{proof}
 From construction $\Re S^\rho=0$. We will consider $S^{\j}$. The reality condition is equivalent to reality condition for $S^{aux}$.
 
 We have equality by lemma \ref{lm:nn-uu}
 \begin{align}
  ts(g^T\z)^\dagger \eta_{N^{can}}^T (g^T\z)&=t(g^T\z)^\dagger \eta_{N^{can}}^T\n\n^\dagger \eta_{N^{can}}^T (g^T\z)+(g^T\z)^\dagger \eta_{N^{can}}^T\u\u^\dagger \eta_{N^{can}}^T (g^T\z)\nonumber\\&\geq t(g^T\z)^\dagger \eta_{N^{can}}^T\n\n^\dagger \eta_{N^{can}}^T (g^T\z)
 \end{align}
This is equivalent to
\begin{equation}
 t(s\langle g^T\z,g^T\z\rangle_{N^{can}})=t|\langle g^T\z,\n\rangle_{N^{can}}|^2+|\langle g^T\z,\u\rangle_{N^{can}}|^2\geq t|\langle g^T\z,\n\rangle_{N^{can}}|^2
\end{equation}
and because $s\langle g^T\z,g^T\z\rangle_{N^{can}}> 0$ and $|\langle g^T\z,\n\rangle_{N^{can}}|>0$ we can write it as
\begin{equation}
 (s\langle g^T\z,g^T\z\rangle_{N^{can}})^{t\ 2\j}\geq |\langle g^T\z,\n\rangle_{N^{can}}|^2)^{t\ 2\j}
\end{equation}
that is equivalent to $\Re S^{aux}\leq 0$.

The equality holds only if
\begin{equation}
 0=\langle \u,g^T\z\rangle_{N^{can}}=st[\n,g^T\z]
\end{equation}
that is $g^T\z=\xi \n$.
\end{proof}

\begin{lm}\label{lm:sij2}
 On real manifold $\Re S\leq 0$ and under condition $\Re S=0$
 \begin{equation}
  \delta_g S=\left(i\frac{\rho}{2}+\jj\right) [\u,\delta g^T \v]=\frac{1}{2}\left(i\frac{\rho}{2}+\j\right) \tr (\u\v^T+\v\u^T)\omega \delta g^T 
 \end{equation}
 where
 \begin{equation}
  \u=s\omega \eta_{N^{can}} \bar{\n},\quad \v=\n
 \end{equation}
 and $s=\langle \n,\n\rangle_{N^{can}}\in\{-1,1\}$
\end{lm}

\begin{proof}
 We have
 \begin{align}
  \delta_g S_\rho&=i\frac{\rho}{2}\frac{\langle g^T \z,\delta g^T g^T\z\rangle_{N^{can}}}{\langle g^T \z, g^T\z\rangle_{N^{can}}}=
  i\frac{\rho}{2}\frac{|\xi|^2\langle \n,\delta g^T \n\rangle_{N^{can}}}{|\xi|^2\langle \n, \n\rangle_{N^{can}}}=\nonumber\\
 &= i\frac{\rho}{2}\frac{s[\u,\delta g^T \n]}{s[\u,\n]}=i\frac{\rho}{2}[\u,\delta g^T \v]
 \end{align}
From lemma \ref{lm:SxSxbarn} when reality conditions are satisfied we have
\begin{equation}
 \delta_g S^{\j}=-t \delta_g S^{aux}
\end{equation}
The latter can be computed
\begin{equation}
 \delta_g S^{aux}=-t\jj\frac{\langle g^T \z,\delta g^T g^T\z\rangle_{N^{can}}}{\langle g^T \z, g^T\z\rangle_{N^{can}}}=-t\jj[\u,\delta g^T \v]
\end{equation}
thus the total variation
\begin{equation}
 \delta_g S=\delta_g S^\rho+\delta_g S^{\j}=\left(i\frac{\rho}{2}+\jj\right) [\u,\delta g^T \v]
\end{equation}
By \eqref{eq:re} it can be written in the given form.
\end{proof}

\subsubsection{Subalgebra of the subgroup preserving the normal $N$}
\label{sec:subalgebra}

We will use description of $\St(N)$ from \ref{sec:subgroup}.
Differentiating \eqref{eq:sub} we obtain that the Lie subalgebra of $\SL$ consists of these traceless matrices $\B$ that satisfies
\begin{equation}
 \B\eta_N+\eta_N \B^\dagger=0,\quad \B=-\eta_N \B^\dagger\eta_N^{-1}
\end{equation}

\begin{lm}
 Suppose that $\M$ is a traceless matrix such that $\M^T$ belongs to the Lie algebra of the group preserving normal $N$, then 
 the eigenvalues of $\M$ are either real or purely imaginary.
\end{lm}

\begin{proof}
We can always write
\begin{equation}
 \M=i\lambda (\v\u^T+\u\v^T)\omega,\quad [\u,\v]=1
\end{equation}
with $\v$ being $i\lambda$ and $\u$ being $-i\lambda$ eigenvectors with $\lambda\in \C$.

The condition for $\M$ is
\begin{equation}
 \M=-(\eta_N^{-1})^T \M^\dagger \eta_N^T
\end{equation}
means that $\M^\dagger$ has the same eigenvalues as $-\M$ ($\pm i\lambda$) so either $\lambda\in\R$ or $\lambda\in i\R$.
\end{proof}

We will call $\M$ spacelike if its eigenvalues are purely imaginary (see section \ref{sec:bivectors}).

\begin{lm}
 Suppose that $\M$ is a traceless matrix such that $\M^T$ belongs to the Lie algebra of the group preserving normal $N$ and is spacelike then 
 there exists
 \begin{equation}
  \lambda>0,\text{ and } \n \text{ spinor } \langle \n,\n\rangle_N=s\in\{-1,1\}
 \end{equation}
 such that
 \begin{equation}
  \M=i\lambda (\v\u^T+\u\v^T)\omega,\quad [\u,\v]=1
 \end{equation}
 where 
 \begin{equation}
  \u=s\omega\eta_N\bar{\n},\quad \v=\n,
 \end{equation}
\end{lm}

\begin{proof}
We can always write
\begin{equation}
 \M=i\lambda (\v\u^T+\u\v^T)\omega,\quad [\u,\v]=1
\end{equation}
with $\v$ being $i\lambda$ and $\u$ being $-i\lambda$ eigenvectors with $\lambda\geq 0$.

The condition for $M$ is
\begin{equation}
 \M=-(\eta_N^{-1})^T \M^\dagger \eta_N^T=t\omega\eta_N (-\omega)\M^\dagger  \eta_N 
\end{equation}
where $t=\det\eta_N=N\cdot N$ from the identity $t(\eta_N^{-1})^T=\omega \eta_N(-\omega)$.
It is equivalent to
\begin{equation}
 i\lambda (\v\u^T+\u\v^T)\omega=it\lambda\omega \eta_N (-\omega)^2(\bar{\v}\bar{\u}^T+\bar{\u}\bar{\v}^T)\eta_N \omega^2
\end{equation}
We can write it in the form
\begin{equation}
 (\v\u^T+\u\v^T)\omega=-(\tilde{\v}\tilde{\u}^T+\tilde{\u}\tilde{\v}^T)\omega
\end{equation}
where we introduced 
\begin{equation}
 \tilde{\u}=\omega\eta_N \bar{\u},\quad \tilde{\v}=t\omega \eta_N\bar{\v}
\end{equation}
Let us notice that 
\begin{equation}
 [\tilde{\u},\tilde{\v}]=t\bar{\u}^T\underbrace{\eta_N^T\omega \eta_N}_{=t\omega}\bar{\v}=
 \bar{\u}^T\omega \bar{\v}=\widebar{[\u,\v]}=1
\end{equation}
From lemma \ref{lm:comp} there exists $\xi\in\C$
 \begin{equation}
  \u=\xi \omega \eta_N\bar{\v},\ \v=-\xi^{-1}t\omega \eta_N \bar{\u}
 \end{equation}
We have also the equality
\begin{equation}
 [\u,\v]=\xi \v^\dagger \eta_n^T\omega^T\omega \v=\xi \langle \v,\v\rangle_N
\end{equation}
Normalizing $\v$ we can introduce spinor $\n$ such that
 \begin{equation}
  \u=s\omega\eta_N\bar{\n},\quad \v=\n
 \end{equation}
Spinor $\n$ is unique up to a phase.
\end{proof}

\subsection{Characterization of the bivectors}

We will now characterize 
$\pi^{-1}\left( -\frac{2\gamma}{\gamma-i}\B_{ij}^T\right)$.
Let us introduce a vector $l^\mu_{ij}$ by an identity
\begin{equation}
 l_{ij}^\mu\hat{\sigma}_\mu=\left(2s_{ij}\rho_{ij}\ {\n}_{ij}\n_{ij}^\dagger\right)^T
\end{equation}
It is null because the matrix on the right-hand side has rank one. It is future directed if $s_{ij}=1$ and past directed if $s_{ij}=-1$.

Contracting with ${\sigma}_\nu$ and taking the trace we obtain
\begin{equation}
 l_{ij}^\mu=s_{ij}\rho_{ij}\n_{ij}^T {\sigma}^\mu \bar{\n}_{ij}
\end{equation}
Let us consider decomposition (introducing vector $v_{ij}$)
\begin{equation}
 l_{ij}={v}_{ij}+cN_i^{can},\quad {v}_{ij}\perp N_i^{can},\quad c=\frac{l_{ij}\cdot N_i^{can}}{N_i^{can}\cdot N_i^{can}}
\end{equation}
We have
\begin{equation}
 l_{ij}\cdot N_i^{can}=s_{ij}\rho_{ij}\n_{ij}^T\eta_{N_i^{can}}\bar{\n}_{ij}=s_{ij}\rho_{ij}\langle \n_{ij}|\n_{ij}\rangle_{N_i^{can}}=\rho_{ij}
\end{equation}
thus
\begin{equation}
 v_{ij}\cdot v_{ij}=l_{ij}\cdot l_{ij}-\frac{(l_{ij}\cdot N_i^{can})^2}{N_i^{can}\cdot N_i^{can}}=-t_i \rho_{ij}^2
\end{equation}

\begin{lm}\label{lm:char}
We have
\begin{equation}
 \pi(B_{ij}')=-\frac{2\gamma}{\gamma-i}\B_{ij}^T
\end{equation}
where $B_{ij}'= {\Hodge}({v}_{ij}\wedge N^{can}_i)={\Hodge}(l_{ij}\wedge N^{can}_i)$.
\end{lm}

\begin{proof}
The traceless matrix $\B'_{ij}=\pi(B_{ij}')$ satisfies
\begin{equation}
 \B'_{ij}=i \frac{1}{4}(\eta_{l_{ij}}\hat{\eta}_{N^{can}_i}-\eta_{N^{can}_i}\hat{\eta}_{l_{ij}})=
 -i\frac{1}{2}(\eta_{N^{can}_i}\hat{\eta}_{l_{ij}}-l_{ij}\cdot N^{can}_i{\mathbb I})
\end{equation}
For any traceless matrix ${\mathbb K}$
\begin{equation}
 \tr \frac{2\gamma}{\gamma-i}\B_{ij} {\mathbb K}=i\rho_{ij} s_{ij}\langle \n_{ij}, {\mathbb K} \n_{ij}\rangle_{N_i^{can}}
\end{equation}
but also
\begin{equation}
 \tr {\B'}_{ij}^T {\mathbb K}=-\tr  (i\rho_{ij}s_{ij}\n_{ij}\n_{ij}^\dagger \eta_{N_i^{can}}^T-l_{ij}\cdot N^{can}_i{\mathbb I}) {\mathbb K}=-i
 \rho_{ij} s_{ij}\langle \n_{ij}, {\mathbb K} \n_{ij}\rangle_{N_i^{can}}
\end{equation}
This equality shows that the traceless matrices $\B_{ij}'$ and $-\frac{2\gamma}{\gamma-i}\B_{ij}^T$ are equal. 
\end{proof}

\subsection{Geometric bivectors}\label{sec:bivectors22}

We will perform the following construction of the $k$ form corresponding to the simplex spanned on the points $x_0,\ldots x_k$ in $\R^n$.

Let us introduce auxiliary space $\R^{n+1}$ and in this space vectors
\begin{equation}
 y_i=\vc{x_i^0\\ \vdots\\ x_i^n\\ 1},\text { and } a=\vc{0\\ \vdots\\ 0\\ 1}
\end{equation}
and covector
\begin{equation}
 A=(0,\ldots 0, 1).
\end{equation}
Let us introduce $k+1$ vectors
\begin{equation}
 \tilde{V}_{\alpha_0\cdots \alpha_{k}}=y_{\alpha_0}\wedge\cdots \wedge y_{\alpha_{k}}
\end{equation}
$k$-vectors in $\R^n$ can be identified with $k$-vectors $\Omega$ in $\R^{n+1}$ that satisfy
\begin{equation}
 A\llcorner\Omega=0
\end{equation}
From $\tilde{V}_{\alpha_0\cdots \alpha_{k}}$ we can produce a $k$-vector in $\R^n$ 
\begin{equation}
V_{\alpha_0\cdots \alpha_k}= A\llcorner \tilde{V}_{\alpha_0\cdots \alpha_k}
\end{equation}
as $A\llcorner A\llcorner \tilde{V}_{\alpha_0\cdots \alpha_k}=0$. 

We can check that
\begin{align}
 &A\llcorner y_{\alpha_0}\wedge\cdots \wedge y_{\alpha_k}=\nonumber\\
&= A\llcorner y_{\alpha_0}\wedge(y_{\alpha_1}-y_{\alpha_0})\cdots \wedge (y_{\alpha_k}-y_{\alpha_0})=(y_{\alpha_1}-y_{\alpha_0})\cdots \wedge (y_{\alpha_k}-y_{\alpha_0})
\end{align}
as the only nonzero last component is in $y_{\alpha_0}$. This gives (after restriction to $\R^n$)
\begin{equation}
 V_{\alpha_0\cdots \alpha_k}=(x_{\alpha_1}-x_{\alpha_0})\cdots \wedge (x_{\alpha_k}-x_{\alpha_0})
\end{equation}
that is the volume $k$-vector of the $k$-simplex multiplied by $k!$.

Let us notice that this $k$-vector depends on the order of $\alpha_0,\ldots,\alpha_k$. As $\tilde{V}_{\alpha_0\cdots \alpha_k}$ change by $(-1)^{\sgn\sigma}$ under permutation of points (even permutations preserves it) the same is true for $V_{\alpha_0\cdots \alpha_k}$.

Suppose that we have a simplex determined by points with indices $0,\ldots n$. We can define codimension $1$ and $2$ simplices by indicating which points we skip.

Let us introduce
\begin{equation}
 V_i=(-1)^iV_{0\cdots \hat{i}\cdots n}
\end{equation}
where $\hat{\cdot}$ means omission,
\begin{align}
 B_{ij}=\left\{\begin{array}{ll}
                (-1)^{i+j+1}V_{0\cdots \hat{i}\cdots\hat{j} \cdots n} & i<j\\
                (-1)^{i+j}V_{0\cdots \hat{j}\cdots\hat{i} \cdots n} & i>j\\
               \end{array}\right.
\end{align}
and similarly $\tilde{V}_i$ and $\tilde{B}_{ij}$.
With this definition $B_{ij}=-B_{ji}$.

\begin{thm}
 The following holds
 \begin{equation}\label{eq:geom_bivector_identities}
 \sum_i V_i=0,\quad \forall_i \sum_{j\not=i} B_{ij}=0,
 \end{equation}
\end{thm}

\begin{proof}
 Let us consider
 \begin{equation}
  B=A\llcorner \tilde{V}_{i}=A\llcorner \tilde{V}_{0\cdots \hat{i}\cdots n}
 \end{equation}
 It can be written as
\begin{equation}
 \sum_{j\not=i} (-1)^{s_j}y_{0}\wedge \cdots\wedge(A\llcorner y_{j})\wedge\cdots\wedge y_{n}=
 \sum_{j\not=i} (-1)^{s_j}y_{0}\wedge \cdots\wedge\hat{y}_{j}\wedge\cdots\wedge y_{n}
\end{equation}
where $s_j$ is the number of site on which we contract.
\begin{equation}
 s_j=\left\{\begin{array}{ll}
             j+1 & j<i\\
             j & j>i
            \end{array}\right.
\end{equation}
The sum can be written as
\begin{equation}
(-1)^i \sum  \tilde{B}_{ij}
\end{equation}
However this $n-1$ vector contracted with $A$ is zero 
\begin{equation}
 0=A\llcorner A\llcorner \tilde{V}_i=A\llcorner B=(-1)^i \sum  B_{ij}
\end{equation}
This finishes the proof of the second equality in \eqref{eq:geom_bivector_identities}. We prove the first equality in similar way:
\begin{equation}
 0=A\llcorner A\llcorner \tilde{V}_{0\cdots n}=A\llcorner \sum_i (-1)^i \tilde{V}_{0\cdots \hat{i}\cdots n}=
 \sum_i V_i.
\end{equation}
\end{proof}

Let us restrict to the case of $n=4$,

\begin{df}
The geometric bivectors of the $4$ simplex determined by vertices $x_0,\ldots x_4\in\R^4$ are
\begin{equation}
 B_{ij}^{\geom}= B_{ij}.
\end{equation}
\end{df}

Let us consider a scaling transformation
 \begin{equation}
  x_i\rightarrow\lambda x_i,\quad \lambda\in\R^*.
 \end{equation}
Under this transformation bivectors changes
\begin{equation}
 B_{ij}^{\geom}\rightarrow \lambda^2 B_{ij}^{\geom}.
\end{equation}
Let us notice that in particular inversion transformation $\lambda=-1$ preserves $B_{ij}^{\geom}$.

\subsubsection{Nondegenerate case}

Let us now assume that $x_i$ do not lay in the hyperplane that is $y_i$ are linearly independent. We introduce a dual basis
$\hat{y}_i$ defined by
\begin{equation}\label{eq:yhat}
 \hat{y}_i\lrcorner y_j=\delta_{ij}
\end{equation}
Let us also introduce $\tilde{y}_i$ defined by
\begin{equation}
 \hat{y}_i=\tilde{y}_i+\mu_i A,\quad \tilde{y}_i\lrcorner a=0
\end{equation}
These covectors can be regarded as belonging to $\R^n$.
We have
\begin{equation}
 V_i=A\llcorner \hat{y}_{i}\llcorner \tilde{V}_{0\cdots n}=- \hat{y}_{i}\llcorner V_{0\cdots n}.
\end{equation}
This can be written with the use of $\tilde{y}_i$ in terms of $\R^n$
\begin{equation}
 V_i=-\tilde{y}_i\llcorner V_{0\cdots n}
\end{equation}
and similarly
\begin{equation}
  B_{ij}=\hat{y}_j\llcorner \hat{y}_i\llcorner V_{0\cdots n}=\tilde{y}_j\llcorner \tilde{y}_i\llcorner V_{0\cdots n}.
\end{equation}

\begin{lm}
 We have 
 \begin{equation}
  \sum_{i=0}^n \hat{y}_i=A
 \end{equation}
\end{lm}

\begin{proof}
 It is enough to check that both sides of equality give the same value contracted with the elements of the basis $y_j$
 \begin{equation}
  \sum_{i=0}^n \hat{y}_i\llcorner y_j=\sum_{i=0}^n \delta_{ij}=1=A\llcorner y_j
 \end{equation}
\end{proof}

Thus we have 
\begin{equation}
  \sum_{i=0}^n \tilde{y}_i=0
\end{equation}
Covectors $\tilde{y}_i$ are conormal to subsimplices $V_{0\cdots \hat{i}\cdots n}$.

Let us now add metric tensor on $\R^n$. It defines also a scalar product on $k$-vectors by
\begin{equation}
 a_1\wedge\cdots a_k \cdou b_1\wedge\cdots b_k:=\sum_{\sigma\in S_k} (-1)^{\sgn\sigma}\prod_{i=1}^k a_i\cdou b_{\sigma_i}
\end{equation}
We can also introduce normalization of  $V_{0\cdots n}$ by
\begin{equation}\label{eq:vol}
 V_{0\cdots n}\cdou V_{0\cdots n}=(-1)^d(Vol^{\geom})^2,\quad Vol^{\geom}>0
\end{equation}
By the definition of Hodge star (see Appendix \ref{sec:conventions})
\begin{equation}
 B_{ij}=-Vol^{\geom}\ {\Hodgu}(\tilde{y}_j\wedge \tilde{y}_i),\quad V_i=-Vol^{\geom}\ {\Hodgu}\tilde{y}_i
\end{equation}
If codimension $1$ subsimplices are not null we can introduce normal vectors $N_i$ and positive numbers $W_i$ and $\T_i\in\{-1,1\}$ such that
\begin{equation}
 \tilde{y}_i=\frac{1}{Vol^{\geom}} W_iN_i,\quad N_i\cdou N_i=\T_i
\end{equation}
then we have 
\begin{equation}\label{B-geom}
 B_{ij}=-\frac{1}{Vol^{\geom}} W_iW_j {\Hodgu}(N_j\wedge N_i),\quad \sum_i W_i N_i=0
\end{equation}
Let us consider an altitude for our simplex from point $x_i$. Its base we will denote by $h_i$. Let us notice that as it lays inside the hyperplane of remaining points
\begin{equation}
 \exists_{\alpha_k,k\not=i}\colon \sum_{k\not=i}\alpha_k=1,\quad\sum_{k\not=i}\alpha_kx_k=h_k
\end{equation}
We have (identifying vectors and covectors using scalar product)
\begin{equation}
\tilde{y}_i\cdou (x_i-h_i)=\hat{y}_i\llcorner (y_i-\sum_{k\not= i}\alpha_ky_k)=1\geq 0 
\end{equation}
Thus vectors 
\begin{equation}
N_i^{\geom}=-\T_iN_i\label{eq:N-geom}
\end{equation}
are outer directed and introducing $W_i^{\geom}=-\T_iW_i$ ($W_i>0$) we get
\begin{equation}
 \sum_i W_i^{\geom} N_i^{\geom}=0
\end{equation}
Let us notice that 
\begin{itemize}
 \item $B_{ij}^{\geom}=-B_{ji}^{\geom}$
 \item $B_{ij}^{\geom}=-\frac{1}{Vol^{\geom}} W_i^{\geom}W_j^{\geom} {\Hodgu}(N_j^{\geom}\wedge N_i^{\geom})$
 \item $B_{ij}^{\geom}\lrcorneu N_i^{\geom}=0$.
\end{itemize}
Let us notice that for a spacelike face its area is equal to 
\begin{equation}
A_{ij}^\geom= \frac{1}{2}|B_{ij}^{\geom}|.
\end{equation}

\subsection{Generalized sine law}\label{sec:sin-law}

Taking the scalar product of \eqref{B-geom} with itself we obtain (for an extension see \cite{Eriksson})
\begin{equation}
 (Vol^{\geom})^2|B_{ij}^\geom|^2={W_i^\geom}^2 {W_j^\geom}^2 (-1)^s |N_j^\geom\wedge, N_i^\geom|^2
\end{equation}
where $s$ is a sign related to the signature of spacetime, defined by
\begin{equation}
 |\Hodgu N_j^\geom\wedge N_i^\geom|^2=(-1)^s|N_j^\geom\wedge N_i^\geom|^2.
\end{equation}
However
\begin{equation}
 |N_j^\geom\wedge N_i^\geom|^2={N_i^\geom}^2{N_j^\geom}^2- (N_i^\geom\cdou N_j^\geom)^2
\end{equation}
Let us consider the following cases:
\begin{itemize}
 \item If $N_i^\geom$ and $N_j^\geom$ are of the same type and signature in the plane $N_i^\geom\wedge N_j^\geom$ is mixed then
 \begin{equation}
  |N_j^\geom\wedge N_i^\geom|^2=-\sinh^2 \theta_{ij},\quad |N_i^\geom\cdou N_j^\geom|=\cosh\theta_{ij}.
 \end{equation}
\item If $N_i$ and $N_j$ are of different types and signature in the plane $N_i\wedge N_j$ is mixed then
 \begin{equation}
  |N_j^\geom\wedge N_i^\geom|^2=-\cosh^2 \theta_{ij},\quad N_i^\geom\cdou N_j^\geom=\sinh\theta_{ij}.
 \end{equation}
\item If signature in the plane $N_i\wedge N_j$ is $++$ then
\begin{equation}
  |N_j^\geom\wedge N_i^\geom|^2=\sin^2 \theta_{ij},\quad N_i^\geom\cdou N_j^\geom=\cos\theta_{ij}.
 \end{equation}
\item If signature in the plane $N_i^\geom\wedge N_j^\geom$ is $--$ then
\begin{equation}
  |N_j^\geom\wedge N_i^\geom|^2=\sin^2 \theta_{ij},\quad N_i^\geom\cdou N_j^\geom=-\cos\theta_{ij}.
 \end{equation} 
\end{itemize}

\subsection{Dihedral angles}\label{app:NN}

We will now describe $R_{N_i^\geom}R_{N_j^\geom}$ in terms of dihedral angles. Let us notice that
${\Hodgu}N_i^\geom\wedge N_j^\geom$ is spacelike. We have
\begin{equation}
 |{\Hodgu}N_i^\geom\wedge N_j^\geom|=\sqrt{({\Hodgu}N_i^\geom\wedge N_j^\geom)\cdot({\Hodgu}N_i^\geom\wedge N_j^\geom)}.
\end{equation}
Let us introduce
$$
d_{ij}=\left\{\begin{array}{ll} 0 & N_i^{can}=N_j^{can}\\
                                    1 & N_i^{can}\not=N_j^{can}
                                   \end{array}\right.,$$
and also
\begin{align}
 s_{ij}(\theta)=\left\{\begin{array}{ll}
                   \sin\theta & \text{if the plane spanned by normals is euclidean}\\
                   \sinh\theta & \text{if the plane spanned by normals is mixed}
                  \end{array}\right.\\
c_{ij}(\theta)=\left\{\begin{array}{ll}
                   \cos\theta & \text{if the plane spanned by normals is euclidean}\\
                   \cosh\theta & \text{if the plane spanned by normals is mixed}
                  \end{array}\right..
\end{align}
We will first prove

\begin{lm}\label{lm:cs}
 The following identities hold
 \begin{align}
 s_{ij}(2\theta_{ij})&=(-1)^{d_{ij}} 2(N_i^\geom \cdou N_j^\geom)|{\Hodgu}N_i^\geom\wedge N_j^\geom|\\
 c_{ij}(2\theta_{ij})&=(-1)^{d_{ij}}(2\T_i^{can}\T_j^{can} (N_i^\geom \cdou N_j^\geom)^2-1)
\end{align}
\end{lm}

\begin{proof}
 Let us consider cases. Using results of appendix \ref{sec:sin-law}
\begin{itemize}
 \item $\T_i^{can}=\T_j^{can}$ and $N_i^{\geom}\cdou N_j^{\geom}>0$ in Lorentzian signature $d_{ij}=0$
 \begin{align}
  &2(N_i^\geom \cdou N_j^\geom)|{\Hodgu}N_i^\geom\wedge N_j^\geom|=2\cosh\theta_{ij} |\sinh\theta_{ij}|=\sinh(2\theta_{ij}),\\
  &2\T_i^{can}\T_j^{can} (N_i^\geom \cdou N_j^\geom)^2-1=2\cosh^2\theta_{ij}-1=\cosh(2\theta_{ij}),
 \end{align}
 because $\theta_{ij}>0$.
\item $\T_i^{can}=\T_j^{can}$ and $N_i^{\geom}\cdou N_j^{\geom}<0$ in Lorentzian signature $d_{ij}=0$
 \begin{align}
  &2(N_i^\geom \cdou N_j^\geom)|{\Hodgu}N_i^\geom\wedge N_j^\geom|=-2\cosh\theta_{ij} |\sinh\theta_{ij}|=\sinh(2\theta_{ij}),\\
  &2\T_i^{can}\T_j^{can} (N_i^\geom \cdou N_j^\geom)^2-1=2\cosh^2\theta_{ij}-1=\cosh(2\theta_{ij})
 \end{align}
 because $\theta_{ij}<0$.
 \item $\T_i^{can}\not=\T_j^{can}$ in Lorentzian signature $d_{ij}=1$
 \begin{align}
  &2(-1)(N_i^\geom \cdou N_j^\geom)|{\Hodgu}N_i^\geom\wedge N_j^\geom|=2\sinh\theta_{ij} \cosh\theta_{ij}=\sinh(2\theta_{ij}),\\
  &(-1)(2\T_i^{can}\T_j^{can} (N_i^\geom \cdou N_j^\geom)^2-1)=2\sinh^2\theta_{ij}+1=\cosh(2\theta_{ij}).
 \end{align}
 \item $\T_i^{can}=\T_j^{can}$ in Euclidean or split signature $d_{ij}=0$
 \begin{align}
  &2(N_i^\geom \cdou N_j^\geom)|{\Hodgu}N_i^\geom\wedge N_j^\geom|=2\T^{can}\cos\theta_{ij} |\sin\theta_{ij}|=\sin(2\theta_{ij}),\\
  &2\T_i^{can}\T_j^{can} (N_i^\geom \cdou N_j^\geom)^2-1=2\cos^2\theta_{ij}-1=\cos(2\theta_{ij}),
 \end{align}
 because $\T^{can}\theta_{ij}\in(0,\pi)$.
\end{itemize}
\end{proof}

Let us introduce the inversion in the plane spanned by $N_i^\geom, N_j^\geom$:
\begin{equation}
 O_{ij}=Ie^{\pi \frac{\Hodgu N_i^\geom\wedge N_j^\geom}{|{\Hodgu}N_i^\geom\wedge N_j^\geom|}}.
\end{equation}

\begin{lm}The following holds for geometric normals:
$$R_{N_i^\geom}R_{N_j^\geom}=O_{ij}^{d_{ij}} e^{2\T_i^{can}\T_j^{can}\theta_{ij} \frac{N_i^\geom\wedge N_j^\geom}{|{\Hodgu}N_i^\geom\wedge N_j^\geom|}},$$
where $\theta_{ij}$ is dihedral angle.
\end{lm}

\begin{proof}
We can restrict ourselves to two dimensional plane spanned by $N_i^\geom$ and $N_j^\geom$. The connected component of the group of Lorentz transformations is generated by $N_i^\geom\wedge N_j^\geom$. Moreover for two elements $G$ and $G'$ from connected component 
\begin{equation}
 G=G'\Longleftrightarrow G-G^{-1}=G'-{G'}^{-1}\text{ and } \tr G=\tr G'
\end{equation}
Let us notice that
\begin{equation}
 G=O_{ij}^{d_{ij}}R_{N_i^\geom}R_{N_j^\geom},\quad G^{-1}=O_{ij}^{d_{ij}}R_{N_j^\geom}R_{N_i^\geom}.
\end{equation}
is in connected component and for any vector $v$ in the plane $R_{N_i^\geom}R_{N_j^\geom}v$ is equal to
\begin{equation}
v-2\T_i^{can}(N_i^\geom\cdou v)N_i^{\geom}-2\T_j^{can}(N_j^\geom\cdou v)N_j^{\geom}+4\T_i^{can}\T_j^{can}(N_i^\geom\cdou N_j^{\geom})(N_j^{\geom}\cdou v)N_i^{\geom},
\end{equation}
thus
\begin{align}
 &G-G^{-1}=4(-1)^{d_{ij}} \T_i^{can}\T_j^{can}(N_i^\geom \cdou N_j^\geom)(N_i^\geom \wedge N_j^\geom)\\
 &\tr G=(-1)^{d_{ij}}(4\T_i^{can}\T_j^{can} (N_i^\geom \cdou N_j^\geom)^2-2).
\end{align}
Similarly expanding in the Taylor series
\begin{equation}
 G'=e^{2\T_i^{can}\T_j^{can}\theta_{ij} \frac{N_i^\geom\wedge N_j^\geom}{|{\Hodgu}N_i^\geom\wedge N_j^\geom|}}=
 c_{ij}(2\theta_{ij})+s_{ij}(2\theta_{ij})\T_i^{can}\T_j^{can} \frac{N_i^\geom\wedge N_j^\geom}{|{\Hodgu}N_i^\geom\wedge N_j^\geom|}
\end{equation}
and 
\begin{align}
 &G'-{G'}^{-1}=\frac{2s_{ij}(2\theta_{ij})\T_i^{can}\T_j^{can} }{|{\Hodgu}N_i^\geom\wedge N_j^\geom|} N_i^\geom\wedge N_j^\geom,\\
 & \tr G'=2c_{ij}(2\theta_{ij}).
\end{align}
The equality now follows from lemma \ref{lm:cs}
\begin{align}
 &s_{ij}(2\theta_{ij})=(-1)^{d_{ij}} 2(N_i^\geom \cdou N_j^\geom)|{\Hodgu}N_i^\geom\wedge N_j^\geom|,\\ &c_{ij}(2\theta_{ij})=(-1)^{d_{ij}}(2\T_i^{can}\T_j^{can} (N_i^\geom \cdou N_j^\geom)^2-1).
\end{align}
which holds for the dihedral angles.
\end{proof}

%% file: ms.bbl
\begin{thebibliography}{10}

\bibitem{turaev-viro}
V.~G. Turaev and O.~Y. Viro, ``State sum invariants of 3 manifolds and quantum
  6j symbols,'' {\em Topology}, vol.~31, pp.~865--902, 1992.

\bibitem{PR}
G.~Ponzano and T.~Regge, {\em Semiclassical limit of {Racah} coefficients, in:
  Spectroscopy and group theoretical methods in physics}.
\newblock Amsterdam: North Holland Publ. Co., 1968.

\bibitem{pr-model}
J.~W. Barrett and I.~Naish-Guzman, ``The {Ponzano}-{Regge} model,'' {\em Class.
  Quant. Grav.}, vol.~26, p.~155014, 2009,
  \href{http://arxiv.org/abs/0803.3319}{{\texttt arXiv:0803.3319}}.
\newblock \href{http://arxiv.org/abs/0803.3319}{[arXiv:0803.3319 [gr-qc]]}.

\bibitem{ReisenbergerPath}
M.~P. Reisenberger and C.~Rovelli, ````sum over surfaces'' form of loop quantum
  gravity,'' {\em Phys.Rev.D}, vol.~56, pp.~3490--3508, 1997,
  \href{http://arxiv.org/abs/gr-qc/9612035}{{\texttt arXiv:gr-qc/9612035}}.

\bibitem{Reisenberger2000}
M.~Reisenberger and C.~Rovelli, ``Spin foams as feynman diagrams,'' {\em
  Conference C01-09-03.3, Proceedings}, pp.~431--448, 2000,
  \href{http://arxiv.org/abs/gr-qc/0002083}{{\texttt arXiv:gr-qc/0002083}}.

\bibitem{Reisenberger2001}
M.~P. Reisenberger and C.~Rovelli, ``Spacetime as a feynman diagram: the
  connection formulation,'' {\em Class.Quant.Grav.}, vol.~18, pp.~121--140,
  2001, \href{http://arxiv.org/abs/gr-qc/0002095}{{\texttt
  arXiv:gr-qc/0002095}}.

\bibitem{Pietri2000}
R.~D. Pietri and C.~Petronio, ``Feynman diagrams of generalized matrix models
  and the associated manifolds in dimension 4,'' {\em J.Math.Phys.}, vol.~41,
  pp.~6671--6688, 2000, \href{http://arxiv.org/abs/gr-qc/0004045}{{\texttt
  arXiv:gr-qc/0004045}}.

\bibitem{GFTII}
L.~Freidel, ``{Group field theory: An Overview},'' {\em Int.J.Theor.Phys.},
  vol.~44, pp.~1769--1783, 2005,
  \href{http://arxiv.org/abs/hep-th/0505016}{{\texttt arXiv:hep-th/0505016}}.

\bibitem{GFTIII}
J.~Ben~Geloun, R.~Gurau, and V.~Rivasseau, ``{EPRL/FK Group Field Theory},''
  {\em Europhys.Lett.}, vol.~92, p.~60008, 2010,
  \href{http://arxiv.org/abs/1008.0354}{{\texttt arXiv:1008.0354}}.

\bibitem{GFTIV}
T.~Krajewski, J.~Magnen, V.~Rivasseau, A.~Tanasa, and P.~Vitale, ``{Quantum
  Corrections in the Group Field Theory Formulation of the EPRL/FK Models},''
  {\em Phys.Rev.}, vol.~D82, p.~124069, 2010,
  \href{http://arxiv.org/abs/1007.3150}{{\texttt arXiv:1007.3150}}.

\bibitem{GFTV}
D.~Oriti, J.~P. Ryan, and J.~Thürigen, ``{Group field theories for all loop
  quantum gravity},'' {\em New J. Phys.}, vol.~17, no.~2, p.~023042, 2015,
  \href{http://arxiv.org/abs/1409.3150}{{\texttt arXiv:1409.3150}}.

\bibitem{Barrett2000}
J.~W. Barrett and L.~Crane, ``A lorentzian signature model for quantum general
  relativity,'' {\em Class.Quant.Grav.}, vol.~17, pp.~3101--3118, 2000,
  \href{http://arxiv.org/abs/gr-qc/9904025}{{\texttt arXiv:gr-qc/9904025}}.

\bibitem{PietriClass}
R.~D. Pietri and L.~Freidel, ``so(4) plebanski action and relativistic spin
  foam model,'' {\em Class.Quant.Grav.}, vol.~16:2187-2196,1999,
  Class.Quant.Grav.16:2187-2196,1999,
  \href{http://arxiv.org/abs/gr-qc/9804071}{{\texttt arXiv:gr-qc/9804071}}.

\bibitem{Alexandrov2008}
S.~Alexandrov, ``Simplicity and closure constraints in spin foam models of
  gravity,'' {\em Phys.Rev.D}, vol.~78:044033,2008, Feb. 2008,
  \href{http://arxiv.org/abs/0802.3389}{{\texttt arXiv:0802.3389}}.

\bibitem{Livine2008b}
E.~R. Livine and S.~Speziale, ``{Consistently Solving the Simplicity
  Constraints for Spinfoam Quantum Gravity},'' {\em EPL}, vol.~81, no.~5,
  p.~50004, 2008, \href{http://arxiv.org/abs/0708.1915}{{\texttt
  arXiv:0708.1915}}.

\bibitem{EPRL}
J.~Engle, E.~Livine, R.~Pereira, and C.~Rovelli, ``{L}{Q}{G} vertex with finite
  {I}mmirzi parameter,'' {\em Nucl. Phys. B}, vol.~799, p.~136, 2008,
  \href{http://arxiv.org/abs/0711.0146}{{\texttt arXiv:0711.0146}}.

\bibitem{FK}
L.~Freidel and K.~Krasnov, ``A {N}ew {S}pin {F}oam {M}odel for 4d {G}ravity,''
  {\em Class. Quant. Grav.}, vol.~25, p.~125018, 2008,
  \href{http://arxiv.org/abs/0708.1595}{{\texttt arXiv:0708.1595}}.

\bibitem{Baratin}
A.~Baratin and D.~Oriti, ``{Group field theory and simplicial gravity path
  integrals: A model for Holst-Plebanski gravity},'' {\em Phys. Rev.},
  vol.~D85, p.~044003, 2012, \href{http://arxiv.org/abs/1111.5842}{{\texttt
  arXiv:1111.5842}}.

\bibitem{Perez2012}
A.~Perez, ``The {S}pin {F}oam {A}pproach to {Q}uantum {G}ravity,'' {\em Living
  Rev.Rel.}, vol.~16, p.~3, May 2013,
  \href{http://arxiv.org/abs/1205.2019}{{\texttt arXiv:1205.2019}}.

\bibitem{Kaminski2009}
W.~Kamiński, M.~Kisielowski, and J.~Lewandowski, ``Spin-foams for all loop
  quantum gravity,'' {\em Class.Quant.Grav.}, vol.~27:095006,2010, Sept. 2009,
  \href{http://arxiv.org/abs/0909.0939}{{\texttt arXiv:0909.0939}}.

\bibitem{Alexandrov2011}
S.~Alexandrov, M.~Geiller, and K.~Noui, ``Spin foams and canonical
  quantization,'' {\em SIGMA}, vol.~8, Dec. 2011,
  \href{http://arxiv.org/abs/1112.1961}{{\texttt arXiv:1112.1961}}.

\bibitem{Alesci2011}
E.~Alesci, T.~Thiemann, and A.~Zipfel, ``Linking covariant and canonical
  {L}{Q}{G}: new solutions to the euclidean scalar constraint,'' {\em Phys.
  Rev. D 86, 024017}, 2012, \href{http://arxiv.org/abs/1109.1290v1}{{\texttt
  arXiv:1109.1290v1}}.

\bibitem{Thiemann2013}
T.~Thiemann and A.~Zipfel, ``Linking covariant and canonical {L}{Q}{G} {I}{I}:
  Spin foam projector,'' {\em Class.Quant.Grav. 31 125008}, 2014,
  \href{http://arxiv.org/abs/1307.5885v2}{{\texttt arXiv:1307.5885v2}}.

\bibitem{Kirillov}
A.~A. Kirillov, {\em Elements of the Theory of Representations}, vol.~220.
\newblock Springer Science \& Business Media, 2012.

\bibitem{WittenCoadj}
E.~Witten, ``{Coadjoint Orbits of the Virasoro Group},'' {\em Commun. Math.
  Phys.}, vol.~114, p.~1, 1988.

\bibitem{Davids}
S.~Davids, ``{Semiclassical limits of extended Racah coefficients},'' {\em J.
  Math. Phys.}, vol.~41, pp.~924--943, 2000,
  \href{http://arxiv.org/abs/gr-qc/9807061}{{\texttt arXiv:gr-qc/9807061}}.

\bibitem{DavidsDoc}
S.~Davids, ``A state sum model for (2+ 1) lorentzian quantum gravity,'' 2001,
  \href{http://arxiv.org/abs/gr-qc/0110114}{{\texttt arXiv:gr-qc/0110114}}.

\bibitem{Perez2001}
A.~Perez and C.~Rovelli, ``3+1 spinfoam model of quantum gravity with spacelike
  and timelike components,'' {\em Phys.Rev. D}, vol.~64, p.~064002, 2001,
  \href{http://arxiv.org/abs/gr-qc/0011037}{{\texttt arXiv:gr-qc/0011037}}.

\bibitem{Livine2002}
E.~R. Livine, ``Projected spin networks for lorentz connection: Linking spin
  foams and loop gravity,'' {\em Class.Quant.Grav.}, vol.~19, pp.~5525--5542,
  2002, \href{http://arxiv.org/abs/gr-qc/0207084}{{\texttt
  arXiv:gr-qc/0207084}}.

\bibitem{Alexandrov2005}
S.~Alexandrov and Z.~Kadar, ``Timelike surfaces in lorentz covariant loop
  gravity and spin foam models,'' {\em Class.Quant.Grav.}, vol.~22,
  pp.~3491--3510, 2005, \href{http://arxiv.org/abs/gr-qc/0501093}{{\texttt
  arXiv:gr-qc/0501093}}.

\bibitem{Liu2017}
H.~Liu and K.~Noui, ``Gravity as an su(1,1) gauge theory in four dimensions,''
  {\em Class.Quant.Grav. 34 135008}, Feb. 2017,
  \href{http://arxiv.org/abs/1702.06793}{{\texttt arXiv:1702.06793}}.

\bibitem{Conrady2010}
F.~Conrady and J.~Hnybida, ``A spin foam model for general lorentzian
  4-geometries,'' {\em Class.Quant.Grav.}, vol.~27, p.~185011, Feb. 2010,
  \href{http://arxiv.org/abs/1002.1959}{{\texttt arXiv:1002.1959}}.

\bibitem{Frank-thesis}
F.~Hellmann, ``State sums and geometry ({P}h{D} thesis),'' Feb. 2011,
  \href{http://arxiv.org/abs/1102.1688}{{\texttt arXiv:1102.1688}}.

\bibitem{Barrett2003}
J.~W. Barrett and C.~M. Steele, ``Asymptotics of {R}elativistic {S}pin
  {N}etworks,'' {\em Class. Quant. Grav.}, vol.~20, pp.~1341--1362, 2003,
  \href{http://arxiv.org/abs/gr-qc/0209023}{{\texttt arXiv:gr-qc/0209023}}.

\bibitem{Frank2}
J.~W. Barrett, R.~J. Dowdall, W.~J. Fairbairn, F.~Hellmann, and R.~Pereira,
  ``Lorentzian spin foam amplitudes: {G}raphical calculus and asymptotics,''
  {\em Class. Quant. Grav.}, vol.~27, p.~165009, 2010,
  \href{http://arxiv.org/abs/0907.2440}{{\texttt arXiv:0907.2440}}.

\bibitem{FrankEPRL}
J.~W. Barrett, R.~J. Dowdall, W.~J. Fairbairn, H.~Gomes, and F.~Hellmann,
  ``{A}symptotic analysis of the {E}{P}{R}{L} four-simplex amplitude,'' {\em J.
  Math. Phys.}, vol.~50, p.~112504, 2009,
  \href{http://arxiv.org/abs/0902.1170}{{\texttt arXiv:0902.1170}}.

\bibitem{freidel-louapre}
L.~Freidel and D.~Louapre, ``Asymptotics of 6j and 10j symbols,'' {\em Class.
  Quant. Grav.}, vol.~20, p.~1276, 2003,
  \href{http://arxiv.org/abs/hep-th/0209134}{{\texttt arXiv:hep-th/0209134}}.

\bibitem{LivineSpeziale}
E.~R. Livine and S.~Speziale, ``Group integral techniques for the spinfoam
  graviton propagator,'' {\em JHEP}, vol.~0611, p.~092, 2006,
  \href{http://arxiv.org/abs/gr-qc/0608131}{{\texttt arXiv:gr-qc/0608131}}.

\bibitem{4dWilliams}
J.~W. Barrett and R.~M. Williams, ``The {A}symptotics of an amplitude for the
  four simplex,'' {\em Adv. Theor. Math. Phys.}, vol.~3, pp.~209--215, 1999,
  \href{http://arxiv.org/abs/gr-qc/9809032}{{\texttt arXiv:gr-qc/9809032}}.

\bibitem{Livine2007}
E.~R. Livine and S.~Speziale, ``{A New spinfoam vertex for quantum gravity},''
  {\em Phys. Rev.}, vol.~D76, p.~084028, 2007,
  \href{http://arxiv.org/abs/0705.0674}{{\texttt arXiv:0705.0674}}.

\bibitem{Hellmann2013}
F.~Hellmann and W.~Kaminski, ``Holonomy spin foam models: {A}symptotic geometry
  of the partition function,'' {\em JHEP 1310 (2013) 165}, vol.~10, p.~165,
  July 2013, \href{http://arxiv.org/abs/1307.1679}{{\texttt arXiv:1307.1679}}.

\bibitem{Han2013}
M.~Han, ``Covariant {L}oop {Q}uantum {G}ravity, low energy perturbation theory,
  and {E}instein gravity with high curvature {U}{V} corrections,'' {\em Phys.
  Rev. D}, vol.~89, p.~124001, Aug. 2013,
  \href{http://arxiv.org/abs/1308.4063}{{\texttt arXiv:1308.4063}}.

\bibitem{Carlo}
C.~Rovelli and F.~Vidotto, {\em Covariant {L}oop {Q}uantum {G}ravity. An
  elementary introduction to {Q}uantum {G}ravity and {S}pinfoam {T}heory}.
\newblock Cambridge University Press, 2015.

\bibitem{Regge}
T.~Regge, ``General relativity without coordinates,'' {\em Nuovo Cim.},
  vol.~19, p.~558, 1961.

\bibitem{Roberts}
J.~Roberts, ``Classical 6j-symbols and the tetrahedron,'' {\em Geom. Topol.},
  vol.~3, pp.~21--66, 1999,
  \href{http://arxiv.org/abs/math-ph/9812013}{{\texttt arXiv:math-ph/9812013}}.

\bibitem{Suarez}
E.~Suárez-Peiró, ``A {S}chl\"afli differential formula for simplices in
  semi-{R}iemannian hyperquadrics, {G}auss-{B}onnet formulas for simplices in
  the de {S}itter sphere and the dual volume of a hyperbolic simplex,'' {\em
  Pacific Journal of Mathematics}, vol.~194 Issue 1, pp.~229--255, 2000.

\bibitem{Bianchi2010a}
E.~Bianchi, C.~Rovelli, and F.~Vidotto, ``Towards spinfoam cosmology,'' {\em
  Phys.Rev.D}, vol.~82, p.~084035, Mar. 2010,
  \href{http://arxiv.org/abs/1003.3483}{{\texttt arXiv:1003.3483}}.

\bibitem{Rennert2016}
J.~Rennert, ``Timelike twisted geometries,'' {\em Phys. Rev. D}, vol.~95,
  p.~026002, Nov. 2016, \href{http://arxiv.org/abs/1611.00441}{{\texttt
  arXiv:1611.00441}}.

\bibitem{Conrady2010a}
F.~Conrady and J.~Hnybida, ``Unitary irreducible representations of
  {S}{L}(2,{C}) in discrete and continuous {S}{U}(1,1) bases,'' {\em J. Math.
  Phys.}, vol.~52, p.~012501, July 2011,
  \href{http://arxiv.org/abs/1007.0937}{{\texttt arXiv:1007.0937}}.

\bibitem{Hormander}
L.~H\"ormander, {\em The analysis of linear partial differential operators I}.
\newblock Springer-Verlag, 1990.

\bibitem{Frank3}
J.~W. Barrett, R.~J. Dowdall, W.~J. Fairbairn, H.~Gomes, F.~Hellmann, and
  R.~Pereira, ``Asymptotics of 4d spin foam models,'' {\em Gen. Rel. Grav.},
  vol.~43, p.~2421, 2011, \href{http://arxiv.org/abs/1003.1886}{{\texttt
  arXiv:1003.1886}}.

\bibitem{Hellmann-Kaminskishort}
F.~Hellmann and W.~Kami{\'n}ski, ``Geometric asymptotics for spin foam lattice
  gauge gravity on arbitrary triangulations,'' 2012,
  \href{http://arxiv.org/abs/1210.5276}{{\texttt arXiv:1210.5276}}.

\bibitem{Han2011}
M.~Han and M.~Zhang, ``Asymptotics of spinfoam amplitude on simplicial
  manifold: Euclidean theory,'' {\em Class. Quantum Grav.}, vol.~29, p.~165004,
  Sept. 2011, \href{http://arxiv.org/abs/1109.0500}{{\texttt arXiv:1109.0500}}.

\bibitem{Conrady2010b}
F.~Conrady, ``Spin foams with timelike surfaces,'' {\em Class. Quant. Grav.},
  vol.~27, p.~155014, Mar. 2010, \href{http://arxiv.org/abs/1003.5652}{{\texttt
  arXiv:1003.5652}}.

\bibitem{Minkowski}
H.~Minkowski, ``Allgemeine {L}ehrs\"atze \"uber die convexen {P}olyeder,'' {\em
  Nachr. Ges. Wiss. G\"ottingen}, pp.~198--219, 1897.

\bibitem{Minkowski2}
A.~D. Alexandrov, {\em Convex Polyhedra,}.
\newblock Springer Verlag, 2005.

\bibitem{Bianchi2010}
E.~Bianchi, P.~Dona, and S.~Speziale, ``Polyhedra in loop quantum gravity,''
  {\em Phys.Rev.D}, vol.~83, p.~044035, Sept. 2011,
  \href{http://arxiv.org/abs/1009.3402}{{\texttt arXiv:1009.3402}}.

\bibitem{Frank}
J.~W. Barrett, W.~J. Fairbairn, and F.~Hellmann, ``Quantum gravity asymptotics
  from the {SU(2)} 15j symbol,'' {\em Int. J. Mod. Phys. A}, vol.~25, p.~2897,
  2010, \href{http://arxiv.org/abs/0912.4907}{{\texttt arXiv:0912.4907}}.

\bibitem{Ashtekar-Lewandowski}
A.~Ashtekar and J.~Lewandowski, ``Background independent quantum gravity: A
  status report,'' {\em Class.Quant.Grav.}, vol.~21, p.~R53, 2004,
  \href{http://arxiv.org/abs/gr-qc/0404018}{{\texttt arXiv:gr-qc/0404018}}.

\bibitem{thiemann}
T.~Thiemann, {\em Modern Canonical Quantum General Relativity}.
\newblock Cambridge University Press, 2007.

\bibitem{Baez2001}
J.~C. Baez and J.~W. Barrett, ``Integrability for relativistic spin networks,''
  {\em Class.Quant.Grav.}, vol.~18, pp.~4683--4700, 2001,
  \href{http://arxiv.org/abs/gr-qc/0101107}{{\texttt arXiv:gr-qc/0101107}}.

\bibitem{Kaminski2010}
W.~Kaminski, ``All 3-edge-connected relativistic {B}{C} and {E}{P}{R}{L}
  spin-networks are integrable,'' Oct. 2010,
  \href{http://arxiv.org/abs/1010.5384}{{\texttt arXiv:1010.5384}}.

\bibitem{Kaminski2013}
W.~Kaminski and S.~Steinhaus, ``The {B}arrett-{C}rane model: asymptotic measure
  factor,'' {\em Class. Quantum Grav.}, vol.~31, p.~075014, Oct. 2013,
  \href{http://arxiv.org/abs/1310.2957}{{\texttt arXiv:1310.2957}}.

\bibitem{Barrett:1994nn}
J.~W. Barrett, ``First order {Regge} calculus,'' {\em Class. Quant. Grav.},
  vol.~11, p.~2723, 1994, \href{http://arxiv.org/abs/hep-th/9404124}{{\texttt
  arXiv:hep-th/9404124}}.

\bibitem{Barrett1999}
J.~W. Barrett, M.~Rocek, and R.~M. Williams, ``A note on area variables in
  {R}egge calculus,'' {\em Class. Quant. Grav.}, vol.~16, pp.~1373--1376, 1999,
  \href{http://arxiv.org/abs/gr-qc/9710056}{{\texttt arXiv:gr-qc/9710056}}.

\bibitem{Dittrich:2008va}
B.~Dittrich and S.~Speziale, ``Area-angle variables for general relativity,''
  {\em New J. Phys.}, vol.~10, p.~083006, 2008,
  \href{http://arxiv.org/abs/0802.0864}{{\texttt arXiv:0802.0864}}.

\bibitem{Baez}
J.~C. Baez, ``Spin foam models,'' {\em Class. Quant. Grav.}, vol.~15,
  pp.~1827--1858, 1998, \href{http://arxiv.org/abs/gr-qc/9709052}{{\texttt
  arXiv:gr-qc/9709052}}.

\bibitem{spinfoams2}
A.~Perez, ``The {S}pin {F}oam {A}pproach to {Q}uantum {G}ravity,'' {\em Living
  Rev. Relativity}, vol.~16, no.~3, 2013,
  \href{http://arxiv.org/abs/1205.2019}{{\texttt arXiv:1205.2019}}.

\bibitem{Barrett-Crane}
J.~W. Barrett and L.~Crane, ``Relativistic spin networks and quantum gravity,''
  {\em J. Math. Phys.}, vol.~39, pp.~3296--3302, 1998,
  \href{http://arxiv.org/abs/gr-qc/9709028}{{\texttt arXiv:gr-qc/9709028}}.

\bibitem{Bargmann}
V.~Bargmann, ``{I}rreducible {U}nitary {R}epresentations of the {L}orentz
  {G}roup,'' {\em Annals of Mathematics}, vol.~48, pp.~568--640, July 1947.

\bibitem{Perelomov}
A.~Perelomov, {\em Generalized coherent states and their applications}.
\newblock Berlin: Springer, 1986.

\bibitem{gelfand5}
I.~M. Gel'fand and N.~Y. Vilenkin, {\em Generalized function {V}ol. 5
  {I}ntegral geometry and representation theory}.
\newblock Academic Press, 1966.

\bibitem{Dittrich2008}
B.~Dittrich and J.~P. Ryan, ``Phase space descriptions for simplicial 4d
  geometries,'' {\em Class.Quant.Grav.}, vol.~28, p.~065006, July 2011,
  \href{http://arxiv.org/abs/0807.2806}{{\texttt arXiv:0807.2806}}.

\bibitem{Dittrich2010}
B.~Dittrich and J.~P. Ryan, ``Simplicity in simplicial phase space,'' {\em
  Phys.Rev.D}, vol.~82, p.~064026, June 2010,
  \href{http://arxiv.org/abs/1006.4295}{{\texttt arXiv:1006.4295}}.

\bibitem{Berger}
M.~Berger, {\em Geometry I}.
\newblock Berlin: Springer-Verlag, 1987.

\bibitem{plebanski}
J.~F. Pleba{\'n}ski, ``On the separation of einsteinian substructures,'' {\em
  J. Math. Phys.}, vol.~18, p.~2511, 1977.

\bibitem{Urbantke1984}
H.~Urbantke, ``On integrability properties of {S}{U}(2) {Y}ang–{M}ills
  fields. {I}. {I}nfinitesimal part,'' {\em Journal of Mathematical Physics},
  vol.~25, p.~2321, 1984.

\bibitem{Bengtsson1995}
I.~Bengtsson, ``Form geometry and the 't {H}ooft-{P}lebanski action,'' {\em
  Class.Quant.Grav.}, vol.~12, p.~1581, 1995,
  \href{http://arxiv.org/abs/gr-qc/9502010}{{\texttt arXiv:gr-qc/9502010}}.

\bibitem{Reisenberger2016}
M.~P. Reisenberger, ``Classical euclidean general relativity from ``left-handed
  area = right-handed area",'' {\em Class.Quant.Grav. 16 1357}, 1998,
  \href{http://arxiv.org/abs/gr-qc/9804061}{{\texttt arXiv:gr-qc/9804061}}.

\bibitem{BarrettFoxon}
J.~W. Barrett and T.~J. Foxon, ``Semi-classical limits of simplicial quantum
  gravity,'' {\em Class.Quant.Grav.}, vol.~11, pp.~543--556, 1994,
  \href{http://arxiv.org/abs/gr-qc/9310016}{{\texttt arXiv:gr-qc/9310016}}.

\bibitem{Brewin}
L.~Brewin, ``{Fast algorithms for computing defects and their derivatives in
  the Regge calculus},'' {\em Class. Quant. Grav.}, vol.~28, p.~185005, 2011,
  \href{http://arxiv.org/abs/1011.1885}{{\texttt arXiv:1011.1885}}.

\bibitem{Sorkin}
R.~Sorkin, ``The {E}lectromagnetic field on a simplicial net,'' {\em J. Math.
  Phys.}, vol.~16, pp.~2432--2440, 1975.
\newblock [{\em Erratum-ibid.} {\bf 19} (1978) 1800].

\bibitem{Rivin}
I.~Rivin and J.-M. Schlenker, ``The {S}chl\"afli formula in {E}instein
  manifolds with boundary,'' {\em Electron. Res. Announc. Amer. Math. Soc.},
  vol.~5, pp.~18--23, 1999.

\bibitem{Engle2015}
J.~Engle and A.~Zipfel, ``The lorentzian proper vertex amplitude: Classical
  analysis and quantum derivation,'' {\em Phys. Rev. D}, vol.~94, p.~064024,
  Feb. 2015, \href{http://arxiv.org/abs/1502.04640}{{\texttt
  arXiv:1502.04640}}.

\bibitem{Haggard2014}
H.~M. Haggard, M.~Han, W.~Kamiński, and A.~Riello, ``{S}{L}(2,{C})
  {C}hern-{S}imons theory, a non-planar graph operator, and 4d loop quantum
  gravity with a cosmological constant: Semiclassical geometry,'' {\em Nucl.
  Phys. B}, vol.~900, pp.~1--79, Dec. 2014,
  \href{http://arxiv.org/abs/1412.7546}{{\texttt arXiv:1412.7546}}.

\bibitem{Haggard-tetrahedra}
H.~M. Haggard, M.~Han, and A.~Riello, ``Encoding curved tetrahedra in face
  holonomies: Phase space of shapes from group-valued moment maps,'' {\em
  Annales Henri Poincare}, vol.~17, no.~8, pp.~2001--2048, 2016,
  \href{http://arxiv.org/abs/1506.03053}{{\texttt arXiv:1506.03053}}.

\bibitem{Han2010}
M.~Han, ``4-dimensional {S}pin-foam {M}odel with {Q}uantum {L}orentz {G}roup,''
  {\em J. Math. Phys.}, vol.~52, p.~072501, Dec. 2010,
  \href{http://arxiv.org/abs/1012.4216}{{\texttt arXiv:1012.4216}}.

\bibitem{Fairbairn2010}
W.~J. Fairbairn and C.~Meusburger, ``Quantum deformation of two
  four-dimensional spin foam models,'' {\em J. Math. Phys.}, vol.~53,
  p.~022501, Dec. 2010, \href{http://arxiv.org/abs/1012.4784}{{\texttt
  arXiv:1012.4784}}.

\bibitem{Eriksson}
F.~Eriksson, ``The law of sines for tetrahedra and n-simplices,'' {\em
  Geometriae Dedicata}, vol.~7, pp.~71--80, 1978.

\end{thebibliography}
